\documentclass[12pt]{article}
\usepackage{caption}
\usepackage{a4wide}
\usepackage{amsmath,amsfonts,
amssymb, amsthm,dsfont,xcolor, MnSymbol, relsize} 
\usepackage{hyperref}
\usepackage{csquotes}
\usepackage{soul}
\usepackage{cancel}
\usepackage{tikz}
\usetikzlibrary{arrows}
\usepackage{capt-of}

\usepackage[active]{srcltx}
\usepackage{dutchcal}

\usepackage{mathrsfs}
\usepackage{amscd,graphicx,mdframed}

\newtheorem{theorem}{Theorem}[section]
\newtheorem{definition}[theorem]{Definition}
\newtheorem{proposition}[theorem]{Proposition}
\newtheorem{corollary}[theorem]{Corollary}
\newtheorem{lemma}[theorem]{Lemma}

\theoremstyle{definition}
\newtheorem{remark}[theorem]{Remark}

\newtheorem{notation}[theorem]{Notation}

\newtheorem{hypothesis}[theorem]{Hypothesis}

\newcommand{\norm}[1]{\left \Vert #1 \right\Vert_{\mathcal{K}_\B}}

\def\x{\mathbf{x}}
\def\e{{\cal{e}}}

\def\R{\mathbb{R}}

\def\Co{{\mathbb C}} 
\def\H{{\cal H}} 
\def\B{\mathfrak{B}}

\def\cc{\mathcal{c}}
\def\ham{\mathfrak{H}}
\def\K{\mathfrak{K}}

\def\Op{\mathfrak{Op}} 

\def\X{\mathcal X}

\def\L{\mathcal{L}}

\def\E{\mathcal{E}}

\def\Rg{{\rm Range}\,}
\def\s{\mathfrak{s}}
\def\p{\mathfrak{p}}

\def\F0{\mathlarger{\mathlarger{\mathbf{\Lambda}}}}
\def\Fb{\mathlarger{\mathlarger{\mathbf{\Lambda}}}_{\text{\tt bd}}}
\def\Fp{\mathlarger{\mathlarger{\mathbf{\Lambda}}}_{\text{\tt pol}}}

\def\z{\mathfrak{z}}

\def\Rd{\mathbb{R}^d}
\def\Zd{\mathbb{Z}^d}

\def\Z{\mathbb{Z}}
\def\T{\mathbb{T}}
\def\bb1{{\rm{1}\hspace{-3pt}\mathbf{l}}}

\def\Id{\text{I}\hspace*{-2pt}\text{I{\sf d}}}
\def\dist{{\rm dist}}

\def\Int{\mathfrak{I}\mathit{nt}}
\def\supp{\mathop{\rm supp} \nolimits} % Support

\def\repi{\rightY\hspace*{-4pt}\rightarrow}

\def\UL1{\mathcal{U}\hspace*{-3pt}L^1}

\def\bz{\boldsymbol{\z}}
\def\bPhi{\mathbf{\Phi}}

\def\lnu{\mathlarger{\nu}}

\def\MmN{\mathscr{M}_{n_\B}}

\def\zz{\mathcal{z}}

\def\tp{\tilde{\mathfrak{p}}}

\def\tpsi{\widetilde{\psi}}

\def\beq{\begin{equation}}
	\def\eeq{\end{equation}}
	\def\Nb{\mathbb{N}_\bullet}

\numberwithin{equation}{section}

\definecolor{RawSienna}{cmyk}{0,0.72,1,0.45}

\title{A rigorous Peierls-Onsager effective dynamics for semimetals in long-range magnetic fields.}
\author{Horia D. Cornean\footnote{Department of Mathematical Sciences, Aalborg University, Thomas Manns Vej 23, 9220 Aalborg, Denmark; cornean@math.aau.dk}, Bernard Helffer\footnote{Laboratoire de Math{\'e}matiques Jean Leray,  Nantes Universit{\'e}  and CNRS, Nantes, France;
Bernard.Helffer@univ-nantes.fr}, Radu Purice\footnote{\enquote{Simion Stoilow} Institute of Mathematics of the Romanian Academy, P.O. Box 1-764, 014700 Bucharest, Romania; Radu.Purice@imar.ro}}

\begin{document}
	\maketitle
\thispagestyle{empty}

\thispagestyle{empty}
\setcounter{page}{1}

\begin{abstract}
We consider periodic (pseudo)differential {elliptic operators of Schr\"odinger type} perturbed by weak magnetic fields not vanishing at infinity, and extend our previous analysis in \cite{CIP,CHP-2,CHP-4} to the case {of a semimetal having a finite family of Bloch eigenvalues whose range may overlap with the other Bloch bands but remains isolated at each fixed quasi-momentum.} We do not make any assumption of triviality for the associated Bloch bundle. In this setting, we  formulate a general form of the Peierls-Onsager substitution {via strongly localized tight-frames and magnetic matrices. We also} prove the existence of an approximate time evolution for initial states supported inside the range of the isolated Bloch family, with a precise error control.
\end{abstract}

\tableofcontents

\section{Introduction}

The phenomenology of solids in the presence of exterior electro-magnetic fields is an extremely rich field for both theoretical and applied research, and building efficient mathematical models for it is a rather difficult challenge. The basic system for these models is that of a quantum particle in a periodic potential (describing the crystalline structure of the solid) and under the action of some exterior electric and magnetic fields. A very interesting physical idea has been to use Peierls-Onsager effective Hamiltonians \cite{Pe} for modelling the behaviour of such a system in { certain} energy windows {containing the Fermi level}. The mathematics behind these models is rather involved and there is a very large amount of literature devoted to it (to cite just a few papers more closely related to our work \cite{Be1, BC, CIP, dNL, FT, GMSj, HS, HS1, Ne-LMP,Ne-RMP, Pa, PST,Sj}). 
In \cite{CIP} we have shown that the magnetic pseudo-differential calculus (see \cite{MP-1},\cite{IMP-1} - \cite{IMP-3}) offers a natural and efficient mathematical framework for using the Peierls-Onsager procedure and in a series of papers (\cite{CHP-1}-\cite{CHP-4}) we have used this technique for analyzing some spectral properties of two-dimensional systems having some particular Bloch structures. While the detailed spectral analysis developed in \cite{CHP-1}-\cite{CHP-4} was done in  a rather restrictive setting, we become aware of some abstract general form of the Peierls-Onsager effective Hamiltonians that can be rigorously formulated and this will be the subject of this paper.

\subsection{Informal brief overview.}
Let us explain in an informal way the results of  this paper. The important feature shared by quantum periodic Hamiltonians $H$ on a $d$-dimensional space is their Bloch structure. This means that they may be decomposed as a smooth family of self-adjoint operators (which we call fibre operators) acting on square integrable functions on the $d$-dimensional torus ($ L^2(\T^d)$ with $\T^d:=\Rd/\Zd$) and indexed by the dual torus $\T^d_*$ (see the paragraph \ref{SSS-d-torus}). Moreover, each fibre operator on $L^2(\T^d)$ has a compact resolvent and thus an unbounded discrete spectrum $\big\{\lambda_k(\bz^*)\,,\,k\in\mathbb{N}\setminus\{0\},\,\bz^*\in\T^d_*\big\}$. These $\lambda_k$, called the \textit{Bloch levels}, are continuous functions on the dual torus whose graphs might cross among themselves and these crossings introduce singularities that may obscure the role of each individual level in the dynamics generated by the Hamiltonian $H$. 

Nevertheless, in many situations there exists a group of Bloch levels which, although {possibly} crossing among themselves, stay well separated from the remaining infinite family of Bloch levels. We say in this case  that we have an \textit{ isolated Bloch family} and  we can associate with  it a closed invariant subspace. We call it  a \textit{strictly isolated Bloch family} when {its range} is separated from the rest of the spectrum by some spectral gaps  of the full operator $H$ (the so called \textit{gap condition}). 

{While the methods which we will develop in this paper are also capable of dealing with the gapped case, we are mostly interested in the more complicated situation in which the Bloch family stays isolated at each fixed quasi-momentum $\bz^*$, but its range overlaps with the other bands, preventing $H$ from developing a spectral gap, a case which is typical for metals and semimetals, see Figure \ref{picture1}. } 
%An important aspect is that for electrons with energies in the region covered by the values of the Bloch levels in an isolated Bloch family, the dynamics generated by $H$ is equivalent with the dynamics associated with some reduced Hamiltonian acting in the subspace associated with this family and an important observation is that this reduced Hamiltonian is in fact a sum of translations by elements from $\Zd$ of a very well localized operator whose smooth integral kernel  decays rapidly  outside the diagonal.

Physicists have inferred that for such an isolated Bloch family, the influence of a constant magnetic field on their dynamics, in a narrow energy window, could be simply described  by applying the minimal coupling procedure to a corresponding Bloch level. Let us recall that the minimal coupling consists in considering that for a system described by a classical Hamiltonian $h(x,\xi)$ with $x$ the position variable and $\xi$ the momentum variable, the influence of a magnetic field $B=dA$, with $A$ its vector potential, may be described by replacing $\xi$ with $\xi-A(x)$  in the functional expression of the Hamiltonian $h(x,\xi)$ and considering the dynamics generated by $h(x,\xi-A(x))$.
For a \textit{simple isolated Bloch family}, i.e. having only one Bloch level, say $\lambda_{k_0}(\theta)$ with $k_0\in\mathbb{N}\setminus\{0\}$, the above assumption leads   to considering the dynamics associated with the Hamiltonian $\lambda_{k_0}(\xi-A(x))$, and this is the so-called Peierls-Onsager substitution for a non-degenerate, isolated Bloch level.  Transforming this heuristic construction into a rigorous mathematical statement is a rather challenging problem with a long history.  We will briefly comment on it in what follows, in order to better highlight the improvements added by our results.

%\subsection{Brief outline of some of the main existing results.}

The first main breakthrough in formulating a rigorous mathematical statement for the Peierls-Onsager substitution can be  traced  back at least to the papers by  Helffer and  Sj\"{o}strand (\cite{HS}) and  Nenciu (\cite{Ne-LMP}) that have treated the case of one strictly isolated Bloch level with an  associated trivial Bloch bundle, perturbed by a  weak constant external magnetic field. Several natural questions remained: 
\begin{itemize}
	\item How to consider isolated Bloch bands  whose range overlap with others, hence when the gap is only "local" \cite{PST, dNL, FT}; in this case the Hamiltonian decomposition induced by the presence of the isolated gap is not associated to a spectral decomposition and the familiar Riesz projection construction cannot be used.
	\item How to formulate a precise Peierls-Onsager substitution statement for an isolated multi-level Bloch band. Several possibilities may be imagined: to obtain a pure spectral result based on a kind of Grushin problem (see \cite{GMSj,IP}), or to define an effective perturbed band Hamiltonian and even more, to try to replace the case with one level function by a matrix-valued effective Hamiltonian, even in the absence of rapidly decaying composite  Wannier functions. 
    \item Another difficulty associated with the above remark comes from the possible non-triviality of the Bloch bundle associated to the isolated family, i.e. the non-existence of well localized composite Wannier functions.  Conditions for the triviality of the Bloch bundle associated to an isolated Bloch family have been given in \cite{FMP, CM, Pa, Hu}, putting into evidence that non-triviality is related with the presence of periodic magnetic fields. 
	\item We note that a constant magnetic field perturbation has the important feature that the perturbed Hamiltonian commutes with a 'projective representation' of the discrete translations (see \cite{Ne-LMP, Ne-RMP,CHN}). In the presence of strongly localized composite  Wannier functions for the unperturbed system, this projective representation allows one to introduce an algebra of infinite magnetic matrices (as in \cite{Ne-LMP,Ne-RMP}) and even more, to associate a pseudo-differential calculus as in \cite{HS,Sj} (see also \cite{Be1} for a $C^*$-algebraic point of view). When either the magnetic field perturbation is not constant, and/or strongly localized Wannier functions do not exist, these approaches do not work.  
    \item The case in which this magnetic {\it field} perturbation is generated by a smooth, bounded and slowly varying magnetic {\it potential} can be treated with space adiabatic perturbation theory as in \cite{PST, DGR, FT}, up to any order in the perturbation. 
\end{itemize}

In this paper we use  the magnetic pseudo-differential calculus and the theory of strongly localized Parseval frames in order to obtain a general variant of the Peierls-Onsager substitution (see Theorem \ref{T-III} and formula \eqref{F-PO}) giving an answer to all the above mentioned problems. More precisely:
\begin{itemize}
	\item  We allow the range of the isolated Bloch bands to overlap with the others, just like in \cite{PST, FT}. 
	\item No 'slowly varying' hypothesis is imposed to the perturbing magnetic field. Let us mention that under these general circumstances (no slow variation and no spectral gap), the $\epsilon^2$ error in our Theorem \ref{T-III} is a significant improvement of the existing results.
	\item  We do not ask the vector potential of the non-constant part of the magnetic field perturbation to be bounded.
	\item A  quite new aspect of our analysis is the use of  strongly spatially localized Parseval frames in order to treat the non-trivial Bloch bundle case, when no rapidly decaying composite Wannier functions are present. In this sense, our approach is closer in spirit to the original Peierls-Onsager substitution \cite{Pe}, providing a tight-binding Hofstadter-like matrix living in $\ell^2(\mathbb{Z}^d)\otimes\mathbb{C}^N$ for our effective magnetic band Hamiltonian. This result - to the best of our knowledge - has not been previously obtained for Bloch bundles with non-zero Chern numbers, where strongly localized Wannier bases do not exist. Moreover, the tight (Parseval) frame we construct is a result of independent interest.
	\item  Finally, let us  mention that we could have also allowed a background constant magnetic field with a certain rational flux condition, but we decided not to, in order to simplify the already complicated setting and notation. We stress though that we work with a large class of elliptic (pseudo)differential symbols that may have topologically non-trivial projections.
\end{itemize}

\subsection{The basic framework}\label{ss1.3} 
We denote by $\mathbb{N}_\bullet:=\mathbb{N}\setminus\{0\}$ and for $n\in\Nb$ let  $\underline{n}:=\{1,\ldots,n\}$. We use the convention $\R_+:=[0,+\infty)$. {Given a normed linear space $\mathscr{V}$ with the norm $|\cdot|$ we use the notation $<v>:=\sqrt{1+|v|^2}$ for $v\in\mathscr{V}$ and respectively $v(\epsilon)=\mathscr{O}(\epsilon)$ for a map $[0,\epsilon_0]\ni\epsilon\mapsto\,v(\epsilon)\in\mathscr{V}$, for some $\epsilon_0>0$, such that $v(0)=0$ and $\epsilon^{-1}|v(\epsilon)|\leq C<\infty$ for any $\epsilon\in(0,\epsilon_0]$.}
We denote by $\bb1$ the identity operator on any vector space, that we may sometimes indicate by an index; $\Id$ will denote the identity map on any set, that we may also sometimes indicate by an index. For a subset $M$ of a topological space we denote by $\mathring{M}$ its interior subset (i.e. the largest open subset contained in $M$). Given two topological vector spaces $\mathscr{V}_1$ and $\mathscr{V}_2$ we denote by $\mathcal{L}(\mathscr{V}_1;\mathscr{V}_2)$ the linear space of linear continuous maps $\mathscr{V}_1\rightarrow\mathscr{V}_2$ with the topology of uniform convergence on bounded sets. Given any complex Hilbert space $\big(\mathcal{H},(\cdot,\cdot)_\mathcal{H}\big)$, the scalar product will be considered anti-linear in the first factor; we shall use the notations $\mathbb{B}(\mathcal{H})$ for the bounded linear operators on $\mathcal{H}$, respectively $\mathbb{F}(\mathcal{H})$ for the finite-rank operators, $\mathbb{U}(\mathcal{H})$ for the unitary operators, $\mathbb{L}(\mathcal{H})$ for its orthogonal projections and $\mathbb{P}(\mathcal{H})$ for the family of 1-dimensional orthogonal projections. Given two vectors $v$ and $w$ in a Hilbert space $\H$, we denote by $v\bowtie w$ the rank one linear operator 
\beq\label{DF-vw-op}
\H\ni u\mapsto (v\bowtie w)(u):=(w,u)_{\H}\,v\in\H.
\eeq
Given a densely defined closable operator $T$ acting in $\mathcal{H}$ we denote by $\overline{T}$ its closure.

We recall that a Fr\'{e}chet space is a linear space endowed with a topology defined by a countable family of semi-norms $\{\mathfrak{q}_k\}_{k\in\Nb}$, for which it is complete. A bounded subset of a Fr\'{e}chet space $\mathscr{V}$ is a subset $M\subset\mathscr{V}$ such that for any $k\in\Nb$ there exists $C_k(M)>0$ for which:
$	
\mathfrak{q}_k(v)\leq\,C_k(M),\ \forall v\in M.
$

 Let $\X$ be the $d$-dimensional real affine space and suppose fixed an origin $O\in\X$ and a frame $\mathfrak{E}:=\big\{\e_j,\,j\in\underline{d}\big\}$ such that we have an induced identification $\X\cong\R^d$. We denote by $\X^*$ the dual of $\X$ identified with $\Rd$ via the frame $(O,\mathfrak{E})$ and by $<\cdot,\cdot>:\X^*\times\X\rightarrow\R$ the canonical bilinear duality. Let us strengthen here that we shall constantly distinguish between $\X$ and $\X^*$ as two copies of $\R^d$ that play very different roles in our analysis.
We shall also consider the unique real scalar product $\X\times\X\ni(x,y)\mapsto x\cdot y\in\R$ making $\mathfrak{E}$ an orthonormal basis. On the dual $\X^*$ we shall introduce the \textit{dual basis} $\mathfrak{E}_*:=\big\{\e^*_j,\,j\in\underline{d}\big\}$ defined by the relations $<\e^*_j\,,\,\e_k>=2\pi\delta_{jk}$ for any pair $(j,k)\in\underline{d}\times\underline{d}$. We denote by $dx$, resp. $d\xi$ the usual Lebesque measures on $\X$, resp. $\X^*$, associated with the above defined orthogonal basis. We recall the natural action by translations of $\R^d$ on functions or distributions defined on $\X$, denoting by $\tau_x$ the translation with $x\in\X$, resp. by translations on functions or distributions defined on $\X^*$, denoting by $\tau_{\xi}$ for translations with $\xi\in\X^*$. We will denote by $\Xi:=\X\times\X^*$ the phase space and we shall use systematically notations of the form $X:=(x,\xi)\in\Xi$, $Y:=(y,\eta)\in\Xi$, and so on. 
Finally let us consider the discrete, regular lattice $\Gamma:=\underset{1\leq j\leq d}{\bigoplus}\Z\e_j\subset\X$.

We constantly use the H\"{o}rmander multi-index notation $\partial^\alpha_x:=\partial_{x_1}^{\alpha_1}\cdot\ldots\cdot\partial_{x_d}^{\alpha_d}$ and $|\alpha|:=\alpha_1+\ldots+\alpha_d$ for any $\alpha\in\mathbb{N}^d$. For any $j\in\underline{d}$ we denote by $\varepsilon_j\in\mathbb{N}^d$ the multi-index with components $(\varepsilon_j)_k=\delta_{jk}$.

Given a finite dimensional real Euclidean space $\mathcal{V}$ with Euclidean norm denoted by $|\cdot|$ let us consider the function spaces $BC^\infty(\mathcal{V})$ of smooth complex functions on $\mathcal{V}$ bounded together with all their derivatives, $C^\infty_c(\mathcal{V})$ of smooth and compactly supported complex functions on $\mathcal{V}$ and $C^\infty_{\text{\tt pol}}(\mathcal{V})$ of smooth complex functions on $\mathcal{V}$, with polynomial growth at infinity together with all their derivatives.
We define the weight functions (having the properties of a norm but allowed to take also the value $+\infty$) on $C^\infty(\mathcal{V})$:
\beq
\forall(p,n)\in\R_+\times\mathbb{N},\quad\lnu_{p,n}(\phi):=\underset{y\in\mathcal{V}}{\sup}\,<y>^p\,\underset{|a|\leq n}{\max}\,\big|\big(\partial^a\phi\big)(y)\big|,\ \forall\phi\in C^\infty(\mathcal{V}).
\eeq
Then $BC^\infty(\mathcal{V})$ is a Fr\'{e}chet space for the familly of norms $\big\{\lnu_{0,n}\big\}_{n\in\mathbb{N}}$ and we define the space of Schwartz test functions:
\beq
\mathscr{S}(\mathcal{V})\,:=\,\big\{\phi\in BC^\infty(\mathcal{V})\,,\,\lnu_{p,n}(\phi)<\infty,\,\forall(p,n)\in\R_+\times\mathbb{N}\big\}.
\eeq
We denote by $\mathscr{S}^\prime(\mathcal{V})$ the space of tempered distributions, defined as the topological dual of $\mathscr{S}(\mathcal{V})$ on which we can consider either the weak topology with respect to $\mathscr{S}(\mathcal{V})$ or the strong dual topology defined as the topology of uniform convergence on bounded sets of $\mathscr{S}(\mathcal{V})$. We  denote by $\langle\cdot,\cdot\rangle_{\mathcal{V}}:\mathscr{S}^\prime(\mathcal{V})\times\mathscr{S}(\mathcal{V})\rightarrow\mathbb{C}$ the bilinear canonical duality map.
\begin{notation}\label{N-per-distr}
We shall denote by $\mathscr{S}^\prime(\Xi)_{\Gamma}$ the tempered distributions on $\Xi$ that are $\Gamma$-periodic with respect to the variable $x\in\X$.
\end{notation}
For a distribution $\mathfrak{K}\in\mathscr{S}^\prime(\X\times\X)$ we  denote by $\Int_{\X}\,\mathfrak{K}\in\mathcal{L}\big(\mathscr{S}(\X);\mathscr{S}^\prime(\X)\big)$ the linear operator defined as: \beq\label{DF-Int}
\langle\big(\Int_{\X}\,\mathfrak{K}\big)\phi,\psi\rangle_{\X}:=\langle\mathfrak{K},\phi\otimes\psi\rangle_{\X\times\X} \mbox{ for any }(\phi,\psi)\in\mathscr{S}(\X)\times\mathscr{S}(\X)\,.
\eeq 
We  use the notation $\big(\Int_{\X}\,\mathfrak{K}^\dagger\big):=\big(\Int_{\X}\,\mathfrak{K}\big)^*$, so that $\langle\mathfrak{K}^\dagger,\phi\otimes\psi\rangle_{\X\times\X}=\overline{\langle\mathfrak{K},\psi\otimes\phi\rangle}_{\X\times\X}$.

For the $d$-dimensional Fourier transform we use the definitions:
\beq\begin{split}\label{DF-FourierTrsf}
	&\big(\mathcal{F}_{\X}f\big)(\xi):=\int_{\X}\hspace*{-5pt}dx\,e^{-i<\xi,x>}\,f(x)=\int_{\X}\hspace*{-5pt}dx\,\exp\Big(-2\pi i\underset{j\in\underline{d}}{\sum}\xi_j\,x_j\Big)\,f(x),\ \forall f\in L^1(\X)\\
	&\big(\mathcal{F}_{\X^*}g\big)(x):=\int_{\X^*}\hspace*{-7pt}d\xi\,e^{i<\xi,x>}\,g(\xi)=\int_{\X^*}\hspace*{-5pt}d\xi\,\exp\Big(2\pi i\underset{j\in\underline{d}}{\sum}\xi_j\,x_j\Big)\,g(\xi),\ \forall g\in L^1(\X^*).
\end{split}\eeq

\subsection{The magnetic Weyl calculus.}

\paragraph{The magnetic fields.} Let us denote by $\F0^p(\X)$ the real space of smooth $p$-forms on $\X$, by $\Fb^p(\X)$ those having components of class $BC^\infty(\X)$ and by $\Fp^p(\X)$ those with polynomial growth together with all their derivatives; let $d:\F0^p(\X)\rightarrow\F0^{p+1}(\X)$ be the exterior derivation. A smooth magnetic field is an element $B\in\F0^2(\X)$ that is closed, i.e. satisfies $dB=0$. Due to the contractibility of our affine space $\X$ any magnetic field $B\in\Fp^2(\X)$ allows for a vector potential $A\in\Fp^1(\X)$ satisfying the equality $B=dA$ and we shall always use such a choice.
We shall work with fields $B\in\Fb^2(\X)$ verifying the closure condition: $dB=0$, that we call \textit{regular magnetic fields}.

We shall consider a large class of magnetic quantum Hamiltonians defined in the spirit of the minimal coupling physical procedure and having as classical counterpart a H\"{o}rmander type symbol.  We make use of a gauge covariant functional calculus developed in \cite{MP-1,IMP-1,IMP-2,IMP-3} as a twisted version of the Weyl calculus defined by Lars H\"{o}rmander (see \cite{H-3}). Let us strengthen that this is a pseudodifferential version of the integral kernel procedure of gauge  covariant perturbations elaborated by H. Cornean and G. Nenciu \cite{C-99,CN-98,CN-00,Ne-02, C-10} and that for the Laplace differential operator it produces the well known magnetic Laplacian. The basic technical ingredient (see \cite{MP-1}) is the replacement of the usual unitary representation of translations in $L^2(\X)$: 
\beq\nonumber 
\X\ni z\mapsto\,U(z)\in\mathbb{U}\big(L^2(\X)\big),\quad\big(U(z)f\big)(x):=f(x+z),\,\forall z\in\X, 
\eeq 
by the magnetic twisted one:
\beq\label{F-UAx}
\big(U^A(z)f\big)(x):=e^{-i\int_{[x,x+z]}A}\,v\equiv\Lambda^A(x,x+z)f(x+z), \, 
\eeq  
where $
\Lambda^A:\X\times\X\rightarrow\{\z\in\mathbb{C},\,|\z|=1\}$
is a smooth function with derivatives having polynomial growth. Working with this twisted unitary representation is in fact the consequence of replacing the usual  $-i\nabla$ by the magnetic one $-i\nabla-A(x)$ as stipulated by the minimal coupling procedure (see \cite{MP-1}).

We introduce here the magnetic pseudodifferential calculus, leaving a very brief reminder of some of its basic properties for Appendix \ref{A-m-PsiDO}; details may be found in \cite{MP-1} and \cite{IMP-1}-\cite{IMP-3}.

\begin{definition}\label{D-OpA}
	Given $B\in\Fb^2(\X)$ with $dB=0$ and some associated vector potential $A\in\Fp^1(\X)$ such that $dA=B$, we  define the magnetic Weyl quantization map by
	\beq\begin{split}\label{DF-OpA-Phi}
		\Op^A: &\ \mathscr{S}(\Xi)\rightarrow\mathcal{L}\big(\mathscr{S}(\X);\mathscr{S}(\X)\big),\\	
&\big(\Op^A(\Phi)\phi\big)(x)\,=\,\int_{\X^*}d\eta\int_{\X}dy\,e^{-i\int_{[x,y]}A}\,e^{i<\eta,x-y>}\,\Phi\big((x+y)/2,\eta\big)\,\phi(y),\\ 
		&\hspace*{10cm}\forall\phi\in\mathscr{S}(\X),\,\forall x\in\X.
	\end{split}\eeq
\end{definition}
	We notice that taking $A=0$ we recover the usual Weyl quantization. It is shown in \cite{MP-1} that for $A\in\Fp^1(\X)$ the map $\Op^A$ from $\mathscr{S}(\Xi)$ into $\mathcal{L}\big(\mathscr{S}(\X);\mathscr{S}(\X)\big)$ has a canonical extension to an isomorphism $\Op^A$ from $\mathscr{S}^\prime$ onto $\mathcal{L}\big(\mathscr{S}(\X);\mathscr{S}^\prime(\X)\big)$.	

 Moreover, if one considers symbols $\Phi$ that do not depend on the 'configuration variables' $x\in\X$, one may write \eqref{DF-OpA-Phi} using \eqref{F-UAx} as:
 \beq\label{F-OpA-trsl}
\Op^A(\Phi)\,=\,(2\pi)^{d/2}\int_{\X}dz\,\Big(\mathcal{F}_{\X^*}\Phi\Big)(z)\,U^A(z).
 \eeq

\begin{definition}\label{D-A-symbol}
	Given a linear operator $T$ in $\mathcal{L}\big(\mathscr{S}(\X);\mathscr{S}^\prime(\X)\big)$,  the distribution $F\in\mathscr{S}^\prime(\Xi)$ such that $T=\overline{\Op^A(F)}$ is called its $A$-symbol.
\end{definition}
\begin{definition}\label{D-MoyalPr}
We consider the bilinear map $\sharp^B:\mathscr{S}(\Xi)\times\mathscr{S}(\Xi)\rightarrow\mathscr{S}(\Xi)$ is defined by the equality: 
$$\Op^A(\Phi)\, \Op^A(\Psi):=\Op^A(\Phi\sharp^B\Psi)$$ and called the magnetic Moyal product.\\
 We denote by $F^-_B$ the inverse of $F\in\mathscr{S}^\prime(\Xi)$ for the magnetic Moyal product $\sharp^B$.
\end{definition}

\section{The problem}

Our analysis starts from the situation of an electron moving under the influence of a $\Gamma$-periodic electric potential and a superposed magnetic field having a vector potential that is also $\Gamma$-periodic (i.e. satisfying Hypothesis \ref{H-BGamma} below). We consider this system as \textit{unperturbed system} and analyze the influence of the addition of a regular magnetic field (see Subsection \ref{ss1.3} ) having small variations around a non-zero constant value (see Hypothesis \eqref{Hyp-magnField}). We emphasize that no hypothesis of slow variation is made.

\subsection{The unperturbed Hamiltonian}\label{SS-problem}\label{ss2.1}

Due to the fact that the exact form of the Hamiltonian function does not influence our arguments, we shall work with a smooth Hamiltonian function $h:\Xi\rightarrow\R$ that generalizes the usual symbol  $h_0(x,\xi):=\frac 12 |\xi|^2+V(x)$ with $V\in BC^\infty(\X)$, defining the non-relativistic Schr\"{o}dinger Hamiltonian.

\paragraph{The H\"{o}rmander classes of symbols.} Given $p\in\mathbb{R}$ and $\rho\in\{0,1\}$ let us define:
\beq\nonumber \begin{split}
	\qquad \qquad S^p_\rho(\Xi)\,:&=\,\big\{F\in C^\infty_{\text{\tt pol}}(\Xi)\,,\,\nu^{p,\rho}_{n,m}(F)<\infty,\,\forall(n,m)\in\mathbb{N}\times\mathbb{N}\big\},\\
	\text{where:}\qquad\qquad \qquad&\\& \nu^{p,\rho}_{n,m}(F):=\underset{|\alpha|\leq n}{\max}\,\underset{|\beta|\leq m}{\max}\,\underset{(x,\xi)\in\Xi}{\sup}<\xi>^{-p+\rho m}\big|\big(\partial_x^\alpha\partial_\xi^\beta F\big)(x,\xi)\big|\,.\qquad \qquad\qquad
\end{split}\eeq
Given $\Gamma\subset\X$ as above, we shall denote by $S^p_\rho(\Xi)_\Gamma$ the symbols of class $S^p_\rho(\Xi)$ that are periodic with respect to $\Gamma$ in the $x$-variable.
We shall use the following family of $\X$ independent symbols $\mathfrak{m}_s(x,\xi):=<\xi>^s$, with $s\in\R$. 
\begin{definition}
	A symbol $F\in S^p_\rho(\Xi)$ is called elliptic, if there exists positive constants $c$ and $R$ such that, for $(x,\xi) \in \X$ s.t.  $|\xi| \geq R$, 
	\beq\nonumber 
	|F(x,\xi)\big|\,\geq\,c<\xi>^p.
	\eeq
\end{definition}
\begin{hypothesis}\label{H-h1}
 The symbol $h$ is a lower semi-bounded elliptic symbol in $S^p_1(\Xi)_\Gamma$ for some given $p>0$. 
\end{hypothesis}
\noindent We notice that $h_0$ in the beginning of this section is an elliptic symbol in $S^2_1(\Xi)_\Gamma$.\\

We shall fix some magnetic field $B^\circ\in\Fb^2(\X)$ satisfying:
\begin{hypothesis}\label{H-BGamma} $B^\circ\in\Fb^2(\X)$ is closed,  $\Gamma$-periodic and verifies the following {\it zero-flux} property:
	\beq\label{H-Bper}
	\forall(j,k)\in\underline{d}\times\underline{d}:\quad\int\limits_{\text{\rm R}_{jk}}B^\circ=0,\quad\text{where}\quad\text{\rm R}_{jk}:=\big\{s\mathfrak{e}_j+t\mathfrak{e}_k,(s,t)\in[0,1]^2\big\}.
	\eeq
\end{hypothesis}
 In \cite{HH} the authors present an elegant cohomological argument proving that this hypothesis is a necessary and sufficient condition for the existence of a $\Gamma$-periodic vector potential $A^\circ\in\Fb^1(\X)$ such that $B^\circ=dA^\circ$.
 \begin{notation}
 We denote by $\Fp^1(\X)_\Gamma$ the real space of $\Gamma$-periodic functions in $\Fp^1(\X)$.
 \end{notation}
 
\paragraph{The unperturbed Hamiltonian.}
	Given $\Gamma\subset\X$, a symbol $h\in C^\infty(\Xi)$ satisfying Hypothesis~\ref{H-h1} and a magnetic field $B^\circ\in\Fb^2(\X)$ satisfying Hypothesis \ref{H-BGamma} with a choice of $\Gamma$-periodic vector potential $A^\circ\in\Fb^1(\X)$, Proposition 2.4 and Theorem 2.7 in \cite{IMP-3} imply that $\Op^{A^\circ}(h)$ is essentially self-adjoint, commutes with the unitary translations with elements from $\Gamma$ and has a lower semi-bounded self-adjoint closure: 
	\beq\label{DF-h-circ}
	H^\circ:=\overline{\Op^{A^\circ}(h)}:\mathscr{H}^p(\X)\rightarrow L^2(\X)
	\eeq
	where
	$
	\mathscr{H}^p(\X)
	$
	is the Sobolev space of order $p>0$ on $\X$. 

\subsubsection{The $d$-dimensional torus.}\label{SSS-d-torus}

An essential role in dealing with periodic problems is played by the $d$-dimensional torus 
\beq\label{D-torus}
\R^d/\Z^d\cong\mathbb{S}^d:=\big\{\z\in\mathbb{C},\ |\z|=1\big\}^d\,.  
\eeq
When dealing with the identification $\X\cong\Rd$ and $\Gamma\cong\Zd$ we  denote the quotient torus $\X/\Gamma$ by $\T^d\cong\mathbb{S}^d$, while for the dual space $\X^*$ with its dual lattice $\Gamma_*$ we use the notation $\X^*/\Gamma_*=:\T^d_*\cong\mathbb{S}^d$ and denote by $\bz^*$ its points. We  work with a Haar measure that is equal to the restriction of the Lebesgue measure on $\R^{2d}$ with a normalization giving the total measure one  to the manifold. With our definitions we get an  explicit form for the canonical quotient projection:
\beq\label{DF-p}
\p:\X\repi\mathbb{T}^d,\quad\p\big(\underset{1\leq j\leq d}{\sum}x_j\e_j\big)\,=\,\big(e^{2\pi ix_1},\ldots,e^{2\pi ix_d}\big),\,\quad\forall(x_1,\ldots,x_d)\in\X\,,
\eeq
and we choose the following discontinuous bijective section for it:
\beq\nonumber 
\s:\mathbb{T}^d\rightarrow\X,\quad\s(\bz)\equiv\s(\z_1,\ldots\z_d):=(2\pi )^{-1}(\arg\,\z_1,\ldots,\arg\,\z_d),\quad\arg\z\in[-\pi,\pi).
\eeq
We  use these notations for the quotient $\T^d:=\X/\Gamma$ and  the notations $\p_*:\X^*\repi\T^d_*$ and $\s_*:\T^d_*\rightarrow\X^*$ for their analogue in the dual setting $\T^d_*:=\X^*/\Gamma_*$.

Let us fix the following unit cell for the lattice $\Gamma\subset\X$ (evidently not unique choice):
\beq\label{eq:defE}
\mathscr{E}:=\big\{x\in\X\,,\,x_j\in[-1/2,1/2),\,\forall j\in\underline{d}\big\}
\eeq
and the decomposition $\X=\Gamma\times\mathscr{E}$ defined by writing each $t\in\mathbb{R}$ as $t=\lfloor t+1/2\rfloor+\hat{t}\,,\mbox{  with } \hat{t}\in[-1/2,1/2)$, where  the lower integer part function $\R\ni t\mapsto\lfloor t\rfloor\in\Z$ is defined by 
$\lfloor t\rfloor:=\max\big\{k\in\Z,\,k\leq t\big\}$.
For $x\in\Rd$ we  use the notation: 
\beq\label{DF-X-dec-Gamma}
\iota(x):=\big(\lfloor x_1+1/2\rfloor,\ldots,\lfloor x_d+1/2\rfloor\big)\in\Zd,\qquad \hat{x}=\big(\hat{x}_1,\ldots,\hat{x}_d\big):=x-\iota(x)\in\mathcal{E}.
\eeq
We recall the discrete Fourier transform and its inverse:
\beq\label{DF-discrFtrsf}\begin{split}
&\big(\mathcal{F}_\Gamma\vec{c}\big)(\bz^*):=\underset{\gamma\in\Gamma}{\sum}\,e^{-i<\s_*(\bz^*),\gamma>}\,\vec{c}_\gamma,\ \forall\vec{c}\in\ell^1(\Gamma),\\
&\big(\mathcal{F}_{\T^d_*}\varphi\big)_\gamma:=(2\pi)^{-d}\int_{\T^d_*}\,d\bz^*\,e^{i<\s_*(\bz^*),\gamma>}\,\varphi(\bz^*),\ \forall\varphi\in L^1(\T^d_*).
\end{split}\eeq

\subsubsection{The Bloch-Floquet structure of the unperturbed Hamiltonian.} \label{SS-BF-str}

Let us recall the Bloch-Floquet representation:
\beq\label{DF-BF-trsf}
\big(\mathfrak{U}_{F\Gamma}f\big)(\hat{x},\bz^*)\:=\underset{\gamma\in\Gamma}{\sum}e^{-i<\s_*(\bz^*),\gamma>}\,f(\hat{x}+\gamma),\quad\forall f\in\mathscr{S}(\X),\ \forall\hat x \in \mathcal E, { \forall \bz^*\in \T^d_*.}
\eeq
While the Bloch-Floquet theory has been mainly developed in connection with differential operators with periodic coefficients, its extension to periodic pseudo-differential operators has been mentioned in \cite{Ku} 
and presented in some detail in \cite{HM} and \cite{Mo}.

 We now list a number of well-known results concerning the structure of $\Gamma$-periodic operators:
\begin{theorem}[ Bloch-Floquet decomposition]\label{T-FBdec} If $H^{\circ}$ is a lower semi-bounded self-adjoint operator in $L^2(\X)$ with domain $\mathscr{H}^p(\X)$ ($p>0$), which commutes  with all the translations $\big\{U_\gamma,\,\gamma\in\Gamma\big\}$ as in \eqref{DF-h-circ}, then: 
	\begin{enumerate}
		\item  $H^\circ$  is unitarily equivalent with a direct integral (in the sense of \cite{Di}). More precisely, considering the complex linear spaces indexed by $\bz^*\in\T^d_*$:
  \beq\label{F-Fzstar}
  \mathscr{F}_{\bz^*}:=\big\{\hat{F}\in L^2_{\text{\tt loc}}(\X)\,,\,\, \hat{F}(x+\gamma)=e^{i<\s_*(\bz^*),\gamma>}\hat{F}(x)\big\}\,,
  \eeq
  endowed with the scalar product $\|\hat{F}\|_{\bz^*}^2:=\int_{\mathcal{E}}d\hat{x}\,|\hat{F}(\hat{x})|^2$, the linear operator dfined by \eqref{DF-BF-trsf} is a unitary operator $L^2(\X)\rightarrow\int_{\T^d_*}^{\oplus}d\bz^*\,\mathscr{F}_{\bz^*}$ and we have the decomposition:
		\beq\nonumber 
		\mathfrak{U}_{F\Gamma}H^\circ \mathfrak{U}_{F\Gamma}^{-1}=\int^\oplus_{\T^d_*}\hspace*{-5pt}d\bz^*\,\widehat{H^\circ}(\bz^*) 
		\eeq
		where each $\widehat{H^\circ}(\bz^*)$ is a lower semi-bounded self-adjoint operator having as domain $\mathcal{D}(\widehat{H^\circ})=\mathscr{F}_{\bz^*}\cap\mathscr{H}^p_{\text{\tt loc}}(\X)$ which is  denoted by $\mathscr{F}^p_{\bz^*}$.
		\item For any $\zz\in\mathbb{C}\setminus\sigma(H^\circ)$, the map $$\T^d_*\ni\bz^*\mapsto\widehat{R}^\circ_\zz(\bz^*):=\big(\widehat{H}^\circ(\bz^*)-\zz\bb1\big)^{-1}\in\mathbb{B}(\mathscr{F}_{\bz^*})$$
		 is well-defined and smooth.
		\item {For each $\bz^*\in\T^d_*$ and $\zz\in\mathbb{C}\setminus\sigma(H^\circ)$, the resolvent $\widehat{R}^\circ_\zz(\bz^*)$ is a compact operator.}
		%% satisfying: 
		%%$$ \clr 
		%%\Rg\,\widehat{R}^\circ_\zz(\bz^*)\subset\mathscr{F}_{\bz^*}\,\cap\,C^\infty(\X),\qquad\Rg\,\widehat{R}^\circ_\zz(\bz^*)^*\subset\mathscr{F}_{\bz^*}\,\cap\,C^\infty(\X).
		%%$$
		\item  For any $\bz^*\in\T^d_*$, the spectrum of $\widehat{H^\circ}(\bz^*)$ consists of discrete eigenvalues with finite multiplicity  $$\sigma\big(\widehat{H^\circ}(\bz^*)\big)=\big\{\widetilde{\lambda}_k(\bz^*)\big\}_{k\in\mathbb{N}_\bullet}\,,$$ which satisfy  $-\infty<\widetilde{\lambda}_k(\bz^*)<\widetilde{\lambda}_{k+1}(\bz^*)$ for any $k\in\mathbb{N}_\bullet$. 
  
		\item If the geometric multiplicity of $\widetilde{\lambda}_k(\bz^*)$ is $N_k(\bz^*)$, then we can re-index the spectrum by using "simple" eigenvalues in the following way: 
  $$\widetilde{\lambda}_1(\bz^*)=\lambda_1(\bz^*)=\lambda_2(\bz^*)=\dots=\lambda_{N_1} (\bz^*)<\widetilde{\lambda}_2(\bz^*)=\lambda_{N_1+1}(\bz^*)=\dots=\lambda_{N_1+N_2}(\bz^*)<\dots $$
  In this way, for each $k\in\mathbb{N}_\bullet$, the function $\T^d_*\ni\bz^*\mapsto\lambda_k(\bz^*)\in\mathbb{R}$ is continuous. 
	\end{enumerate}
\end{theorem} 
\begin{remark}\label{R-1}
 Let us notice that the restriction to the unit cell $\mathcal{E}$ defines an identification of any space $\mathscr{F}_{\bz^*}$ (as Hilbert space) with $L^2(\mathcal{E})$, this identification {depending on $\bz^*\in\T^d_*$ by the boundary relations imposed by \eqref{F-Fzstar}}; using the differential structure of $\T^d_*$ one can then consider differentiabiliyty properties of functions of the form $\T^d_*\ni\bz^*\mapsto f(\bz^*)\in\mathscr{F}_{\bz^*}$ and $\T^d_*\ni\bz^*\mapsto T(\bz^*)\in\mathbb{B}\big(\mathscr{F}_{\bz^*}\big)$. The distinction between the spaces  $\mathscr{F}_{\bz^*}$ for different $\bz^*\in\T^d_*$ becomes important only when we consider the action of the group $\X\cong\Rd$ by translations.
\end{remark}

\paragraph{The Bloch eigenprojections.}
In order to define the eigenprojections $\widehat{\pi}_k(\bz^*)$ associated with the Bloch eigenvalues $\big\{\lambda_k(\bz^*)\big\}_{k\in\mathbb{N}_\bullet}$, we apply the following procedure. To each $k\in \mathbb{N}_\bullet$ and $\bz^*\in \mathbb T^d_*$, we introduce the minimal labelling of $\lambda_k(\bz^*)$:  
$$\nu (k,\bz^*)=\inf\big\{j\in\mathbb{N}_\bullet\,,\,\lambda_j(\bz^*) =\lambda_k(\bz^*)\big\}\,,$$
and  define the eigenprojections by the following formulas:
\beq\label{DF-pi-k}
\begin{split}
	\widehat{\pi}_k(\bz^*):=&\left\{
	\begin{array}{ll}
		=-\frac{1}{2\pi i}\oint_{\mathscr{C}_k(\bz^*)}d\z\,\big(\widehat{H^\circ}(\bz^*)-\z\bb1\big)^{-1}\in  {\mathbb{L}(\mathscr{F}_{
				\bz^*})} & \mbox{ if } \nu(k,\bz^*)=k\,,\\
		=0 & \mbox{ if }  \nu(k,\bz^*) < k\,,
	\end{array}
	\right.
\end{split}
\eeq 
where $\mathscr{C}_k(\bz^*)\subset\mathbb{C}$ is a circle surrounding $\lambda_k(\bz^*)$ in a anticlockwise direction and no other point from $\sigma\big(\widehat{H^\circ}(\bz^*)\big)$. They define measurable $\mathbb{B}(\mathscr{F}_{\bz^*})$-valued functions on $\mathbb{T}^d_*$. {Some standard arguments using the ellipticity of $h\in S^p_1(\X^*\times\X)$ imply the following statement.
\begin{theorem}\label{T-FBdec-proj}
Under the hypothesis of Theorem \ref{T-FBdec} the range of any finite dimensional orthogonal projection $\widehat{\pi}_k(\bz^*)$ is contained in the space of regular functions $\mathscr{F}_{\bz^*}\bigcap\,BC^\infty(\X)$.
\end{theorem}

\subsubsection{The isolated Bloch family hypothesis.}\label{SS-iso-band-hyp}

	\begin{hypothesis}\label{H-I}
		Given $H^\circ$ as in \eqref{DF-h-circ} with Bloch eigenvalues $\{\lambda_k(\bz^*)\}_{k\in\mathbb{N}_\bullet}$, there exist $k_0\in\mathbb{N}_\bullet$ and $N\in\mathbb{N}$ such that:\begin{enumerate}
			\item $\lambda_{k_0-1}(\bz^*)<\lambda_{k_0}(\bz^*),\quad\lambda_{k_0+N}(\bz^*)<\lambda_{k_0+N+1}(\bz^*),\quad\forall\bz^*\in\mathbb{T}^d_*$, (where by convention $\lambda_0:=-\infty$).
			\item $E_-:=\underset{\bz^*\in\mathbb{T}^d_*}{\sup}\lambda_{k_0-1}(\bz^*)\, <\, E_+:=\underset{\bz^*\in\mathbb{T}^d_*}{\inf}\lambda_{k_0+N+1}(\bz^*)$, $\text{\tt d}_0:=E_+-E_->0\,.$ 
		\end{enumerate}
	\end{hypothesis}
Note that	under this hypothesis,  the interval $
			J_\B:=(E_-,E_+)$
			is not empty.
\begin{center}
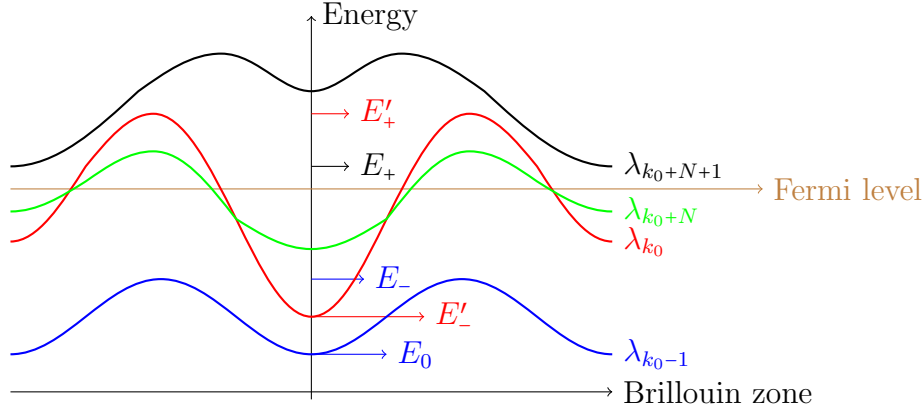

\begin{tikzpicture}
   \draw [->] (-4, 0) -- (4, 0) node[right] {Brillouin zone};
   \draw [->] (0, -0.1) -- (0, 5) node[right] {Energy};
   \draw [->, brown] (-4, 2.7) -- (6, 2.7) node[right] {Fermi level};
   \draw[->, domain= 0: 0.5, black] plot(\x, 3) node[right]{$E_+$};
   \draw[->, domain= 0: 0.5, red] plot(\x, 3.7) node[right]{$E_+'$};
   \draw[->, domain= 0: 1.5, red] plot(\x, 1) node[right]{$E_-'$};
   \draw[->, domain= 0: 0.7, blue] plot(\x, 1.5) node[right]{$E_-$};
      \draw[->, domain= 0: 1, blue] plot(\x, 0.5) node[right]{$E_0$};
   \draw[red, thick] (-4,2) cos (-3, 3) sin (-2.1,3.7) cos (-1,2.3) sin (0, 1)  cos (1,2.3) sin (2.1,3.7) cos (3, 3) sin (4, 2) node[right] {$\lambda_{k_0}$};
   \draw[green, thick] (-4,2.4) cos (-3, 2.8) sin (-2.1,3.2) cos (-1,2.3) sin (0, 1.9)  cos (1,2.3) sin (2.1,3.2) cos (3, 2.8) sin (4, 2.4) node[right] {$\lambda_{k_0+N}$};
   \draw[blue, thick] (-4,0.5) cos (-3, 1) sin (-2,1.5 )  cos (-1, 1) sin (0, 0.5) cos (1, 1) sin (2,1.5) cos (3,1) sin (4, 0.5) node[right] {$\lambda_{k_0-1}$};
   \draw[black, thick] (-4,3) cos (-2.3, 4) sin (-1.2,4.5) cos (-0.5,4.2)  sin (0, 4) cos (0.5, 4.2) sin (1.2 ,4.5) cos (2.3,4) sin (4,3) node[right] {$\lambda_{k_0+N+1}$};
\end{tikzpicture}
\captionof{figure}{\label{picture1} Here $k_0=2$ and $N=1$. The isolated band consists of two crossing eigenvalues in red and green. Formally, the green colour should always be on top of the red colour because $\lambda_{k_0}\leq \lambda_{k_0+N}$, but we will never treat them individually, only as a well-defined isolated family. The energy interval we are interested in is $(E_-,E_+)$ where $E_-$ is the maximum of the blue eigenvalue $\lambda_{k_0-1}$ and $E_+$ is the minimum of the the black one $\lambda_{k_0+N+1}$. The Hamiltonian $H^\circ$ does not have a spectral gap.}
\end{center}

\begin{definition} Any set $\B:=\big\{\lambda_k:\T^d_*\rightarrow\R\,,\,k_0\leq k\leq k_0+N\big\}$ satisfying Hypothesis \ref{H-I} is called an isolated Bloch family of the Hamiltonian $H^\circ$.
\end{definition}

Let us emphasize that we do not assume the existence of any spectral gap below or above the image of the isolated family  $\mathfrak{B}$ in the spectrum of $H^\circ$, i.e. \textbf{we do not assume} that:
\beq\label{H-SpGap}
\underset{\bz^*\in\mathbb{T}^d_*}{\sup}\lambda_{k_0-1}(\bz^*)\,<\,\underset{\bz^*\in\mathbb{T}^d_*}{\inf}\lambda_{k_0}(\bz^*),\quad\text{and}\quad
	\underset{\bz^*\in\mathbb{T}^d_*}{\sup}\lambda_{k_0+N}(\bz^*)\,<\,\underset{\bz^*\in\mathbb{T}^d_*}{\inf}\lambda_{k_0+N+1}(\bz^*)
\eeq
and thus, our results extend also to the case of semimetals.

For some technical reasons that will become clear further, we  choose our energy scale in order to have:
	\begin{hypothesis}\label{H-0} 
	 $$\underset{\bz^*\in\T^d_*}{\inf}\lambda_1(\bz^*)=:E_0>0\,.$$
	\end{hypothesis}
	
Let us put into evidence the following smooth operator-valued $\Gamma_*$-periodic maps associated with  our isolated Bloch family, that play a key role in our arguments:
\beq\begin{split}\label{DF-hat-H-B}
&{ \T_*^d\ni\bz^*}\mapsto\widehat{P}_\B(\bz^*):=\underset{k_0\leq k\leq k_0+N}{\sum}\hat{\pi}_k(\bz^*)\in\mathbb{B}\big(L^2(\mathcal{E})\big),\\
&{ \T_*^d} \ni\bz^*\mapsto\widehat{H}_\B(\bz^*):=\underset{k_0\leq k\leq k_0+N}{\sum}\,\lambda_k(\bz^*)\,\hat{\pi}_k(\bz^*)\in\mathbb{B}\big(L^2(\mathcal{E})\big).
\end{split}\eeq
Hypothesis \ref{H-I} determines the following partition
\beq \label{F-Nb-part}
\Nb=\underline{k_0}\sqcup{\big\{k_0+j,\,0\leq j\leq N\big\}}\sqcup{\big\{k>k_0+N\big\}}
\eeq 
and an induced orthogonal decomposition of the unperturbed Hamiltonian by using \eqref{DF-hat-H-B} and some similar formulas for the other two Hamiltonians (defined in \eqref{eq:3.1}) and \eqref{eq:3.2}):
\beq\label{DF-HB}
\begin{aligned}
H^\circ &=H_0\,\oplus\,H_\B\,\oplus\,H_\infty \mbox{  with } \\
H_\B&=\mathfrak{U}_{F\Gamma}^{-1}\Big(\int^{\oplus}_{\T^d_*}d\bz^*\,\widehat{H}_\B(\bz^*)\Big)\mathfrak{U}_{F\Gamma},\quad {P_\B=\mathfrak{U}_{F\Gamma}^{-1}\Big(\int^{\oplus}_{\T^d_*}d\bz^*\,\widehat{P}_\B(\bz^*)\Big)\mathfrak{U}_{F\Gamma}.}
\end{aligned}
\eeq
In the absence of the gap condition this is not a spectral decomposition.
	
\subsection{The perturbed magnetic Hamiltonian.}\label{SS-pert-m-Ham}
Our main objective in this paper is to analyze  the influence of  a perturbing magnetic field on the dynamics associated with the isolated Bloch family $\B$ of the unperturbed Hamiltonian, as given by  Hypothesis \ref{H-I}. We shall consider that the perturbing magnetic field has a constant part of order $\epsilon$ with possible  fluctuations of strictly less amplitude controlled by the factor $\cc\epsilon$ with $\cc\in[0,1)$ (see \eqref{Hyp-magnField}).

\begin{hypothesis}\label{H-magnField}
	We consider a family of magnetic fields $B^{\epsilon,\cc}\in\Fb^2(\X)$ controlled by two parameters $(\epsilon,c)\in[0,\hat \epsilon_0]\times[0,1]$ for some $\hat \epsilon_0>0$:
	\beq\label{Hyp-magnField}
	B^{\epsilon,\cc}\,:=\,\epsilon B^\bullet\,+\,\cc\epsilon B^\epsilon\,\in\,\Fb^2(\X),
	\eeq 	
where    $B^\bullet$ is a non-zero constant magnetic field and $B^\epsilon $ belongs to a bounded subset ${\cal{S}}_B$ in the Fr\'echet space $\Fb^2(\X)$.

\end{hypothesis}
Under this hypothesis,  we choose some vector  potential $A^\epsilon\in\Fp^1(\X)$ for $B^\epsilon$ and $$A^\bullet_k(x):=(1/2)\underset{1\leq j\leq d}{\sum}B^\bullet_{j,k} x_j$$
so that $B^{\epsilon,\cc}=dA^{\epsilon,\cc}$  for:
 $$A^{\epsilon,\cc}(x):=\epsilon A^\bullet(x)\,+\,\cc\epsilon A^\epsilon(x).$$
Let us emphasize that \textbf{the total magnetic field} of our problem is:
\beq\label{F-A}
B(\epsilon,\cc)\,:=\,B^\circ\,+\,B^{\epsilon,\cc}\,=\,B^\circ\,+\,\epsilon(B^\bullet+\cc\,B^\epsilon)\,,
\eeq
with associate magnetic potential
\beq\label{F-Ab}
A(\epsilon,\cc):= A^\circ+A^{\epsilon,\cc}\,,
\eeq
but we shall frequently use in our arguments the shorthand notation:
\beq\label{N-sh-B}
B\equiv B(\epsilon,\cc),\qquad A\equiv A(\epsilon,\cc).
\eeq
The perturbed magnetic Hamiltonian is then defined by
\beq\label{DF-Heps-c}
H^{\epsilon,\cc}:=\overline{\Op^{A}(h)} \mbox{ on }L^2(\X).
\eeq

 We intend to emphasize the constant part of the perturbing magnetic field and considering that the notations of the form $\Lambda^{A^{\epsilon,c}}$ are rather cumbersome, we shall prefer the following notations:
\beq\label{F-dec-p-m-phase}
\begin{array}{rl}
\Lambda^{\epsilon,\cc}&:=\widetilde{\Lambda}^{\epsilon,\cc}\,\Lambda^{\epsilon},\\\widetilde{\Lambda}^{\epsilon,\cc}(x,y)&:=\exp\Big(-i\cc\epsilon\int_{[x,y]}A^\epsilon\Big),\\ \Lambda^{\epsilon}(x,y)&:=\exp\Big(-i\epsilon\int_{[x,y]}A^\bullet\Big)=\exp\big(-(i\epsilon/2)\underset{1\leq j,k\leq d}{\sum}B^\bullet_{k,j}\,x_k\,y_j\Big).
\end{array}\eeq

 We emphasize that in general the perturbation produced by the magnetic field \eqref{Hyp-magnField} is not relatively bounded with respect to the unperturbed Hamiltonian, and we use the magnetic pseudodifferential calculus to control the perturbation of the $A$-symbols associated with the projections and the Hamiltonians. 
 
\subsection{The main result.}\label{SS-MRes}

Our interest in this paper is to try to understand the behaviour of the restricted dynamics associated with  an isolated Bloch family under the action of a weak {long-range} magnetic field. 

In a first step, {in Section \ref{S-proof-T-I}}, using an argument {based on} the {Feshbach}-Schur {inversion} procedure, we {identify} a quasi-invariant subspace {for $H^{\epsilon,\cc}$} associated with  {the energy window $J_\B$ in Hypothesis \ref{H-I}} and {we construct} an effective Hamiltonian that leaves this subspace invariant. An important fact is that this effective Hamiltonian {approximates} the spectrum of {$H^{\epsilon,\cc}$} and its {unitary time} evolution in the quasi-invariant subspace {up to}  an error of order $\epsilon^2$.
We emphasise {that} these results are proved without the use of Wannier functions, thus without a triviality {condition on} the sub-bundle defined by the isolated Bloch family $\B$. 

In order to obtain a {matricial}  version of these results, even in the absence of exponentially localized Wannier functions, we shall construct strongly localized Parseval tight-frames {instead, consisting of Schwartz functions, see Section \ref{S-Parseval}}. 

In fact, we would like to obtain a more {refined} result, {by highlighting} the effective Hamiltonian associated with  the constant part $\epsilon\,B^\bullet$ of the perturbing magnetic field, in a way extending to our situation the well-known Peierls-Onsager formula obtained by \cite{Ne-RMP,HS} for the case $N=0$ with spectral gap, and in a matrix form by \cite{CIP} for $N\in\mathbb{N}$ with spectral gap and triviality of the Bloch bundle {associated with  the isolated Bloch family}.

Our construction {of the strongly localized Parseval frame} starts from the following fact that we shall discuss in Subsection \ref{SSS-BF-vbdl} {(see Propositions \ref{P-bd-triv} and \ref{P-Pfr-Tstar-sect})}:
\begin{proposition}\label{P-triv-FB}
There exists $n_\B\in\Nb$ (with some precise bounds described in Subsection \ref{SSS-BF-vbdl}) and a family of $n_\B$ smooth functions $\psi_p:\T_*\rightarrow\mathscr{F}$ with {$1\leq p\leq n_\B$}, such that $\psi_p(\bz^*)\in\mathscr{F}_{\bz^*}$ with $\|\psi_p(\bz^*)\|_{\mathscr{F}_{\bz^*}}=1$  and for any $f\in L^2(\T_*;\mathscr{F})$ with $f(\bz^*)\in\mathscr{F}_{\bz^*}$ there exists a \st{unique} {well-defined} $\mathfrak{C}_\B(f)\in L^2(\T_*;\Co^{n_\B})$ such that:
\[
f(\bz^*)=\underset{p\in\underline{n_\B}}{\sum}\mathfrak{C}_\B(f)_p(\bz^*)\,\psi_p(\bz^*),\quad\|f(\bz^*)\|^2=\underset{p\in\underline{n_\B}}{\sum}|\mathfrak{C}_\B(f)_p(\bz^*)|^2.
\]
\end{proposition}
 \begin{notation}\label{N-matrix}~
		\begin{itemize}
			\item  Given $M\in\mathbb{N}_\bullet$, we denote by $\mathscr{M}_M$ the $C^*$-algebra of $M\times M$ complex matrices; 
			\item $\mathscr{M}_\Gamma[\mathfrak{A}]$ is the complex linear space of infinite matrices indexed by $\Gamma\times\Gamma$, having entries in a $C^*$-algebra $\mathfrak{A}$;
			\item $\mathscr{M}^\circ_\Gamma[\mathfrak{A}]$ is the complex subspace in $\mathscr{M}_\Gamma[\mathfrak{A}]$ of matrices having rapid decay outside the diagonal; given any faithful representation $\rho:\mathfrak{A}\rightarrow\mathbb{B}(\mathcal{H})$ we may view $\mathscr{M}^\circ_\Gamma[\mathfrak{A}]$ as a sub-algebra of $\mathbb{B}\big(\ell^2(\Gamma;\mathcal{H})\big)$ endowed with the operator norm that we denote by $\|\cdot\|_{\mathbb{B}(\ell^2(\Gamma;\mathcal{H}))}$;
			\item ${\cal{s}}\big(\Gamma;\mathfrak{A}\big):=\big\{\mathring{V}:\Gamma\rightarrow\,\mathfrak{A},\ \forall n\in\mathbb{N},\ \underset{\gamma\in\Gamma}{\sup}<\gamma>^n\big\|\mathring{V}(\gamma)\big\|_{\mathfrak{A}}\hspace*{-0,2cm}<\infty\ \big\}$.
   \item $\big\{\mathcal{e}_\gamma,\ \gamma\in\Gamma\big\}$ is the canonical orthonormal basis of $\ell^2(\Gamma)$ and for any $M\in\Nb$ let $\big\{\mathcal{e}_p,\ p\in\underline{M}\big\}$ be the canonical orthonormal basis of $\Co^M$.
			\end{itemize}
	\end{notation}

The above facts allow us to associate with  our reduced Hamiltonian $H_\B$ in \eqref{DF-HB} {an equivalent version acting in $\ell^2(\Gamma;\mathscr{M}_{n_\B})$ given by the sequence of matrices }$\mathfrak{m}^\circ_\B\in\mathcal{s}(\Gamma;\mathscr{M}_{n_\B})$ defined in \eqref{DF-m-circ-B} and to construct {an approximate} matricial model for the dynamics of the isolated Bloch family in a {long-range}  magnetic field. 

More precisely, given a a Hamiltonian $H^{\epsilon,\cc}$ as in \eqref{DF-Heps-c} {obeying all the previously formulated assumptions, we will perform the following steps:} %associated with  an unperturbed Hamiltonian $H^\circ$ as in \eqref{DF-h-circ}, satisfying Hypotheses \ref{H-I} and \ref{H-0} and to a regular magnetic field as in \eqref{Hyp-magnField}, we construct the following structures associated with  the isolated Bloch family $\B$ in the given regular field $B^{\epsilon,\cc}$ as in \eqref{Hyp-magnField}:
\begin{itemize}
    \item {Starting from the unperturbed projection $P_\B$ (see \eqref{DF-HB}) we will construct} a family of orthogonal projections $\{P^{\epsilon,\cc}_\B\}_{(\epsilon,\cc)\in[0,\epsilon_0]\times[0,1]}$ for some $\epsilon_0>0$ (Definition \ref{D-m-vers-B}), such that:
    \begin{itemize}
        \item the subspace $P^{\epsilon,\cc}_\B\,L^2(\X)$ is almost  invariant for $H^{\epsilon,\cc}$ modulo an error of order $\epsilon$ (see formula \eqref{F-PH-epsc});
        \item given a compact sub-interval $J\subset\mathring{J}_\B$, we prove that the spectral projection of $H^{\epsilon,\cc}$ associated with  $J$ is contained in $P^{\epsilon,\cc}_\B$ for $\epsilon$ small enough (Proposition \ref{P-est-Jepsilon});
\item $P^{0,0}_\B=P_\B$ and $P^{\epsilon,\cc}_\B=\Op^{A}(p^{\epsilon,\cc}_\B)$ with $p^{\epsilon,\cc}_\B\in S^{-\infty}(\X\times\X^*)$ and $\nu(p^{\epsilon,\cc}_\B-p^{0,0}_\B)\leq C_\nu\epsilon$ for all semi-norms $\nu:S^{-\infty}(\X\times\X^*)\rightarrow\R_+$ with some $C_\nu>0$ independent of $\epsilon\in[0,\epsilon_0]$ (Corollary \ref{C-est-p-symb-pert}).
    \end{itemize}
    \item {We define} the {effective} Hamiltonians $\ham^{\epsilon,\cc}_\B:=P^{\epsilon,\cc}_\B\,H^{\epsilon,\cc}\,P^{\epsilon,\cc}_\B$ (Definition \ref{D-m-vers-B}).
    \item {We construct} a family of sequences $\mathfrak{m}^\epsilon_\B\in\mathcal{s}(\Gamma;\mathscr{M}_{n_\B})$ for $\epsilon\in[0,\epsilon_0]$ {(see Definition \ref{D-m-eps-B})}, {describing the isolated Bloch family dynamics with a constant magnetic field} and such that {(see Proposition \ref{P-4.30})}:
    \begin{align*}
&\underset{\gamma\in\Gamma}{\sup}<\gamma>^n\big\|[\mathfrak{m}^{\epsilon}_\B]_\gamma-[\mathfrak{m}^{\circ}_\B]_\gamma\big\|_{\mathscr{M}_M}\leq C_n\cc\epsilon,\quad 
[\mathfrak{m}^{\circ}_\B]_\gamma\,=\,(2\pi)^{-d}\int_{\T^d_*}d\bz^*\,e^{i<\s_*(\bz^*),\gamma>}\,\widehat{H}_\B(\bz^*).
\end{align*}
    \item {We introduce} a family of isometric linear maps $\mathfrak{C}^{\epsilon,\cc}_\B:P^{\epsilon,\cc}_\B\,L^2(\X)\rightarrow\mathcal{K}_\B:=\ell^2(\Gamma)\otimes\Co^{n_\B}$, with $n_\B\in\Nb$, {already mentioned in} Proposition \ref{P-triv-FB}, {(see \eqref{DF-C-eps-c-B}) that give an equivalent model for the approximate effective dynamics for the isolated Bloch family $\B$ in a magnetic field $B^{\epsilon,\cc}$ satisfying \eqref{Hyp-magnField}.}
\end{itemize}

Given a matrix-valued {rapidly decaying} sequence  $\mathfrak{m}\in\mathcal{s}(\Gamma;\mathscr{M}_{n_\B})$ we {introduce} its discrete Fourier transform:
\[
\hat{\mathfrak{m}}\in\,C^\infty(\T_*;\mathscr{M}_{n_\B}),\quad\hat{\mathfrak{m}}(\bz^*):=\underset{\gamma\in\Gamma}{\sum}\,e^{-i<\s_*(\bz^*),\gamma>}\,\mathfrak{m}_\gamma\in\mathscr{M}_{n_\B},
\]
{and} define a bounded operator associated with  the matrix $\Lambda^{\epsilon,\cc}\mathfrak{m}_{\alpha-\beta}$ given by:
\beq\label{DF-Op-m-eps}
\big(\widetilde{\Op}^{\epsilon,\cc}(\mathfrak{m})\, \vec{c}\big)_\alpha\,:=\,\underset{\beta\in\Gamma}{\sum}\Lambda^{\epsilon,\cc}(\alpha,\beta)\mathfrak{m}_{\alpha-\beta}\, \vec{c}_\beta. 
\eeq
{We} notice that we have the formula:
\beq\label{F-PO}
\widetilde{\Op}^{\epsilon,\cc}(\mathfrak{m})\,=\,(2\pi)^{d/2}\int_{\X}dz\,\Big(\mathcal{F}_{\X^*}\big(\hat{\mathfrak{m}}\circ\p_*\big)\Big)(z)\,U^{A^{\epsilon,\cc}}(z)
\eeq
with $U^{A^{\epsilon,\cc}}(z)$ from \eqref{F-UAx}, {that gives a generalized version of the Peierls-Onsager substitution}.

We are now ready to formulate our main result:
\begin{theorem}\label{T-III}
Let $H^\circ$ be  a Hamiltonian  as in \eqref{DF-h-circ} satisfying Hypotheses \ref{H-I} and \ref{H-0} and  let $H^{\epsilon,\cc}$ be the perturbed Hamiltonian associated with  it as in \eqref{DF-Heps-c} for a magnetic field satisfying \eqref{Hyp-magnField}. {Let $\mathfrak{m}^\epsilon_\B$ be the rapidly decaying sequence from Definition \ref{D-m-eps-B}}. Then there exists $\delta>0$ such that for any $\delta\in[0,\delta_0]$ the interval $J^\delta_\B:=\big(E_-+{2}\delta\,,\,E_+-{2}\delta\big)$ is not void and there exists $\epsilon_0 >0$ and $C>0$ such that, for any pair $(\epsilon,\cc)\in[0,\epsilon_0]\times[0,1]$, the objects introduced above exist and satisfy the following properties for any compact sub-interval $J\subset J^{\delta}_\B$:
\begin{itemize}
\item  We have:
\beq 
\begin{aligned}
&\max\Big\{\underset{\lambda\in\sigma(\widetilde{\Op}^{\epsilon,\cc}(\mathfrak{m}^\epsilon))}{\sup}\dist\Big(\lambda,\sigma(H^{\epsilon,\cc})\bigcap\,J\Big )\,,\,\underset{\lambda\in\sigma(H^{\epsilon,\cc}))}{\sup}\dist\Big(\lambda,\sigma(\widetilde{\Op}^{\epsilon,\cc}(\mathfrak{m}^\epsilon_\B)\big)\bigcap\,J\Big ) \Big\}\\ \label{F-I}
&\qquad \leq\,C\epsilon(\cc+\epsilon)\,.
\end{aligned}
\eeq
\item  For every $v\in E_J( H^{\epsilon,\cc})\,L^2(\X)$  we have:
\beq\label{F-II}
\big\|e^{-itH^{\epsilon,\cc}}v - [\mathfrak{C}^{\epsilon,\cc}_\B]^*e^{-it\widetilde{\Op}^{\epsilon,\cc}(\mathfrak{m}^\epsilon_\B)}\,\mathfrak{C}^{\epsilon,\cc}_\B\,P^{\epsilon,\cc}_\B \, v\big\|_{L^2(\X)}\leq\,C\big[\epsilon\,+\,(1+|t|) ^4\, \epsilon\, (\cc+\epsilon)\big]\|v\|_{L^2(\X)}.
\eeq
\end{itemize}
\end{theorem}

The first point of the theorem follows from the results in Section \ref{S-proof-T-I} (see Theorem 3.1) and Section \ref{S-Parseval}, while the second point of the theorem is proved in Section \ref{S-ev}.

\begin{remark}\label{R-main}
Let us comment on the spectral result contained in the first conclusion of our main Theorem \ref{T-III}. We emphasize that one of the main features of the magnetic pseudo-differential calculus is that the "diverging" part of the perturbation induced by a long-range magnetic field is contained in the complex phase function defining the magnetic quantization (see \eqref{DF-OpA-Phi}, \eqref{DF-Op-m-eps} and \eqref{F-PO}). After factorizing this phase away, the problem becomes more regular and we have a stable and systematic perturbative procedure for the symbols. Thus the variation of the spectrum as a set is mainly induced by the replacement of the quantization $\widetilde{\Op}^{A^\circ}$ by the quantization $\widetilde{\Op}^{\epsilon,\cc}$, which in many cases it may open spectral gaps in regions of absolutely continuous spectrum. The abstract results obtained in \cite{CP-1} imply that the order of magnitude of these spectral gaps can be at most $\sqrt{\epsilon}$, an order which can actually be attained (see \cite{CHP-4}). In some other particular situations, these gaps are much smaller and only of order $\epsilon$ (see \cite{CHP-1, CHP-2}). In this later case, it is essential to construct the "dressed" matrix valued symbol $\mathfrak{m}^\epsilon_\B$ in \eqref{F-I} instead of $\mathfrak{m}^\circ_\B$, by incorporating the constant part of the magnetic field up to an error of order $\epsilon^2$. This argument and the techniques developed in our papers \cite{CHP-1, CHP-2, CHP-4}) provide a procedure for analyzing interesting spectral modifications in narrow energy windows located in the spectrum of the isolated Bloch families as in the situations studied in the previously mentioned papers. %Moreover, the result in Corollary 1.6 in \cite{CP-2} gives a general procedure of comparing the spectra of $\widetilde{\Op}^{\epsilon,\cc}(\mathfrak{m})$ and $\widetilde{\Op}^{\epsilon,0}(\mathfrak{m}^\circ_\B)$ as in Remark 1.7 in \cite{CHP-4}.
\end{remark}

\subsection{Some important particular cases}

Let us comment on the form that  Theorem \ref{T-III} takes if we assume the gap condition \eqref{H-SpGap}, i.e. that the  following inequalities  hold (see Figure \ref{picture2}): 
 $$0<E_0<E_-<E'_-<E'_+<E_+<\infty .$$ 
\begin{center}
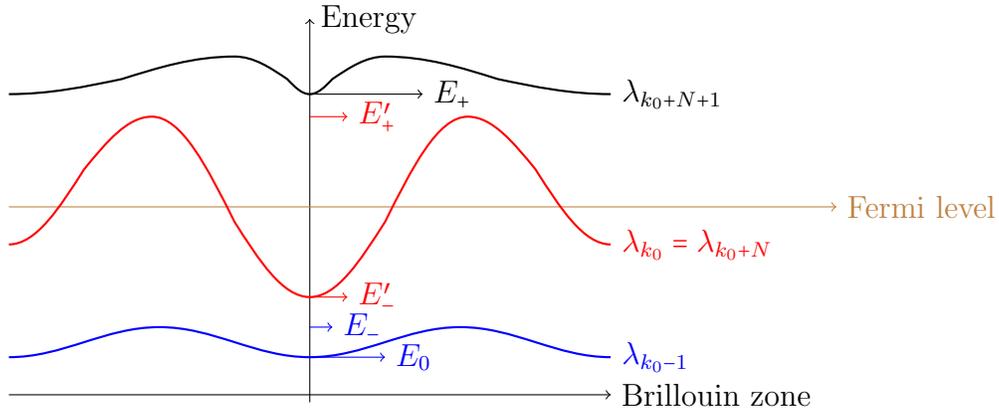

\begin{tikzpicture}
   \draw [->] (-4, 0) -- (4, 0) node[right] {Brillouin zone};
   \draw [->] (0, -0.1) -- (0, 5) node[right] {Energy};
   \draw [->, brown] (-4, 2.5) -- (7, 2.5) node[right] {Fermi level};
   \draw[->, domain= 0: 1.5, black] plot(\x, 4) node[right]{$E_+$};
   \draw[->, domain= 0: 0.5, red] plot(\x, 3.7) node[right]{$E_+'$};
   \draw[->, domain= 0: 0.5, red] plot(\x, 1.3) node[right]{$E_-'$};
   \draw[->, domain= 0: 0.3, blue] plot(\x, 0.9) node[right]{$E_-$};
      \draw [->, domain= 0: 1, blue] plot(\x, 0.5) node[right]{$E_0$};
   \draw[red, thick] (-4,2) cos (-3, 3) sin (-2.1,3.7) cos (-1,2.3) sin (0, 1.3)  cos (1,2.3) sin (2.1,3.7) cos (3, 3) sin (4, 2) node[right] {$\lambda_{k_0}=\lambda_{k_0+N}$};
   \draw[blue, thick] (-4,0.5) cos (-3, 0.7) sin (-2,0.9) cos (-1,0.7) sin (0, 0.5) cos (1, 0.7) sin (2,0.9) cos (3,0.7) sin (4, 0.5) node[right] {$\lambda_{k_0-1}$};
   \draw[black, thick] (-4,4) cos (-2.5, 4.2) sin (-1,4.5) cos (-0.3,4.2)  sin (0, 4) cos (0.3, 4.2) sin (1,4.5) cos (2.5,4.2) sin (4,4) node[right] {$\lambda_{k_0+N+1}$};
\end{tikzpicture}
\captionof{figure}{\label{picture2} Here $k_0=2$ and $N=0$. The isolated band consists of just one eigenvalue and it forms a spectral island for $H^\circ$, separated from the rest of the spectrum by two gaps.}
\end{center}
 
{Then we} can find a circle $\mathring{\mathscr{C}}_\B\subset\mathbb{C}$ containing  $[E'_-,E'_+]$ in its interior domain, leaving $(-\infty,E_0]$ and $[E_+,+\infty)$ in its exterior domain and remaining at a distance $d>0$ from the spectrum $\sigma(H^\circ)$. Then we have the equality:
\beq \nonumber 
H_\B=-(2\pi i)^{-1}\int_{\mathring{\mathscr{C}}_\B}d\zz\,\big(H^\circ-\zz\bb1\big)^{-1}=\Op^{A^\circ}\Big(-(2\pi i)^{-1}\int_{\mathring{\mathscr{C}}_\B}d\zz\,\mathfrak{r}^\circ_\zz\Big)\,.
\eeq
Using the results on spectral continuity obtained  in \cite{AMP,CP-1} as in the arguments leading to \eqref{F-m-sp-est}, we conclude that for any $\delta\in[0,d]$ there exists some $\epsilon(\delta)>0$ such that the conditions verified by $\mathring{\mathscr{C}}_\B$ with respect to $\sigma(H^\circ)$ remain valid for $\sigma(H^{\epsilon,\cc})$ with $(\epsilon,\cc)\in[0,\epsilon(\delta)]\times[0,1]$ and we have the equalities:
\beq   \nonumber \begin{split}
	P^{\epsilon,\cc}_\B&=-(2\pi i)^{-1}\int_{\mathscr{C}_\B}d\zz\big(H_\bot^{\epsilon,\cc}-\zz\bb1\big)^{-1}=-(2\pi i)^{-1}\int_{\mathring{\mathscr{C}}_\B}d\zz\big(H^{\epsilon,\cc}-\zz\bb1\big)^{-1}.
\end{split}\eeq
Moreover we obtain that:
\beq\nonumber 
\ham^{\epsilon,\cc}_\B=H^{\epsilon,\cc}_\B=P^{\epsilon,\cc}_\B\,H^{\epsilon,\cc}\,P^{\epsilon,\cc}_\B.
\eeq

In this case, identifying our perturbing magnetic field $B^{\epsilon,\cc}=\epsilon(B^\bullet+\cc\,B^\epsilon)$ with the field $\epsilon\,B^\epsilon$ used in \cite{CIP}, we can identify our isolated Bloch family operators in magnetic field: $P^{\epsilon,\cc}_\B$ and $\ham^{\epsilon,\cc}_\B$  with the operators $P^{\epsilon}_{I,n}$ (for any $n\in\mathbb{N})$ and $E_I(H^\epsilon)H^\epsilon$ in \cite{CIP} and use Theorem~1.10 in \cite{CIP} in order to obtain an asymptotic expansion  with respect to $\epsilon\in[0,\epsilon_0]$.

{Still working under the global gap condition, if} the isolated Bloch family $\B$ of the  unperturbed Hamiltonian $H^\circ$ admits a system of composite Wannier functions (i.e. if the associated bundle may be trivialized) then Theorem 1.12 {in \cite{CIP}} can also be applied. If the associated bundle of the Bloch isolated family $\B$ is non-trivial, then we may replace the composite Wannier orthonormal basis with the {Parseval} frame  defined in this paper and use our current Theorem \ref{T-III}. 

 Let us now go back to the situation illustrated in Figure \ref{picture1}. We assume that $N=0$ in Hypothesis \ref{H-I} and also assume the triviality of the associated Bloch bundle, i.e. $n_\B=1$. Then we may conclude that there exist $\epsilon_0>0$, {some} $C_n>0$ for all $n\in\mathbb{N}$, {and some} $ C_a>0$ for all $a\in\mathbb{N}^d$, such that   for any pair $(\epsilon,\cc)\in[0,\epsilon_0]\times[0,1]$  there exists a unitary operator $\mathfrak{C}^{\epsilon,\cc}_{k_0}:P^{\epsilon,\cc}_{k_0}L^2(\X)\rightarrow\ell^2(\Gamma)$ such that:
\beq\nonumber 
\underset{(\alpha,\gamma)\in\Gamma\times\Gamma}{\sup}<\gamma>^n \Big |\big[\mathfrak{C}^{\epsilon,\cc}_{k_0}[\ham^{\epsilon,\cc}_{k_0}][\mathfrak{C}^{\epsilon,\cc}_{k_0}]^{-1}\big]_{\alpha,\alpha-\gamma}\hspace*{-8pt}-\Lambda^{A^{\epsilon,\cc}}(\alpha,\alpha-\gamma)\,[\hat{\lambda}^{\epsilon}_{k_0}]_\gamma \Big |\, \leq C_n\cc\epsilon\,,\,\forall n\in \mathbb N
\eeq 
with:
\beq\nonumber 
[\hat{\lambda}^{\epsilon}_{k_0}]_\gamma=(2\pi)^{-d}\int_{\T^d_*}d\bz^*\,e^{i<\s_*(\bz^*)\,,\,\gamma>}\,\lambda^\epsilon_{k_0}(\bz^*)\,,
\eeq
and 
\beq\nonumber 
  \underset{\bz^*\in\T^d_*}{\sup}\big|\big(\partial^a_{\bz*}(\lambda^\epsilon_{k_0}-\lambda_{k_0})(\bz^*)\big)\big|\leq C_a\,\epsilon \,,\, \forall a \in \mathbb N^d\,.
\eeq
Taking also $\cc=0$, i.e. {the} magnetic perturbation to be  constant, we {recover} the results obtained in \cite{Ne-RMP,HS}.

\section{The {effective} magnetic Hamiltonian for the isolated Bloch family}\label{S-proof-T-I}

Our first {step} is to define a subspace of states which {is} quasi-invariant for the perturbed dynamics $H^{\epsilon,\cc}$ (i.e. invariant modulo terms of order $\epsilon$)
and an approximate  dynamics that {lives in} this subspace and approximates the real perturbed dynamics up to errors of order $\epsilon^2$. This provides us with the magnetic version of the isolated Bloch family and allows us to put into evidence the effect of the perturbing magnetic field in the first order in $\epsilon$.
This section is devoted to the proof of the following {result}:
\begin{theorem}\label{T-I} Given a Hamiltonian $H^\circ$ as in \eqref{DF-h-circ} satisfying Hypotheses \ref{H-I} and \ref{H-0}, a magnetic field as in \eqref{Hyp-magnField} and the perturbed Hamiltonian $H^{\epsilon,\cc}$ as in \eqref{DF-Heps-c}, there exists $\delta_0>0$ such that for any $\delta\in(0,\delta_0]$, if we denote by $J^\delta_\B:=\big(E_-+{2}\delta\,,\,E_+-{2}\delta\big)$, there exist constants $\epsilon_0 >0$ and $C>0$ such that, for any pair $(\epsilon,\cc)\in[0,\epsilon_0]\times[0,1]$,
there exists an orthogonal projection $P^{\epsilon,\cc}_\B$ and an {effective} magnetic Hamiltonian $\ham^{\epsilon,\cc}_\B:=P^{\epsilon,\cc}_\B\,H^{\epsilon,\cc}\,P^{\epsilon,\cc}_\B$ commuting with $P^{\epsilon,\cc}_\B$, satisfying:
	\begin{enumerate}
		\item For any $\lambda\in J^\delta_\B$ the operator $(\bb1-P^{\epsilon,\cc}_\B)( H^{\epsilon,\cc}-\lambda\bb1)(\bb1-P^{\epsilon,\cc}_\B)$ is invertible in the subspace $(\bb1-P^{\epsilon,\cc}_\B)L^2(\X)$ and its inverse {denoted by} $[R^\bot_{\epsilon,\cc}(\lambda)]$ is uniformly bounded in $\mathbb{B}\big((\bb1-P^{\epsilon,\cc}_\B)L^2(\X)\big)$ for $\lambda$ in any compact subinterval $J\subset\,J^\delta_\B$. 
	\item For any compact subinterval $J\subset J^{\delta}_\B$  we have the {equality}:
 \[J\cap\sigma(H^{\epsilon,\cc}) =  J\cap\sigma\Big(\big \{\ham^{\epsilon,\cc}_\B\,-\,P^{\epsilon,\cc}_\B H^{\epsilon,\cc}[R^\bot_{\epsilon,\cc}(\lambda)]H^{\epsilon,\cc}P^{\epsilon,\cc}_\B\big \}\big|_{P^{\epsilon,\cc}_{\B}\,L^2(\X)}\Big)\,.\]
 
		%$
		%\lambda\in J\cap\sigma(H^{\epsilon,\cc})\ \Longleftrightarrow\ \lambda\in J\cap\sigma\big([\ham^{\epsilon,\cc}_\B\,-\,P^{\epsilon,\cc}_\B H^{\epsilon,\cc}[R^\bot_{\epsilon,\cc}(\lambda)]H^{\epsilon,\cc}P^{\epsilon,\cc}_\B]\big|_{P^{\epsilon,\cc}_{\B}\,L^2(\X)}\big)\,. $
  \item For any $\lambda\in\,J^\delta_\B\setminus\sigma(H^{\epsilon,\cc})$, considering the orthogonal decomposition $$L^2(\X)=P^{\epsilon,\cc}_{\B}\,L^2(\X)\,\oplus\,(\bb1-P^{\epsilon,\cc}_{\B})L^2(\X)$$ and  denoting by $[R^\sim_{\epsilon,\cc}(\lambda)]$ the inverse {in $P^{\epsilon,\cc}_\B\,L^2(\X)$} of the operator $$\ham^{\epsilon,\cc}_\B\,-\,P^{\epsilon,\cc}_\B H^{\epsilon,\cc}[R^\bot_{\epsilon,\cc}(\lambda)]H^{\epsilon,\cc}P^{\epsilon,\cc}_\B\,-\,\lambda\,P^{\epsilon,\cc}_\B$$ we have the {Feshbach}-Schur block decomposition:
		\begin{align*}
			&\big(H^{\epsilon,\cc}-\lambda\bb1\big)^{-1}\\ 
   &\qquad = 
			\left(\begin{array}{cc}
				[R^\sim_{\epsilon,\cc}(\lambda)] & -[R^\sim_{\epsilon,\cc}(\lambda)]H^{\epsilon,\cc}[R^\bot_{\epsilon,\cc}(\lambda)] \\
				-[R^\bot_{\epsilon,\cc}(\lambda)] H^{\epsilon,\cc}[R^\sim_{\epsilon,\cc}(\lambda)] & [R^\bot_{\epsilon,\cc}(\lambda)]+ [R^\bot_{\epsilon,\cc}H^{\epsilon,\cc}[R^\sim_{\epsilon,\cc}(\lambda)] H^{\epsilon,\cc}[R^\bot_{\epsilon,\cc}(\lambda)]
			\end{array}\right)
		\end{align*}
		and the estimate:
		\beq\label{F-est-FS-epsilon}
		\big\|P^{\epsilon,\cc}_\B H^{\epsilon,\cc}[R^\bot_{\epsilon,\cc}( \lambda)]H^{\epsilon,\cc}P^{\epsilon,\cc}_\B\big\|_{\mathbb{B}(L^2(\X))}\,\leq \,C  \epsilon^2\,\big\|R^\bot_{\epsilon,\cc}(\lambda)\big\|_{\mathbb{B}(L^2(\X))}.
		\eeq
	\end{enumerate}
\end{theorem}

	\begin{remark}\label{R-ext-T-I}
		We can extend the estimate \eqref{F-est-FS-epsilon} and point (3) of the above theorem {to}  any complex number $\lambda$ in $\{\zz\in \mathbb C\,,\, \Re\hspace*{-1pt}{\cal{e}}\zz\in\mathring{J}_\B\}$, with uniform {estimates} with respect to $\Im\hspace*{-1pt}\mathcal{m}\,\lambda$.
\end{remark}

The main {ingredient} in the proof of Theorem \ref{T-I} is the abstract spectral  result presented in Paragraph~\ref{A-SchurC} (that we also {used} in \cite{CHP-2, CHP-3}), together with the magnetic pseudodifferential calculus. We begin by studying closer the decomposition introduced before Equation~\eqref{DF-HB} in order to define (in Definition \ref{D-m-vers-B}) the orthogonal projection $P^{\epsilon,\cc}_\B$ and the {effective} magnetic Hamiltonian $\ham^{\epsilon,\cc}_\B$ appearing in Theorem \ref{T-I} and in the list of results presented in Subsection \ref{SS-MRes}.

\subsection{The decomposition of the unperturbed Hamiltonian}\label{SS-free-dyn}

As explained before {Equation} \eqref{DF-HB},  {from} Theorem \ref{T-FBdec} and  Hypothesis \ref{H-I} we infer that we can write $H^\circ$ as an orthogonal sum of three commuting self-adjoint operators 
$$ H^\circ\,=\,H_0\,\oplus\,H_\B\,\oplus\,H_\infty\,.$$

Starting from the decomposition \eqref{F-Nb-part} and proceeding as in \eqref{DF-hat-H-B}, we can define for any $\bz^*
\in\T^d_*$:
\begin{align}\label{eq:3.1}
&\widehat{H_0}(\bz^*):=\underset{1\leq k\leq k_0-1}{\sum}\lambda_k(\bz^*)\,\hat{\pi}_k( \bz^*)\,\in\mathbb{B}\big(L^2(\mathcal{E})\big),\\
\label{eq:3.2}&\widehat{H_\infty}(\bz^*):=\overline{\underset{k_0+N+1\leq k}{\sum}\lambda_k(\bz^*)\,\hat{\pi}_k( \bz^*)}\,\in  \mathcal L (\mathscr{F}^p_{\bz^*};\mathscr{F}_{\bz^*})\,,
\end{align}
and by an inverse Bloch-Floquet transform we obtain:
\begin{align*}
&H_0:=\mathfrak{U}_{F\Gamma}^{-1}\Big[\int_{\T^d_*}^{\oplus}d \bz^*\,\widehat{H_0}(\bz^*)\Big] \mathfrak{U}_{F\Gamma}\,\in\,\mathbb{B}\big(L^2(\X)\big),\\
&H_\infty:=\mathfrak{U}_{F\Gamma}^{-1}\Big[\int_{\T^d_*}^{\oplus}d \bz^*\,\widehat{H_\infty}(\bz^*)\Big]\mathfrak{U}_{F\Gamma}\in \mathcal L (\mathscr{H}^p(\X); L^2(\X))\,. 
\end{align*}
These formulas together with \eqref{DF-hat-H-B}  provide the stated decomposition of $H^\circ$ in \eqref{DF-HB}. Each of these three commuting self-adjoint operators acts in a closed subspace of $L^2(\X)$ and we have a corresponding decomposition of the identity in $L^2(\X)$ as:
\beq\nonumber 
\bb1_{\mathcal{H}}=P_0\,\oplus\,P_\B\,\oplus\,P_\infty.
\eeq

Using Hypotheses \ref{H-I} and \ref{H-0} we notice that (for $E_0\leq E_-<E_+$ as defined in these Hypotheses):
\[
\sigma(H_0)\subset\{0\}\,\bigcup\, [E_0,E_-],\quad\sigma(H_\B)\subset\{0\}\bigcup\, [E_-^\prime,E_+^\prime],\quad\sigma(H_\infty)\subset\{0\}\bigcup\, [E_+,\infty),
\]
for some $E'_-$ and $E'_+$ such that $0<E_0\leq E'_-<E'_+<\infty$. If we have the strict inequalities $E_-<E_-^\prime$ and $E_+^\prime<E_+$ {as in Figure \ref{picture2}}, then the three orthogonal projections in the above decomposition are in fact spectral projections of $H^\circ$ associated with disjoint components of the spectrum $\sigma(H^\circ)$ so that the "{global} gap condition" is fulfilled and we are in the situation studied in \cite{CIP}. Thus, we shall be interested in the case {illustrated in Figure \ref{picture1}}:
\beq\label{F-sp-points}
0<E_0\leq E_-^\prime\leq E_-<E_+\leq E_+^\prime<\infty,
\eeq
when the three orthogonal projections of the above decomposition of the identity are not necessarily spectral projections of $H^\circ$. {We stress though, that the non-zero spectrum of the fibre of $\H_\B$ is always well isolated from the other Bloch bands which build up $H_0$ and $H_\infty$, for all $\zz^*\in \T^d_*$.}

\subsection{The orthogonal projection associated with  the isolated Bloch family $\mathbf{\B}$.}

\begin{proposition}\label{R-p-symb} There exists $p_\B$ in $S^{-\infty}(\Xi)_\Gamma$ such that $P_\B\equiv\Op^{A^\circ}(p_{\B})$.
\end{proposition}
\begin{proof}
	Let $\sigma_\B(\bz^*):=\underset{k_0\leq k\leq k_0+N}{\bigcup}\lambda_k(\bz^*)$. Due to Hypothesis \ref{H-I}, for any $\bz^*\in\T_*$ we can find a cut-off function $\chi_{\bz^*}\in\,C^\infty_0(\R)$ with support in $\widetilde{I}_{\bz^*}:=(\lambda_{k_0-1}(\bz^*),\lambda_{k_0+N+1}(\bz^*)\ne\emptyset$ and equal to 1 on $J_{\bz^*}:=[\lambda_{k_0}(\bz^*),\lambda_{k_0+N}(\bz^*)]$ so that:
 \[
\widehat{P}_{\B}(\bz^*)\,=\,\chi_{\bz^*}\big(\widehat{H}^\circ(\bz^*)\big).
 \]
 Let us notice that for any $\tilde{\bz}^*\in\T^d_*$, we can find an open neighbourhood $\tilde{O}\subset\T^d_*$ such that we may take $\chi_{\bz^*}$ constant for all $\bz^*\in\tilde{O}$ and using the second statement of Theorem \ref{T-FBdec} we deduce the smoothness of the application:
 \begin{equation}
	\begin{aligned}\label{F-201}
		&\tilde{O}\ni\bz^*\mapsto\widehat{P}_{\B}(\bz^*)\in\mathbb{B}\big(L^2(\mathcal{E})\big).
	\end{aligned}
 \end{equation}
 Due to the compactness of $\T^d_*$ it follows that we can find an open cover $\big\{\mathscr{O}_k,\ k\in\underline{N}\big\}$ for some $N\in\Nb$ and an associated finite partition of unity $\big\{\theta_k\in\,C^\infty_0(\T^d_*),\ k\in\underline{N}\big\}$, as well as a family of functions $\chi_k\in\,C^\infty_0(\R),\ k\in\underline{N}\big\}$ such that $\chi_k=1$ on $J_{\bz^*}$ and $\supp\chi_k\subset\widetilde{I}_{\bz^*}$ for any $\bz^*\in\mathscr{O}_k$. Moreover, having in mind the Helffer-Sj\"{o}strand formula (see \cite{D-95}) we consider the following almost analytic extensions {$\widetilde{\chi}_{k}$ with $k\in\underline{N}$}:
\beq\label{hcd1}
{\zz=x+iy},\quad\widetilde{\chi}_{k}{(\zz,\overline{\zz})}:=\underset{{m}\in\mathbb{N}}{\sum}({\partial^m \chi_k})(x)(iy)^{ m}{({m}!)}^{-1}\rho_{k,m}(y)
\eeq
with $\rho_{k,m}\in\,C^\infty_0(\R)$ with support included in $[-2C_{k,m},2C_{k,m}]$ and $\rho_{k,m}=1$ on $[-C_{k,m},C_{k,m}]$ for $C_{k,m}={\min \Big \{1,\, \Big[\underset{x\in\R}{\sup}\big|\big(\partial^m\chi_k\big)(x)\big|\Big]^{-1}\Big \}}$.
 
{For}  $\zz\in\Co\setminus\sigma(H^\circ)$ {we denote by}:
\beq\label{DF-Rcirc}
R^\circ_{\zz}\,\equiv\,\big(H^\circ-\zz\bb1\big)^{-1}\equiv\,\mathfrak{U}_{F\Gamma}^{-1}\Big(\int^{\oplus}_{\T^d_*}d\bz^*\,\big(\widehat{H}^\circ(\bz^*)-\zz\bb1\big)^{-1}\Big)\mathfrak{U}_{F\Gamma}.
\eeq
We know from \cite{IMP-2} that {the resolvent $R^\circ_{\zz}$} has a symbol $\mathfrak{r}_{\zz}^\circ$ of class $S^{-p}_1(\X^*\times\X)$ and we introduce the notation:
\beq\label{DF-Rzstar}
\widehat{R^\circ_\zz}(\bz^*)\,:=\,\big(\widehat{H}^\circ(\bz^*)-\zz\bb1\big)^{-1}.
\eeq
From Theorem \ref{T-FBdec} we know that the map $\T^d_*\ni\mapsto\widehat{R^\circ_\zz}(\bz^*)\in\mathbb{B}\big(L^2(\mathcal{E})\big)$ is smooth. {Using \eqref{hcd1} we have}
\[
\widetilde{\chi}_{k}\big|_{\R}=\chi_{k},\quad  \underset{y\rightarrow0}{\lim}\,{\Big (|y|^{-n}}\, \partial_{\overline{\zz}}\widetilde{\chi}_{k}{(\zz,\overline{\zz})\Big )}=0,\quad  \forall\,x\in\R,\ \forall n\in\mathbb{N} \,.
\]
With this choice, for $\bz^*\in\mathscr{O}_k$ we have the equality:
\[
\widehat{P}_{\B}(\bz^*)\,=\,\pi^{-1}\iint\big(-\frac i2d\zz\,d\overline{\zz}\big)\,\big(\partial_{\overline{\zz}}\widetilde{\chi}_{k}\big)(\zz,\overline{\zz})\,\widehat{R^\circ_{\zz}}(\bz^*)
\]
 where the function $\zz\mapsto\big(\partial_{\overline{\zz}}\widetilde{\chi}_{k}\big)(\zz,\overline{\zz})\,\widehat{R^\circ_{\zz}}(\bz^*)\in\mathbb{B}(\mathscr{F}_{\bz^*})$ has a continuous extension bounded in norm on the support of $\widetilde{\chi}_{k}$.

 Let us notice that there exist some compact set $K_\B\subset\Co$ such that $\supp\widetilde{\chi}_{k}\subset\,K_\B$ for all $k\in\underline{N}$ and let us consider the inverse Bloch-Floquet transform:
 \beq\label{hdc2}\begin{split}
P_\B:&=\mathfrak{U}_{F\Gamma}^{-1}\Big(\int^{\oplus}_{\T^d_*}d\bz^*\,\widehat{P}_{\B}(\bz^*)\Big)\mathfrak{U}_{F\Gamma}\\
&=\underset{k\in\underline{N}}{\sum}\mathfrak{U}_{F\Gamma}^{-1}\left[\int^{\oplus}_{\T^d_*}d\bz^*\,\theta_k(\bz^*)\iint_{K_\B}\big(-\frac{i}{2\pi}d\zz\,d\overline{\zz}\big)\,\Big(\big(\partial_{\overline{\zz}}\widetilde{\chi}_{k}\big)(\zz,\overline{\zz})\,\widehat{R^\circ_{\zz}}(\bz^*)\Big)\right]\mathfrak{U}_{F\Gamma}
 \end{split}\eeq

In order to show that $P_\B$ is a smoothing operator, we will perform two steps. The first one is to apply  to $P_\B$ the Beals' criterion for magnetic pseudo-differential operators as given in \cite{IMP-2,CHP-3}, and show that $P_\B$ has a symbol of class at least $S_0^0$. The second step is to show by a boot-strap procedure that the symbol is $S^{-\infty}$. 

For the first step let us start by computing the Bloch-Floquet transform of the basic unbounded operators:\\ $\big(Q_j\varphi\big)(x):=x_j\varphi(x)$ and $\big(\Pi^{A^\circ}_j\varphi\big)(x):=-i\big(\partial_j\varphi\big)(x)-A^\circ(x)\varphi(x)$ for $\varphi\in\mathscr{S}(\X)$,\\
which is defined by:
 \begin{align*}
\big(\mathfrak{U}_{F\Gamma}Q_j\varphi\big)_{\bz^*}(x)&=\underset{\gamma\in\Gamma}{\sum}e^{-i<\s_*(\bz^*),\gamma>}\,(x_j+\gamma_j)\varphi(x+\gamma)\\
&=\big(\widehat{Q}_j\mathfrak{U}_{F\Gamma}\varphi\big)_{\bz^*}(x)+\big[i\partial_{\bz^*_j}\big(\mathfrak{U}_{F\Gamma}\varphi\big)\big]_{\bz^*}(x),\\
\big(\mathfrak{U}_{F\Gamma}\Pi^{A^\circ}_j\varphi\big)_{\bz^*}(x)&=\big[\big(-i\partial_{x_j}-A^\circ\big)\big(\mathfrak{U}_{F\Gamma}\varphi\big)\big]_{\bz^*}(x)=\big[\widehat{\Pi}^{A^\circ}_j\big(\mathfrak{U}_{F\Gamma}\varphi\big)\big]_{\bz^*}(x).
 \end{align*}
 where we have also denoted by $\widehat{Q}_j$ and $\widehat{\Pi}^{A^\circ}_j$ for $j\in\underline{d}$ the operator of multiplication with the variable $x_j$ and respectively the covariant derivation with respect to $x\in\X$ on $\mathscr{S}^\prime(\X)$. We notice that for $\varphi\in\mathscr{S}(\X)$:
 \[
\big(\mathfrak{U}_{F\Gamma}\varphi\big)_{\bz^*}(x)=\underset{\gamma\in\Gamma}{\sum}e^{-i<\s_*(\bz^*),\gamma>}\,\varphi(x+\gamma)
 \]
 implies that for each $x\in\X$ the sum over $\gamma\in\Gamma$ is convergent due to the rapid decay of $\varphi\in\mathscr{S}(\X)$ and the map $\T^d_*\ni\bz^*\mapsto\big(\mathfrak{U}_{F\Gamma}\varphi\big)_{\bz^*}\in\mathscr{F}_{\bz^*}$ takes values in $\mathscr{F}_{\bz^*}\bigcap\,BC^\infty(\X)$ and is of class $C^\infty\big(\T^d_*;BC^\infty(\X)\big)$ for the Fr\'{e}chet topology on $BC^\infty(\X)$.

 Now let us compute the commutators (for some pair $(N,M)\in\Nb\times\Nb$:
 \begin{align}\label{F-m-comm-PB}
&[Q_{j_1}\,,\ldots[Q_{j_N}\,,\,[\Pi^{A^\circ}_{k_1}\,,[\Pi^{A^\circ}_{k_M}\,,\ldots\,P_\B]\ldots]\varphi=\\ \nonumber
&\hspace*{12pt}=\underset{k\in\underline{N}}{\sum}\mathfrak{U}_{F\Gamma}^{-1}\int^{\oplus}_{\T^d_*}d\bz^*\,\theta_k(\bz^*)\iint_{K_\B}\big(-\frac{i}{2\pi}d\zz\,d\overline{\zz}\big)\,\big(\partial_{\overline{\zz}}\widetilde{\chi}_{k}\big)(\zz,\overline{\zz})\,\times\\ \nonumber
&\times\,\big[\big(\widehat{Q}_{j_1}+i\partial_{\bz^*_{j_1}}\big)\,,\ldots\big[\big(\widehat{Q}_{j_N}+i\partial_{\bz^*_{j_N}}\big)\,,\big[\widehat{\Pi}^{A^\circ}_{k_1}\,,\ldots\big[\widehat{\Pi}^{A^\circ}_{k_M}\,,\,\widehat{R^\circ_{\zz}}(\bz^*)\big]\ldots\big]\,\big(\mathfrak{U}_{F\Gamma}\varphi\big).
 \end{align}
 Since  our resolvent $R^o_\zz$ in \eqref{DF-Rcirc} has a symbol of class $S_1^{-p}(\X^*\times \X)$, the commutators 
 \[ [Q_{j_1}\,,\ldots[Q_{j_N}\,,\,[\Pi^{A^\circ}_{k_1}\,,[\Pi^{A^\circ}_{k_M}\,,\ldots\,R^\circ_{\zz}]\ldots]
 \]
are also pseudo-differential operators which can be extended to bounded operators on $L^2(\X)$, with norms which are not faster than polynomially increasing in $|y|^{-1}$. These commutators are also $\Gamma$-periodic and we can write:
\begin{align*}
&[Q_{j_1}\,,\ldots[Q_{j_N}\,,\,[\Pi^{A^\circ}_{k_1}\,,[\Pi^{A^\circ}_{k_M}\,,\ldots\,R^\circ_{\zz}]\ldots]=\mathfrak{U}_{F\Gamma}^{-1}\int^{\oplus}_{\T^d_*}d\bz^*\,\times\\
&\times\,\big[\big(\widehat{Q}_{j_1}+i\partial_{\bz^*_{j_1}}\big)\,,\ldots\big[\big(\widehat{Q}_{j_N}+i\partial_{\bz^*_{j_N}}\big)\,,\big[\widehat{\Pi}^{A^\circ}_{k_1}\,,\ldots\big[\widehat{\Pi}^{A^\circ}_{k_M}\,,\,\widehat{R^\circ_{\zz}}(\bz^*)\big]\ldots\big]\,\mathfrak{U}_{F\Gamma},
 \end{align*}
from which we conclude that the operator valued functions 
\[
\T^d_*\ni\bz^*\mapsto\big[\big(\widehat{Q}_{j_1}+i\partial_{\bz^*_{j_1}}\big)\,,\ldots\big[\big(\widehat{Q}_{j_N}+i\partial_{\bz^*_{j_N}}\big)\,,\big[\widehat{\Pi}^{A^\circ}_{k_1}\,,\ldots\big[\widehat{\Pi}^{A^\circ}_{k_M}\,,\,\widehat{R^\circ_{\zz}}(\bz^*)\big]\ldots\big]\in \mathbb{B}\big (L^2(\E)\big )
\]
have norms which can grow at most polynomially in $|y|^{-1}$, uniformly in $\bz^*$.
Using this information in \eqref{F-m-comm-PB} allows us to conclude that $P_\B$ is a magnetic pseudo-differential operator for the vector potential $A^\circ$, with a symbol in the class $S_{0}^0(\X^*\times \X)$. In order to show that the symbol is actually smoothing, let us write 
$$P_\B= (H^\circ+\bb1)^{-n}\, (H^\circ+\bb1)^n P_\B,\quad n\geq 1$$
where $(H^\circ+\bb1)^n P_\B$ can be expressed just like in \eqref{hdc2} but with an extra factor $(\zz+1)^n$ in the integrand. Reasoning like in the case of $P_\B$, the operator $(H^\circ+\bb1)^n P_\B$ has a symbol in $S_{0}^0(\X^*\times \X)$, which shows that the symbol of $P_\B$ belongs to the class $S_0^{-np}$ for all $n\geq 1$, and we are done.

\end{proof}

\begin{remark}\label{R-P-reg-op}
	Similar  arguments  may be applied to $P_0$.
\end{remark}

	 We then conclude that {both}  $H_0=P_0H^\circ P_0$ and $H_\B=P_\B H^\circ P_\B$ have $A^\circ$-symbols $h_0$ and $h_\B$ in $S^{-\infty}(\Xi)_\Gamma$ while $H_\infty=H^\circ-H_{\B}-H_0$ and $H_\bot=H^\circ-H_{\B}$ have $A^\circ$-symbols $h_\infty$ and $h_\bot$ in $S^p_1(\Xi)_\Gamma$. Similarly we deduce that $P_\infty$ has a $A^\circ$-symbol $p_\infty\in S^0_1(\Xi)_\Gamma$.

Due to  Hypothesis \ref{H-0} (that amounts {to} a simple shift of the energy origin) the interval $[0,E_0]\subset\R$ is always non-void and it is a spectral gap for the operator $$H_\bot:= H_0+H_{\infty}\,$$
and  $P_\B$ is simply the orthogonal projection $P_{\ker H_\bot}$ on the kernel of $H_\bot$. Thus we choose {the} circle $\mathscr{C}_\B\subset\mathbb{C}$ {centred at} $0\in\mathbb{C}$ and having a radius $r_0:=E_0/2$ and notice that we have the equalities:
\beq\label{F-R-PB}
P_\B=P_{\ker H_\bot}=- (2\pi i)^{-1}\int_{\mathscr{C}_\B}d\zz\,\big(H_\bot-\zz\bb1\big)^{-1}.
\eeq
Moreover we have the formulas:
$$
p_\B=- (2\pi i)^{-1}\int_{\mathscr{C}_\B}d\zz\,\big(h_\bot-\zz\big)^-_{B^\circ},\quad\,h_\B=p_\B\sharp^{B^\circ}h\sharp^{B^\circ}p_\B.
$$
Since the interval $(E_-,E_+)$ is a spectral gap for $H_\bot$, we consider a contour $\mathscr{C}_0$ which does not contain $0$ in its interior but contains the interval $[E_0,(E_-+E_+)/2]$, and stays at a positive distance from the spectrum of $H_\bot$. Then we have the similar formulas:
\beq\label{DF-p0}
p_0={ -}(2\pi i)^{-1}\int_{\mathscr{C}_0}d\zz\,\big(h_\bot-\zz\big)^-_{B^\circ},\quad\,h_0=p_0\sharp^{B^\circ}h\sharp^{B^\circ}p_0.
\eeq
Finally we have the equalities 
\beq\label{DF-pa} p_\infty=1-p_0-p_\B \mbox{ and  } h_\infty=p_\infty\sharp^{B^\circ}h\sharp^{B^\circ}p_\infty=1-h_0-h_\B\,.
\eeq

\subsection{The perturbed Hamiltonian decomposition induced by $\B$.}
The symbol decomposition $$h=h_0+h_{\B}+h_\infty$$ defined above allows us to consider a similar operator decomposition when the perturbing magnetic field given in \eqref{Hyp-magnField} is added:
\beq\label{F-descOepscch}
\Op^{A}(h)\,=\,\Op^{A}(h_0)+\Op^{A}(h_\B)+\Op^{A}(h_\infty).
\eeq

We denote by  $H^{\epsilon,\cc}_0$,  $H^{\epsilon,\cc}_{\B}$ and $H^{\epsilon,\cc}_\infty$ the self-adjoint operators obtained by closing $\Op^{A}(h_0)$, $\Op^{A}(h_\B)$ and $\Op^{A}(h_\infty)$ respectively.
We  also consider  the lower-semibounded self-adjoint operator $$ H^{\epsilon,\cc}_\bot:=H_0^{\epsilon,\cc}+H^{\epsilon,\cc}_\infty$$ defined as the closure of $\Op^{A}(h_\bot)$ with symbol $h_\bot\in S^p_1(\Xi)_\Gamma$.

An important aspect is that these {perturbed} Hamiltonians no longer commute among themselves, as it was the case in Subsection \ref{SS-free-dyn}. {Another important observation is that their spectra do not change much seen as sets with the Hausdorff distance:} the results in \cite{AMP,CP-1} imply that there exists $\delta_0 >0$ and for any $\delta \in (0,\delta_0]$ there exists $\epsilon_0>0$ such that for any $(\epsilon,c)\in [0,\epsilon_0]\times [0,1]$ we have the inclusions:
\beq\begin{split}\label{F-m-sp-est}
	&\sigma\big(H_0^{\epsilon,\cc}\big)\subset[-\delta\,,\,\delta]\cup[E_0-\delta,E_-+\delta],\quad 
	\sigma\big(H_\B^{\epsilon,\cc}\big)\subset[-\delta\,,\,\delta]\cup[E'_--\delta,E'_++\delta],\\
	&\sigma\big(H^{\epsilon,\cc}_\infty\big)\subset[-\delta\,,\,\delta]\cup[E_+-\delta,+\infty),\\ 
 &\sigma\big(H^{\epsilon,\cc}_\bot\big)\subset[-\delta\,,\,\delta]\cup[E_0-\delta,E_-+\delta]\cup[E_+-\delta,+\infty).
\end{split}\eeq
Actually, much more precise regularity estimates {for the spectral edges seen as functions of $\epsilon$} {were} obtained {in} \cite{CP-2}{, but they are not needed here.}

\subsubsection{The magnetic perturbation of the isolated Bloch family}

We intend to associate with the projection $P_{\B}$ of the isolated Bloch family $\B$ a magnetic version $P^{\epsilon,\cc}_{\B}$ of it when the magnetic field $B^{\epsilon,\cc}$ is added. The first idea would be to consider the operator having integral kernel $\Lambda^{\epsilon,\cc}\cdot\K[P_{\B}]$ with $\K[P_{\B}]$ the integral kernel of $P_\B$ as the general quantization procedure indicates, but we immediately notice that although being still hermitian it is no longer a projector, {only close to one. Our choice is as follows:} 
%as the following computation proves (using also definition \eqref{DF-Omega}):
%\beq\begin{split}\nonumber
	%&\int_{\X}dy\Lambda^{\epsilon}(x,y)\K P_{\B}(x,y)\Lambda^{\epsilon}(y,z)\K P_{\B}(y,z)=\Lambda^{\epsilon}(x,z)\int_{\X}dy\,\Omega^{\epsilon}(x,y,z) \K P_{\B}(x,y)\,\K P_{\B}(y,z)\\
	%&\hspace*{0,5cm}\ne\Lambda^{\epsilon}(x,z)\int_{\X}dy\,\K P_{\B}(x,y)\,\K P_{\B}(y,z)\,=\,\Lambda^{\epsilon}(x,z)\,\K P_{\B}(x,z).
%\end{split}\eeq 
\begin{definition}\label{D-m-vers-B}~
\begin{itemize}
		\item Consider the circle $\mathscr{C}_\B$ introduced in 
  \eqref{F-R-PB} and assume that $\delta$ from \eqref{F-m-sp-est} is less than $E_0/4$. Then the {effective} magnetic projector  of the isolated Bloch family $\B$ is defined to be the spectral projection of $H^{\epsilon,\cc}_\bot$ corresponding to its spectrum located near zero and equals: 
		\beq\begin{split}\label{DF-p-epsilonc-B}
			P^{\epsilon,\cc}_\B\,:=\,- \,(2\pi i)^{-1}\int_{\mathscr{C}_{\B}}d\zz\,\big(H^{\epsilon,\cc}_\bot-\zz\bb1\big)^{-1}.
		\end{split}\eeq
  Note that $P^{\epsilon,\cc}_\B$ is not a spectral projection for $H^{\epsilon,\cc}$.
		\item The {effective} magnetic Hamiltonian of the isolated Bloch family $\B$ is:
		\beq\label{DF-h-epsilonc-B}
\mathfrak{H}^{\epsilon,\cc}_{\mathfrak{B}}\,:=\,P^{\epsilon,\cc}_\B\,H^{\epsilon,\cc}\,P^{\epsilon,\cc}_\B.
		\eeq
	\end{itemize}
\end{definition}
The properties of the magnetic pseudo-differential calculus listed in Appendix~\ref{A-m-PsiDO} imply that both the operators introduced  in Definition \ref{D-m-vers-B} are regularizing operators with $A$-symbol in $S^{-\infty}(\Xi)$.

Let us notice that in general $\mathfrak{H}^{\epsilon,\cc}_{\B}$ is different from  $H^{\epsilon,\cc}_{\B}=\Op^A(h_\B)$ and we have the relation:
\beq\label{dhc2}
	\mathfrak{H}^{\epsilon,\cc}_{\mathfrak{B}}\,=\,P^{\epsilon,\cc}_\B\,H^{\epsilon,\cc}\,P^{\epsilon,\cc}_\B\,=\,P^{\epsilon,\cc}_\B\,H^{\epsilon,\cc}_\B\,P^{\epsilon,\cc}_\B+P^{\epsilon,\cc}_\B\,\Op^A(h_\bot)\,P^{\epsilon,\cc}_\B.
\eeq

%A first step in our approach is to study the relation between the perturbed  Hamiltonian \break  $H^{\epsilon,\cc}=\Op^A(h)$ and the  Hamiltonian $\ham^{\epsilon,\cc}_{\B}$ defined in \eqref{DF-h-epsilonc-B}. 
%In doing that we take advantage of the non-void spectral gap $[E_-,E_+]\subset(0,+\infty)$ that condition (2) in  Hypothesis \ref{H-I} assures to exist for the extra-band Hamiltonian $H^{\epsilon,\cc}_\bot$  and use the reduction procedure based on the Schur complement (that we present in Paragraph \ref{A-SchurC})  for the pair $\big(H^{\epsilon,\cc},P^{\epsilon,\cc}_\B\big)$, taking $\kappa=\epsilon$ and verifying  Hypothesis \ref{H-II} for this pair.

 Consider the circle $\mathscr{C}_0$ used to define $p_0$ in \eqref{DF-p0}. If $\delta$ in \eqref{F-m-sp-est} is small enough, then $\mathscr{C}_0$ encloses the interval $[E_0-\delta,E_-+\delta]$ in its interior domain and remains at a positive distance from the spectrum of $H^{\epsilon,\cc}_\bot$. 
We {define} the following {spectral} projections {of $H^{\epsilon,\cc}_\bot$}:
\beq\begin{split}\label{DF-p-epsilonc-0}
	&P_0^{\epsilon,\cc}\,:={-}\,(2\pi i)^{-1}\int_{\mathscr{C}_0}d\z\,\big(H^{\epsilon,\cc}_\bot-\z\bb1\big)^{-1},\quad P^{\epsilon,\cc}_\infty\,:=\,\bb1\ominus\big(P^{\epsilon,\cc}_0\oplus P^{\epsilon,\cc}_\B\big),\quad P_\bot^{\epsilon,\cc}:=\bb1\ominus\ P^{\epsilon,\cc}_\B.
\end{split}\eeq
From \eqref{DF-p-epsilonc-B} and \eqref{DF-p-epsilonc-0} we have: 
	\beq  \label{R-p-epsilonc} \bb1=P_0^{\epsilon,\cc}\oplus P^{\epsilon,\cc}_\B\oplus P^{\epsilon,\cc}_\infty\,,
	\eeq
	as spectral projections of $H^{\epsilon,\cc}_\bot$, although they are not spectral projections for $H^{\epsilon,\cc}$. Moreover, they do not commute with $H^{\epsilon,\cc}$ either. 
 %In fact we shall obtain further (see Proposition~\ref{P-est-epsilonc}) estimates for their commutators with the three components of $H^{\epsilon,\cc}$ appearing in the decomposition \eqref{F-descOepscch}.

We write the $A$-symbols of the above projection operators as:
\beq\begin{split}\label{F-p-epsilon-c}
	&p_\B^{\epsilon,\cc}:=-(2\pi i)^{-1}\int_{\mathscr{C}_\B}d\zz\,\big(h_\bot-\zz\big)^-_{B},\quad p^{\epsilon,\cc}_0:=-(2\pi i)^{-1}\int_{\mathscr{C}_0}d\zz\,\big(h_\bot-\zz\big)^-_{B} \, ,\\ 
 & p^{\epsilon,\cc}_\infty:=1-p^{\epsilon,\cc}_\B-p^{\epsilon,\cc}_0.
\end{split}\eeq

\subsubsection{Estimating the magnetic perturbation on symbols.}

We recall the definition in \eqref{DF-Rcirc} for the unperturbed resolvent $R^\circ_\zz:=\big(H^\circ-\zz\bb1\big)^{-1}=:\Op^{A^\circ}(\mathfrak{r}^\circ_\zz)$ for $\zz\in\mathbb{C}\setminus\sigma(H^\circ)$ and also define the perturbed resolvent $R^{\epsilon,\cc}_\zz:=\big(H^{\epsilon,\cc}-\zz\bb1\big)^{-1}=:\Op^{A}(\mathfrak{r}^{\epsilon,\cc}_\zz)$ for $\zz\in\mathbb{C}\setminus\sigma(H^{\epsilon,\cc})$.
Due to the statement introducing \eqref{F-simb-rez} we conclude that both $\mathfrak{r}^\circ_\zz$ and $\mathfrak{r}^{\epsilon,\cc}_\zz$ belong to  $S^{-p}_1(\Xi)$ with $\mathfrak{r}^\circ_\zz$  being $\Gamma$-periodic.

Let us formulate {a technical result which - in our current setting -} replaces Proposition~3.5 in \cite{CIP} and which can be proved with similar arguments.

\begin{proposition}\label{P-replP3_5}
Given $(m,s)\in\R\times\R$, there exists $\epsilon_0>0$ such that for any $(\epsilon,c) \in[0,\epsilon_0]\times [0,1]$ there exists a  continuous bilinear map  
$$ S^m_1(\Xi)_\Gamma\times S^s_1(\Xi)_\Gamma \ni (F,G) \mapsto \mathcal{r}_{\epsilon,\cc}\big(F,G)\in S^{m+s-2}_1(\Xi)$$ 
such that (with $B\equiv B(\epsilon,\cc)$ and $B^\circ$ as in \eqref{F-A}, \eqref{Hyp-magnField} and \eqref{H-Bper}), 
\beq\nonumber 
F\sharp^BG\,=\,F\sharp^{B^\circ}G\,+\,\epsilon\,\mathcal{r}_{\epsilon,\cc}\big(F,G)
\eeq
{Moreover, the bilinear map $\mathcal{r}_{\epsilon,c}$ is uniformly bounded with respect to $(\epsilon,c) \in[0,\epsilon_0]\times [0,1]$.}
\end{proposition}

We shall also need to work with the resolvents:
\begin{itemize}
	\item $R^\circ_{\bot,\zz}:=\big(H_\bot-\zz\bb1\big)^{-1}=:\Op^{A^\circ}(\mathfrak{r}^\circ_{\bot,\zz})$ for $\zz\in\mathbb{C}\setminus\sigma(H_\bot)$,
	\item $R^{\epsilon,\cc}_{\bot,\zz}:=\big(H^{\epsilon,\cc}_\bot-\zz\bb1\big)^{-1}=:\Op^{A}(\mathfrak{r}^{\epsilon,\cc}_{\bot,\zz})$ for $\zz\in\mathbb{C}\setminus\sigma(H^{\epsilon,\cc}_\bot)$.
\end{itemize}
\begin{corollary}\label{P-1_8-CIP}~
 For $K\subset\big[\mathbb{C}\setminus\sigma(H_\bot)\big]$, there exists $\epsilon_0>0$ such that:
	\begin{enumerate}
		\item  $K\subset\big[\mathbb{C}\setminus\sigma(H^{\epsilon,\cc}_\bot)\big]$ for any $(\epsilon,\cc)\in[0,\epsilon_0]\times[0,1]$.
		\item The map $\mathbb{C}\setminus\sigma(H^{\epsilon,\cc}_\bot)\ni\zz\mapsto\mathfrak{r}^{\epsilon,\cc}_{\bot,\zz}\in S^{-p}_1(\Xi)$ is continuous for the Fr\'{e}chet topology on $S^{-p}_1(\Xi)$, uniformly in $(\epsilon,\cc)\in[0,\epsilon_0]\times[0,1]$.
		\item We have
		\beq \nonumber 
\mathfrak{r}^{\epsilon,\cc}_{\bot,\zz}\,=\,(1+\epsilon\,\mathcal{r}_{\epsilon\cc}(h_\bot,\mathfrak{r}^{\circ}_{\bot,\zz})\big)^-_B\,\overset{B}{\sharp}\mathfrak{r}^{\circ}_{\bot,\zz}.
		\eeq
	\end{enumerate}
\end{corollary}
\begin{proof}
	We shall  only present the small changes  to be done in the proof of Proposition 1.8 in \cite{CIP}. In fact, the main point is to replace Equation (3.22) in \cite{CIP} by the  equality:
	\beq\nonumber 
	1=\big(h_\bot-\zz\big)\sharp^{B^\circ}\mathfrak{r}^{\circ}_{\bot,\zz}=\big(h_\bot-\zz\big)\sharp^{B}\mathfrak{r}^{\circ}_{\bot,\zz}\,-\,\epsilon\,\mathcal{r}_{\epsilon,\cc}\big(h_\bot,\mathfrak{r}^{\circ}_{\bot,\zz}\big)
	\eeq 
	with $\mathcal{r}_{\epsilon\cc}\big(h_\bot,\mathfrak{r}^{\circ}_{\bot,\zz}\big)$ in $ S^{-2}_1(\Xi)$ uniformly for $(\epsilon,\cc)\in[0,\epsilon_0]\times[0,1]$. Once we made these replacements in \cite{CIP} the arguments in the proof of Proposition 1.8 therein remain  true.
\end{proof}
Using the notations \eqref{DF-p0} and \eqref{DF-pa} one obtains:
\begin{corollary}\label{C-est-p-symb-pert}
	There exists $\epsilon_0>0$ such that for any index  $a\in\{0,\B,\infty,\bot\}$ the differences $p_a^{\epsilon,\cc}-p_a$ belong to $S^{-\infty}(\Xi)$ and given any seminorm $\lnu:S^{-\infty}(\Xi)\rightarrow\R_+$ defining its Fr\'{e}chet topology, there exists $C_\nu>0$ such that  for any $\epsilon,\cc$ in $[0,\epsilon_0]\times[0,1] $, $$\lnu\big(p_a^{\epsilon,\cc}-p_a\big)\leq\,C_\nu\epsilon \,.
	$$
\end{corollary}

This corollary together with Theorem 2.6 in \cite{IMP-3} (see also Subsection \ref{SS-problem}) implies the following result.
\begin{corollary}\label{C-est-p-pert}
	There exist $\epsilon_0>0$ and $C>0$ such that for any index  $a\in\{0,\B,\infty,\bot\}$ 
	\beq\nonumber 
	\big\|P_a^{\epsilon,\cc}\,-\,\Op^{A(\epsilon,c)}(p_a)\big\|_{\mathbb{B}(L^2(\X))}\,\leq\,C \,\epsilon\,,\qquad\forall(\epsilon,\cc)\in[0,\epsilon_0]\times[0,1]\,.
	\eeq
\end{corollary}

\begin{proposition}\label{P-est-epsilonc}
	There exist $\epsilon_0>0$ and $C>0$ such that  for any $(\epsilon,\cc)\in[0,\epsilon_0]\times[0,1]$:
	\begin{itemize}
		\item $H_\B^{\epsilon,\cc}P_0^{\epsilon,\cc}\in\mathbb{B}\big(L^2(\X)\big),\ H^{\epsilon,\cc}_0P_\B^{\epsilon,\cc}\in\mathbb{B}\big(L^2(\X)\big),\ H^{\epsilon,\cc}_\infty P_0^{\epsilon,\cc}\in\mathbb{B}\big(L^2(\X)\big)$,\\ 
		$H^{\epsilon,\cc}_0P_\infty^{\epsilon,\cc}\in\mathbb{B}\big(L^2(\X)\big),\ H_\infty^{\epsilon,\cc}P_\B^{\epsilon,\cc}\in\mathbb{B}\big(L^2(\X)\big),\ H_\B^{\epsilon,\cc}P_\infty^{\epsilon,\cc}\in\mathbb{B}\big(L^2(\X)\big)$,
		\item  $\big\|H_\B^{\epsilon,\cc}P_0^{\epsilon,\cc}\big\|_{\mathbb{B}(L^2(\X))}\leq\,C\,\epsilon\,$\,;\quad$\ \ \big\|H^{\epsilon,\cc}_0P_\B^{\epsilon,\cc}\big\|_{\mathbb{B}(L^2(\X))}\leq\,C\,\epsilon$\,;\quad$\big\|H^{\epsilon,\cc}_\infty P_0^{\epsilon,\cc}\big\|_{\mathbb{B}(L^2(\X))}\leq\,C\,\epsilon$\,;\\
		$\big\|H^{\epsilon,\cc}_0P_\infty^{\epsilon,\cc}\big\|_{\mathbb{B}(L^2(\X))}\leq\,C\,\epsilon$\, ;\quad$\big\|H_\infty^{\epsilon,\cc}P_\B^{\epsilon,\cc}\big\|_{\mathbb{B}(L^2(\X))}\leq\,C\,\epsilon$\, ;\quad$\ \big\|H_\B^{\epsilon,\cc}P_\infty^{\epsilon,\cc}\big\|_{\mathbb{B}(L^2(\X))}\leq\,C\,\epsilon$.
	\end{itemize}
\end{proposition}
\begin{proof}
	The following identities are evident due to the definitions of the magnetic Weyl operators corresponding to the given symbols:
	$$
	h_0\sharp^{B^\circ} p_\B=h_\B\sharp^{B^\circ} p_0=h_0\sharp^{B^\circ}\p_\infty=h_\infty\sharp^{B^\circ} p_0=h_\B\sharp^{B^\circ} h_\infty=h_\infty\sharp^{B^\circ} p_\B\,=\,h_\B\sharp^{B^\circ}(1-p_{\B})\,=\,0\, .
	$$
	Using these identities and Proposition \ref{P-replP3_5}, we get, {with $A=A(\epsilon,c)$ and $B=B(\epsilon,c)$,} that there exists some $C>0$ such that:
	\begin{align*}
		&H_{0}^{\epsilon,\cc}P_\B^{\epsilon,\cc}=\Op^{A}\big(h_{0}\sharp^{B}p_\B^{\epsilon,\cc}\big)=\epsilon\,\Op^{A}\big(\mathcal{r}^{\epsilon,\cc}(h_{0},p_\B^{\epsilon,\cc})\big)\,\Rightarrow\,\big\|H_{0}^{\epsilon,\cc}P_\B^{\epsilon,\cc}\big\|_{\mathbb{B}(L^2(\X))}\leq\,C\epsilon\,;\\
		&H_{\B}^{\epsilon,\cc}P_0^{\epsilon,\cc}=\Op^{A}\big(h_{\B}\sharp^{B}p_0^{\epsilon,\cc}\big)=\epsilon\,\Op^{A}\big(\mathcal{r}^{\epsilon,\cc}(h_{\B},p_0^{\epsilon,\cc})\big)\,\Rightarrow\,\big\|H_{\B}^{\epsilon,\cc}P_0^{\epsilon,\cc}\big\|_{\mathbb{B}(L^2(\X))}\leq\,C\epsilon\,;\\
		&H^{\epsilon,\cc}_\infty\,P_0^{\epsilon}=\Op^{A}\big(h_\infty\sharp^{B}p_0^{\epsilon,\cc}\big)=\epsilon\,\Op^{A}\big(\mathcal{r}^{\epsilon,\cc}(h_\infty,p_0^{\epsilon,\cc})\big)\,\Rightarrow\,\big\|H^{\epsilon,\cc}_\infty P_0^{\epsilon,\cc}\big\|_{\mathbb{B}(L^2(\X))}\leq\,C\epsilon\,;\\
		&H_{\B}^{\epsilon,\cc} P_\infty^{\epsilon,\cc}=\Op^{A}\big(h_{\B}\sharp^{B}p_\infty^{\epsilon,\cc}\big)=\Op^{A}\big(h_{\B}\sharp^{B}(1-p^{\epsilon,\cc}_{\B}-p^{\epsilon,\cc}_0\big)\,\Rightarrow\,\big\|H_{\B}^{\epsilon,\cc}P_\infty^{\epsilon,\cc}\big\|_{\mathbb{B}(L^2(\X))}\leq\,C\epsilon\,;\\
		&H_{\infty}^{\epsilon,\cc} P_\B^{\epsilon,\cc}=\Op^{A}\big(h_{\infty}\sharp^{B}p_\B^{\epsilon,\cc}\big)=\epsilon\,\Op^{A}\big(\mathcal{r}^{\epsilon,\cc}(h_\infty,p_\B^{\epsilon,\cc})\big)\,\Rightarrow\,\big\|H^{\epsilon,\cc}_\infty P_\B^{\epsilon,\cc}\big\|_{\mathbb{B}(L^2(\X))}\leq\, C\epsilon.
	\end{align*}
\end{proof}

\begin{corollary}\label{cordhc1} We have 
    $$\big \| \mathfrak{H}^{\epsilon,\cc}_{\mathfrak{B}}-H^{\epsilon,\cc}_\B\big \|_{\mathbb{B}(L^2(\X))}\leq\, C\epsilon.$$
\end{corollary}
\begin{proof}
    It is a direct consequence of applying the estimates of Proposition \ref{P-est-epsilonc} to the identity \eqref{dhc2}. 
\end{proof}

\subsection{Reduction to the subspace ${P^{\epsilon,\cc}_\B\,L^2(\X)}$}

\subsubsection{Schur complement and reduction to a quasi-invariant subspace.}\label{A-SchurC}

We shall consider the following abstract setting already introduced and used in \cite{CHP-2,CHP-4}.

\begin{hypothesis}\label{H-II}
	In a separable complex Hilbert space $\mathcal{H}$ we consider a family of pairs $(H_\kappa,P_\kappa)$ indexed by $\kappa\in[0,\kappa_0]$ for some $\kappa_0>0$, where $H_\kappa:\mathcal{D}(H_\kappa)\rightarrow\mathcal{H}$ is a lower-semibounded self-adjoint operator and $\Pi_\kappa=\Pi_\kappa^*=\Pi_\kappa^2$ is an orthogonal projection, such that, for  $\Pi_\kappa^\bot:=\bb1-\Pi_\kappa$, we have the properties:
	\begin{enumerate}
		\item there exists $C>0$ such that for any $\kappa\in[0,\kappa_0]$  we have that  $\Pi_\kappa\mathcal{H}\subset\mathcal{D}(H_\kappa)$ and $ \|\Pi_\kappa^\bot H_\kappa\Pi_\kappa \|_{\mathbb B(\mathcal{H})}\, \leq\,C\,\kappa$; 
		\item there exists an interval $J\subset\mathbb{R}$ with non-void interior, such that for any $\kappa\in[0,\kappa_0]$ and any $t\in J$ the operator $\Pi_\kappa^\bot H_\kappa\Pi_\kappa^\bot\,-\,t\Pi_\kappa^\bot$ is invertible  as operator in $\Pi_\kappa^\bot\mathcal{H}$ with the inverse being uniformly bounded on $J$.
	\end{enumerate}
\end{hypothesis} 

	We notice that under Hypothesis \ref{H-II} we have the identity:
	\beq\nonumber 
	\Pi_\kappa H_\kappa R^\bot_\kappa(t)H_\kappa\Pi_\kappa=\Pi_\kappa H_\kappa \Pi_\kappa^\bot R^\bot_\kappa(t)\Pi_\kappa^\bot H\Pi_\kappa\,\in\,\mathbb{B}(\mathcal{H})
	\eeq 
	and the estimate (for some $C>0$):
	\beq\label{F-est-FS}
	\big\|\Pi_\kappa H_\kappa R^\bot_\kappa(t)H_\kappa\Pi_\kappa\big\|_{\mathbb{B}(\mathcal{H})}\,\leq\,C \kappa^2\,\big\|R^\bot_\kappa(t)\big\|_{\mathbb{B}(\mathcal{H})}.
	\eeq

A simple algebraic computation allows us to prove the following statement about the Schur complement  (\cite{CH}).

\begin{proposition}\label{P-SchurC} 
	Under Hypothesis \ref{H-II} we have that:
	\begin{itemize}
		\item $t\in J\cap\sigma(H)$ if and only if $t\in J\cap\sigma\big(\Pi_\kappa H_\kappa\Pi_\kappa\,-\,\Pi_\kappa HR^\bot_\kappa(t)H_\kappa\Pi_\kappa\big)$ \,;
		\item if we denote by $R^\sim_\kappa(t):=\big(\Pi_\kappa(H_\kappa-t)\Pi_\kappa\,-\,\Pi_\kappa H_\kappa R^\bot_\kappa(t)H_\kappa\Pi_\kappa\big)^{-1}$ as operator in $\Pi_\kappa\mathcal{H}$, we have the identity
		\vspace*{-1cm}
		
		\beq\begin{split}\nonumber
			\big(H_\kappa-t\bb1_{\mathcal{H}}\big)^{-1}&=\left(\begin{array}{cc}
				\Pi_\kappa(H_\kappa-t)\Pi_\kappa & \Pi_\kappa H_\kappa\Pi_\kappa^\bot \\
				\Pi_\kappa^\bot H_\kappa\Pi_\kappa & \Pi_\kappa^\bot(H_\kappa-t)\Pi_\kappa^\bot
			\end{array}\right)^{-1}\\ &=\  
			\left(\begin{array}{cc}
				R^\sim_\kappa(t) & -R^\sim_\kappa(t)H_\kappa R^\bot_\kappa(t) \\
				-R^\bot_\kappa(t) HR^\sim_\kappa(t) & R^\bot_\kappa(t)+ R^\bot_\kappa(t)H_\kappa R^\sim_\kappa(t) H_\kappa R^\bot_\kappa(t)
			\end{array}\right).
		\end{split}\eeq 
	\end{itemize}
\end{proposition}

\begin{corollary}\label{C-SchurC} 
	Under  Hypothesis \ref{H-II}, the operator $\ham_{\kappa}:=\Pi_{\kappa}H_\kappa\Pi_{\kappa}$ defines  a bounded self-adjoint operator acting in $\Pi_{\kappa}\mathcal{H}$ and there exists $C>0$ such that, for all $\kappa\in [0,\kappa_0]$, 
	$$
	\max\left\{\underset{\lambda\in J\cap\sigma(H_\kappa)}{\sup}\dist\Big(\lambda\,,\,\sigma\big(\ham_{\kappa}\big)\Big)\ ,\ \underset{\lambda\in J\cap\sigma([H_{\kappa}])}{\sup}\dist\Big(\lambda\,,\,\sigma(H_\kappa)\Big)\right\}\,\leq\,C\,\kappa^2\,\big\|R^\bot_\kappa(t)\big\|_{\mathbb{B}(\mathcal{H})}.
	$$
\end{corollary}

\subsubsection{Verification of Hypothesis \ref{H-II}. End of the proof of Theorem \ref{T-I}}

In the context of Paragraph \ref{A-SchurC}, let us take $\mathcal{H}=L^2(\X)$ and look at the magnetic field in \eqref{Hyp-magnField} by fixing any value of $\cc\in[0,1]$ and considering $B^{\epsilon,\cc}=\epsilon\big(B^\bullet\,+\,\cc\,B^\epsilon\big)$ as depending only on the  parameter $\epsilon\in[0,\epsilon_0]$. Then we shall take $\kappa=\epsilon$ and $\Pi_\kappa=P^{\epsilon,\cc}_\B$ and $H_\kappa=H^{\epsilon,\cc}$ for any fixed value of $\cc\in[0,1]$. 

\paragraph{Verification of condition 1.} We notice that:
\beq\nonumber 
P^{\epsilon,\cc}_\B\mathcal{H}\,=-\,(2\pi i)^{-1}\int_{\mathscr{C}_\bot}d\zz \, \big(H^{\epsilon,\cc}_\bot-\zz\bb1\big)^{-1}\,\mathcal{H}\,\subset\,\mathcal{D}\big(H^{\epsilon,\cc}_\bot\big)\,=\,\mathcal{D}\big(H^{\epsilon,\cc}\big).
\eeq
The equality {between the domains} is implied by the fact that $H^{\epsilon,\cc}=H^{\epsilon,\cc}_\bot+H^{\epsilon,\cc}_\B$ with  $H^{\epsilon,\cc}_\B\in\mathbb{B}(\mathcal{H})$.
 Moreover, this last equality and Proposition \ref{P-est-epsilonc} imply that there exist $C>0$ and $\epsilon_0>0$ such that for any $(\epsilon,\cc)\in[0,\epsilon_0]\times[0,1]$ we have the following estimates:
 \begin{equation}\label{F-PH-epsc}
\begin{aligned}
	\|[P^{\epsilon,\cc}_\B]^\bot H^{\epsilon,\cc}P^{\epsilon,\cc}_\B \|_{\mathbb{B}(L^2(\X))}&\leq\ \|P^{\epsilon,\cc}_0 H^{\epsilon,\cc}P^{\epsilon,\cc}_\B \|_{\mathbb{B}(L^2(\X))}+\|P^{\epsilon,\cc}_\infty H^{\epsilon,\cc}P^{\epsilon,\cc}_\B \|_{\mathbb{B}(L^2(\X))}\\
	&\hspace*{-3cm}\leq\ \ \|P^{\epsilon,\cc}_0 H^{\epsilon,\cc}_0P^{\epsilon, \cc }_\B \|_{\mathbb{B}(L^2(\X))}+\|P^{\epsilon,\cc}_0 H^{\epsilon,\cc}_\B P^{\epsilon}_\B \|_{\mathbb{B}(L^2(\X))}+\|P^{\epsilon,\cc}_0 H^{\epsilon,\cc}_\infty P^{\epsilon,\cc}_\B \|_{\mathbb{B}(L^2(\X))}\\
	&\hspace*{-3cm}\ \ +\|P^{\epsilon,\cc}_\infty H^{\epsilon}_0P^{\epsilon,\cc}_\B \|_{\mathbb{B}(L^2(\X))}+\|P^{\epsilon,\cc}_\infty H^{\epsilon,\cc}_\B P^{\epsilon,\cc}_\B \|_{\mathbb{B}(L^2(\X))}+\|P^{\epsilon,\cc}_\infty H^{\epsilon,\cc}_\infty P^{\epsilon,\cc}_\B \|_{\mathbb{B}(L^2(\X))}\ \leq\\ 
 &\,\leq\,C\epsilon.
\end{aligned}
\end{equation}

\paragraph{Verification of condition 2.} 

Let us consider the product $P^{\epsilon,\cc}_\bot H^{\epsilon,\cc}P^{\epsilon,\cc}_\bot$. Using Remark \ref{R-p-symb} we notice that $H_\bot=P_\bot HP_\bot$ has two spectral gaps: $[0,E_0]$ and $[E_-,E_+]$, with $0<E_0\leq E_-<E_+<\infty$. We shall
prove that condition 2 in Hypothesis \ref{H-II} is verified for some interval $J\subset[E_-,E_+]$ with non-void interior and the  pair $\big(H^{\epsilon,\cc}\,,\,P^{\epsilon,\cc}_\B\big)$ in $\mathcal{H}=L^2(\X)$ where we consider, as explained above, any fixed value of $\cc\in[0,1]$ and take $\kappa=\epsilon$. Then we can write the equalities:
\begin{align}\nonumber
	P^{\epsilon,\cc}_\bot H^{\epsilon,\cc}P^{\epsilon,\cc}_\bot&=\big(P^{\epsilon,\cc}_0+P^{\epsilon,\cc}_\infty\big)\big(H^{\epsilon,\cc}_0+H^{\epsilon,\cc}_\B+H^{\epsilon,\cc}_\infty\big) \big(P^{\epsilon}_0+P^{\epsilon}_\infty\big)\\ \label{F-c2}
	&=\big(P^{\epsilon,\cc}_0 H^{\epsilon,\cc}_0 P^{\epsilon,\cc}_0\big)+\big(P^{\epsilon,\cc}_\infty H^{\epsilon,\cc}_\infty P^{\epsilon,\cc}_\infty\big)\\ \nonumber
	&\hspace*{12pt}+P^{\epsilon,\cc}_0\big(H^{\epsilon,\cc}_\B+H^{\epsilon,\cc}_\infty\big) \big(P^{\epsilon,\cc}_0+P^{\epsilon,\cc}_\infty\big)+P^{\epsilon,\cc}_\infty\big(H^{\epsilon,\cc}_0+H^{\epsilon,\cc}_\B\big) \big(P^{\epsilon,\cc}_0+P^{\epsilon,\cc}_\infty\big)\\ \nonumber
	&\hspace*{12pt}+\big(P^{\epsilon,\cc}_0 H^{\epsilon,\cc}_0 P^{\epsilon,\cc}_\infty\big)+\big(P^{\epsilon,\cc}_\infty H^{\epsilon,\cc}_\infty P^{\epsilon,\cc}_0\big)
\end{align}
and the same arguments using Proposition \ref{P-est-epsilonc} imply that there exist $C>0$ and $\epsilon_0>0$ such that for any $(\epsilon,\cc)\in[0,\epsilon_0]\times[0,1]$ we have the following estimate :
\begin{align*}
	\big\|P^{\epsilon,\cc}_0\big(H^{\epsilon,\cc}_\B+H^{\epsilon,\cc}_\infty\big) \big(P^{\epsilon,\cc}_0+P^{\epsilon,\cc}_\infty\big)&+P^{\epsilon,\cc}_\infty\big(H^{\epsilon,\cc}_0+H^{\epsilon,\cc}_\B\big) \big(P^{\epsilon,\cc}_0+P^{\epsilon,\cc}_\infty\big)+\\
	&+\big(P^{\epsilon,\cc}_0 H^{\epsilon,\cc}_0 P^{\epsilon,\cc}_\infty\big)+\big(P^{\epsilon,\cc}_\infty H^{\epsilon,\cc}_\infty P^{\epsilon,\cc}_0\big)\big\|_{\mathbb{B}(L^2(\X))}\leq\,C\epsilon.
\end{align*}
Concerning the first two terms in the last line  of \eqref{F-c2} we recall from \eqref{F-A} that for $H^{\epsilon,\cc}$ as in \eqref{DF-Heps-c} the total magnetic field is in fact (see \eqref{F-A} and \eqref{F-Ab}) of the form $B=B^\circ+B^{\epsilon,\cc}=dA$ with $A=A^\circ+\epsilon\,A^\bullet+\cc\epsilon\,A^\epsilon$ and we notice that (the bar on top of the operators in the formulas below  means closure):
\beq\nonumber \begin{split}
	\big(P^{\epsilon,\cc}_0 H^{\epsilon,\cc}_0 P^{\epsilon,\cc}_0\big)+\big(P^{\epsilon,\cc}_\infty H^{\epsilon,\cc}_\infty P^{\epsilon,\cc}_\infty\big)&=\Op^{A}\big(\p^{\epsilon,\cc}_0\sharp^{B}h_0\sharp^{B}\p^{\epsilon,\cc}_0\big)\,+\,\overline{\mathfrak{Op}^{A}\big(\p^{\epsilon,\cc}_\infty\sharp^{B}h_\infty\sharp^{B}\p^{\epsilon,\cc}_\infty\big)}\\
	&=\Op^{A}\big(\p_0\sharp^{B^\circ} h_0\sharp^{B^\circ}\p_0\big)\,+\,\overline{\mathfrak{Op}^{A}\big(\p_\infty\sharp^{B^\circ} h_\infty\sharp^{B^\circ}\p_\infty\big)}\,+\,\mathcal{O}(\epsilon)\\
	&=\Op^{A}\big(h_0\big)\,+\,\overline{\mathfrak{Op}^{A}\big(h_\infty\big)}\,+\,\mathcal{O}(\epsilon)=\overline{\mathfrak{Op}^{A}\big(h_\bot\big)}\,+\,\mathcal{O}(\epsilon)
\end{split}\eeq
in order to conclude that:
\beq\label{F-dif-h-bot}
P^{\epsilon,\cc}_\bot H^{\epsilon,\cc}P^{\epsilon,\cc}_\bot\,=\,H^{\epsilon,\cc}_\bot\,+\,\mathcal{O}(\epsilon)
\eeq
where $\mathcal{O}(\epsilon)$ is  a bounded self-adjoint operator with norm of order $\epsilon$.

The arguments in Subsection \ref{SS-free-dyn} imply that: $\sigma\big(H_\bot\big)\subset\{0\}\cup[E_0,E_- ]\cup [ E_+,+\infty)$. Thus we can use {the last inclusion listed in} \eqref{F-m-sp-est} in order to conclude that there exists $\delta_0>0$ such that for any $\delta\in(0,\delta_0]$ there exist constants $\epsilon_0(\delta)>0$ and $C>0$ such that for any pair $(\epsilon,\cc)\in[0,\epsilon_0]\times[0,1]$:
$$
 J:=\big(E_-+\delta\,,\,E_+-\delta\big)\subset\,\mathbb{R}\setminus\sigma\big(H^{\epsilon,\cc}_\bot\big).
$$ 
Thus we have proved the following statement:
\begin{proposition}\label{P-red-magn-Hamilt-0}
Let $H^\circ$ satisfying Hypothesis \ref{H-I} and \ref{H-0} and a perturbation by a magnetic field satisfying Hypothesis \ref{Hyp-magnField}. With $J_{\B}:=[E_-,E_+]\subset\R$ as in Hypothesis \ref{H-I}, there exists $\delta_0>0$  such that for any $\delta\in(0,\delta_0]$ the interval $J^{\delta}_\B:=\big(E_-+\delta\,,\,E_+-\delta\big)$ is not void and there exists a constant $\epsilon_0(\delta)>0$ such that for any pair $(\epsilon,\cc)\in[0,\epsilon_0]\times[0,1]$ and for any $t\in J^{\delta}_\B$ the operator   $(\bb1-P^{\epsilon,\cc}_\B)\big(H^{\epsilon,\cc}_\bot\,-\,t\bb1\big)(\bb1-P^{\epsilon,\cc}_\B)$ is invertible as a self-adjoint operator in  $(\bb1-P^{\epsilon,\cc}_\B)L^2(\X)$. This  inverse is denoted by $[R^\bot_{\epsilon,\cc}(t)]$ and is uniformly bounded on $J^{\delta}_\B$. 
\end{proposition}

Putting this result together with \eqref{F-dif-h-bot}, we obtain {that the} Condition 2 in Hypothesis~\ref{H-II} is satisfied, and Theorem \ref{T-I} is then a direct consequence of Proposition \ref{P-SchurC} and Corollary \ref{C-SchurC}.

\subsection{Spectral properties of the subspace ${P^{\epsilon,\cc}_\B\,L^2(\X)}$}

In this subsection we prove a result estimating the distance between the spectral measure of $H^{\epsilon,\cc}$ restricted to the energy window $\big(E_-+{ 2}\delta\,,\,E_+-{2}\delta\big)$ and \textit{the magnetic band projection} $P^{\epsilon,\cc}_\B$, result that we shall need in the proof of our Theorem \ref{C-T-I}¸

\begin{proposition}\label{P-est-Jepsilon}
 Under the hypothesis of Proposition \ref{P-red-magn-Hamilt-0} there exist $\delta_0>0$ so that the interval $$J^{\delta}_\B:=\big(E_-+{ 2}\delta\,,\,E_+-{2}\delta\big)$$ is not empty for any $\delta \in (0,\delta_0]$, and there exist $\epsilon_0 >0$ and $C>0$ such that for any pair $(\epsilon,\cc)\in[0,\epsilon_0]\times[0,1]$ and any test function $\varphi\in\,C^\infty_0(\R)$ with $\supp\varphi\subset J^{\delta}_\B$ we have the following estimation:
	  $$\big\|P^{\epsilon,\cc}_\B\,\varphi\big(H^{\epsilon,\cc}\big)\,-\,\varphi\big(H^{\epsilon,\cc}\big)\big\|_{\mathbb{B}(L^2(\X))}\,\leq\,C\,\epsilon.$$
\end{proposition}

\begin{proof}
The hypothesis $\supp\varphi\subset(E_-+2\delta,E_+-2\delta)$ implies that: 
\beq\label{F-Ecirc-h-J}
E^\circ_h(J)=E^\circ_{h_\bot}(J)\oplus\,E^\circ_{h_\B}(J)=E^\circ_{h_\B}(J)\subset\,E^\circ_{h_\bot}(\{0\})=P_\B.
\eeq
If we use Proposition 6.33 in \cite{IMP-2} and the notations:
\beq\begin{split}
&\varphi\big(H^{\epsilon,\cc}\big)\,=\,\Op^{\epsilon,\cc}\big(\varphi^{\epsilon,\cc}(h)\big),\quad\varphi^{\epsilon,\cc}(h)\in\,S^{-p}_{1}(\X^*\times\X),\\
&\varphi\big(H^{\circ}\big)\,=\,\Op^{A^\circ}\big(\varphi^{\circ}(h)\big),\quad\varphi^{\circ}(h)\in\,S^{-p}_{1}(\X^*\times\X)
\end{split}\eeq
we notice that \eqref{F-Ecirc-h-J} implies the equality:
\beq
P_\B\,\varphi\big(H^{\circ}\big)\,=\,\varphi\big(H^{\circ}\big)\quad\text{i.e.:}\quad\,p_\B\sharp^\circ\varphi^{\circ}(h)\,-\,\varphi^{\circ}(h)\,=\,0.
\eeq
Thus, the usual estimation of the magnetic perturbation  on symbols implies that:
\beq\begin{split}
\big\|P^{\epsilon,\cc}_\B\,\varphi\big(H^{\epsilon,\cc}\big)\,-\,\varphi\big(H^{\epsilon,\cc}\big)\big\|_{\mathbb{B}(L^2(\X))}\,=\,\mathscr{O}(\epsilon).
\end{split}\eeq
\end{proof}

\section{Magnetic matrices in Parseval tight-frames.}\label{S-Parseval} 

%We come now to the problem of using linear bases in order to associate complex matrices with  the observables associated with  an isolated Bloch family in a regular magnetic field and the related question of the existence of composite Wannier functions. 

%The above question is deeply connected with some mathematical aspects concerning the structure of the Bloch-Floquet representation briefly introduced in \eqref{DF-BF-trsf} and { Theorem~\ref{T-FBdec}.} 

While the theory of direct integrals of Hilbert spaces {roughly} implies that \textit{"any bounded operator $T\in\mathbb{B}\big(L^2(\X)\big)$ commuting with all the unitary translations with vectors from $\Gamma$ is unitary equivalent with the multiplication in $ \L^2\big(\T_*^d;L^2(\mathcal{E})\big)$ with a measurable operator-valued function} $\T^d_*\ni\bz^*\mapsto\widehat{T}(\bz^*)\in\mathbb{B}\big(L^2(\mathcal{E})\big)$",   the study of the image in this representation of the canonical unitary representations of $\X$ and resp. $\X^*$ on $L^2(\X)$ induces some important differential structures related to the $\Zd$-principal bundle defined by the following short sequence of groups: $ 0\hookrightarrow\Zd\hookrightarrow\Rd\repi\mathbb{S}^d\repi1$. The topological property of \textit{associated vector bundles} described in Theorem 7.2 in Chapter 8 of \cite{Hu} that we recall in the next subsection allows to construct in Subsection \ref{SS-Peps-c-B-frame} a Parseval frame for the subspace $P^{\epsilon,\cc}_\B\,L^2(\X)$.

\subsection{The Bloch-Floquet vector bundle.} \label{SSS-BF-vbdl}
 Firstly, we notice that the Bloch-Floquet unitary transformation takes the canonical unitary action of $\X$ by translations on $L^2(\X)$ into the canonical unitary action by translations on each space $\mathscr{F}_{\bz^*}$ defined in point (1) of Theorem \ref{T-FBdec}.\\
 
 Secondly, let us identify the direct integral $\int_{\T^d_*}d\bz^*\,\mathscr{F}_{\bz^*}$ with a subspace of $L^2_{\text{\tt loc}}\big(\X^*;L^2(\T^d)\big)$ after multiplication with the smooth function $\T^d_*\times\X\ni(\bz^*,x)\mapsto\,e^{-i<\s_*(\bz^*),x>}\in\mathbb{S}$ and of periodic functions, on $\X$ or $\X^*$, with functions on the corresponding torus. One obtains a unitary equivalence of $\mathscr{F}$ with the Hilbert space of functions $$ \mathscr{G}:=\big\{f\in L^2_{\text{\tt loc}}\big(\X^*;L^2(\T^d)\big),\ f(\xi+\gamma^*)=e^{-i<\gamma^*,\cdot>}f(\xi)\big\}$$  endowed with the scalar product  $$\big(f\,,\,g\big)_{\mathscr{G}}:=\int_{\mathcal{B}}d\hat{\xi}\,\big(f(\hat{\xi})\,,\,g(\hat{\xi})\big)_{L^2(\T^d)}\,.$$
 Under the above unitary transformation and after using on $\X^*=\Gamma^*\times\mathcal{B}$ a similar decomposition to \eqref{DF-X-dec-Gamma}, the canonical representation of $\X^*$ on $L^2(\X)$ becomes:
 \beq\label{DF-Xstar-repr}
\X^*\ni\xi\mapsto \hat{V}(\xi)\in\mathbb{U}(\mathscr{G}),\quad\big(\hat{V}(\xi)f\big)(\zeta,\bz):=e^{-i<\iota(\xi),\s(\bz)>}\,f(\hat{\xi}+\zeta,\bz)
\eeq
 that is exactly \textit{the equivariant representation} of $\X^*$ induced by the following representation of its sub-group $\Gamma*\subset\X^*$, \cite{Ma-51}:
 \beq\label{DF-Gstar-repr}
\Gamma_*\ni\gamma^*\mapsto \hat{V}_\circ(\gamma^*)\in\mathbb{U}\big(L^2( \T^d)\big),\quad\big(\hat{V}_\circ(\gamma^*)f\big)(\zeta,\bz):=e^{-i<\gamma^*,\s(\bz)>}\,f(\zeta,\bz).
\eeq

Thirdly, it is known that this equivariant representation may be identified with the representation by translations on sections of the vector bundle $\tp_*:\mathfrak{F}\repi\T^d_*$ associated with  the principal bundle $\X^*\repi\T^d_*$ with fibre isomorphic to $\Gamma_*$ by its representation \eqref{DF-Gstar-repr}, \cite{Bo,PP}. Let us recall that the fibre of this representation may have the following identifications (by the above unitaries):
\beq\nonumber 
\tp_*^{-1}(\bz^*)\,\cong\,\mathscr{F}_{\bz^*}\,\cong\,L^2(\T^d)\,\cong\,L^2(\mathcal{E})
\eeq 
and let us denote by $\big(\cdot,\cdot\big)_{\tp_*^{-1}(\bz^*)}$ the scalar product on this fibre. 
Smoothness of sections with respect to the representation \eqref{DF-Xstar-repr} of $\X^*$ is equivalent to smoothness for the differentiable structure of the vector bundle $\mathfrak{F}\repi\T^d_*$.

%Recalling the Riesz formula \eqref{F-R-PB} for $P_\B$, the Bloch-Floquet decomposition of the resolvent $\big(H_\bot-\zz\bb1\big)^{-1}$:
%\beq\nonumber 
%\big(H_\bot-\zz\bb1\big)^{-1}=\mathfrak{U}_{F\Gamma}^{-1}\Big[\int_{\T^d_*}d\bz^*\,\widehat{R_{\bot,\zz}}(\bz^*)\Big]\mathfrak{U}_{F\Gamma},
%\eeq
%and Remark \ref{R-p-symb}, we notice that Theorem \ref{T-FBdec} is also valid with $H_\bot$ replacing $H^\circ$ (by simply replacing $h\in S^p_1(\Xi)_\Gamma$ by $h_\bot\in S^p_1(\Xi)_\Gamma$). For any $\tilde{\bz}^*\in\T^d_*$ we can find a small open neighbourhood $\tilde{O}\subset\T^d_*$ such that we may take $\tilde{\mathscr{C}}_\B(\bz^*)=\tilde{\mathscr{C}}_\B$ constant for all $\bz^*\in\tilde{O}$. Using the second statement of Theorem \ref{T-FBdec} it follows that both applications:
%\beq 
%\begin{aligned}\label{F-301}
	%&\tilde{O}\ni\bz^*\mapsto(2\pi i)^{-1}\int_{\tilde{\mathscr{C}}_\B}\widehat{R_{\bot,\zz}}(\bz^*)\,d\zz=\widehat{P}_\B(\bz^*)\in\mathbb{L}(\mathscr{F}_{\bz^*})\\
	%&\tilde{O}\ni\bz^*\mapsto(2\pi i)^{-1}\int_{\tilde{\mathscr{C}}_\B}\widehat{R_{\bot,\zz}}(\bz^*)\,\zz\,d\zz=\widehat{H}_\B(\bz^*)\in\mathbb{B}(\mathscr{F}_{\bz^*})
%\end{aligned}
%\eeq

Within the just introduced context, the family of projections $\big\{\widehat{P}_\B(\bz^*)\big\}_{\bz^*\in\T^d_*}$ {which builds up $P_\B$ in \eqref{hdc2}} defines a smooth sub-bundle $\mathfrak{F}_\B\overset{\tp_*}{\repi}\T^d_*$ of finite rank $N+1$. 
Let us denote by $\mathfrak{S}(\mathfrak{F}_\B;O)$ the complex linear space of smooth sections of the vector bundle $\mathfrak{F}_\B\repi\T^d_*$ over the open set $O\subset\T^d_*$. We  emphasize that the smoothness of these sections implies the rapid decay of their Bloch-Floquet unitary transformed functions living in $L^2(\X)$. We also introduce the notation $\mathfrak{S}^2(\mathfrak{F}_\B;O)$ for the square integrable measurable sections of $\mathfrak{F}_\B\repi\T^d_*$ over the open subset $O\subset\T^d_*$.

We call the sub-bundle $\mathfrak{F}_\B\repi\T^d_*$ trivializable {if it exists} an identification $\mathfrak{S}(\mathfrak{F}_\B;\T^d_*)\cong C^\infty(\T^d_*;\mathbb{C}^{N+1})$.  There are a lot of results about finding necessary and sufficient conditions under which such a trivialisation is possible when the dimension $d\in \{1,2,3\}$, and which allows one to construct strongly localized composite Wannier functions  (see \cite{HS, Ne-RMP, FMP, CHN, CM, Mon, Ku, MMP}). 

For the case of a non-trivializable sub-bundle $\mathfrak{F}_\B\repi\T^d_*$ (for example for a sub-bundle having a non-zero Chern class), we shall develop in what follows a technique based on {tight}-frames which are no longer orthonormal, but which are still smooth and form an over-complete system of generators  (see \cite{AK, Ku2, CMM}). 
The main property of fibre bundles that we shall need in our construction is the following statement (Theorem 7.2 in Chapter 8 of \cite{Hu}):
\begin{theorem} [Embedding Theorem]  Given a smooth vector bundle $\mathfrak{B}\repi M$ of rank $n$ over a smooth real manifold $M$ of dimension $d$, there exists $m\in\mathbb{N}$ with 
$$0\leq m\leq\,m_d:=\inf\{k\in\mathbb{N},\ k\geq(d/2)\}$$ 
and a smooth bundle homomorphism $\mathfrak{I}:\mathfrak{B}\rightarrow M\times\mathbb{C}^{n+m}$.
\end{theorem}

Applying this to our case where $n=N+1$ and $M=\T^d_*$ we have:
\begin{corollary}\label{C-bd-triv}
There exist some $n_\B\in \Nb$ with $N+1\leq n_\B\leq\,{N+1+ m}$ and a smooth bundle homomorphism  $\mathfrak{I}_\B$ from $\mathfrak{F}_\B$ into $\T^d_*\times\mathbb{C}^{n_\B}$ that is isometric on each fibre.
\end{corollary} 
\begin{proposition}\label{P-bd-triv}
	The application $\sigma\mapsto\mathfrak{I}_\B\circ\sigma$ defines an isometric map from $\mathfrak{S}^2(\mathfrak{F}_\B,\T^d_*)$ into $ L^2(\T^d_*;\mathbb{C}^{n_\B})$ and a continuous linear embedding $\mathfrak{S}(\mathfrak{F}_\B,\T^d_*)\hookrightarrow\big[C^\infty(\T^d_*)\big]^{n_\B}$.
\end{proposition}

Let us mention that explicit constructions of a smooth vector bundle homomorphism $\mathfrak{I}_\B:\mathfrak{F}_\B\rightarrow\T^d_*\times\mathbb{C}^{n'}$ {with $n'\geq n$} may be given. The most elementary one is by using the usual covering of $\mathbb{S}^d$ with $2^d$ open charts; this  unfortunately has $n'=n2^d$ which can be much larger then the optimal value $n_\B=n+m$ with $m\leq m_d$ given by the above Embedding Theorem. In the non-trivializable case, where $d=2,3$ and $m= m_d=1$, more or less explicit constructions of the homomorphism are presented in \cite{Ku2, CM, CMM, AK}.

We finish this subsection by a very brief outlook of the construction that follows in this Section. In Paragraph \ref{SS-band-P-frame} we define a Parseval frame for the closed subspace $P_\B\,L^2(\X)$ starting from the linear embedding given by Proposition \ref{P-bd-triv}. This allows us to associate with  our pseudo-differential operators an algebra of infinite matrices having off-diagonal rapid decay.  In Subsection \ref{SS-Peps-c-B-frame} we extend a procedure of 'deformed' translations introduced in the case of constant magnetic fields by \cite{Z} and \cite{HS,Ne-LMP} in order to construct a strongly localized Parseval frame for the subspace $P^{\epsilon,\cc}_\B\,L^2(\X)$, and associate with our effective Hamiltonian an infinite matrix which is strongly localized near the diagonal. 

\subsection{The unperturbed projection}

We have the following consequence of the abstract statement in Proposition \ref{P-bd-triv}:
\subsubsection{A strongly localized Parseval tight-frame for the range of $P_\B$}\label{SS-band-P-frame}

\begin{proposition}\label{P-Pfr-Tstar-sect}
There exists a family of $n_\B$ smooth global sections $\{\hat{\psi}_{p}\}_{1\leq p\leq n_\B}$ in the Bloch bundle $\mathfrak{F}_\B\repi\T^d_*$ such that the closed complex linear space they generate is equal to $\mathfrak{S}(\mathfrak{F}_\B,\T^d_*)$ and for any section $\sigma\in\mathfrak{S}^2(\mathfrak{F}_\B,\T^d_*)$ the following equalities hold:
\beq\nonumber 
\sigma(\bz^*)\,=\,\underset{1\leq p\leq n_\B}{\sum}\,\big(\hat{\psi}_{p}(\bz^*)\,,\,\sigma(\bz^*)\big)_{\tilde{\p}_*^{-1}(\bz^*)}\hat{\psi}_{p}(\bz^*),\quad\forall\bz^*\in\T^d_*,
\eeq
with the identity
\beq\nonumber 
\big\|\sigma(\bz^*)\big\|_{\tilde{\p}_*^{-1}(\bz^*)}^2\,=\,\underset{1\leq p\leq n_\B}{\sum}\,\big|\big(\hat{\psi}_{p}(\bz^*)\,,\,\sigma(\bz^*)\big)_{\tilde{\p}_*^{-1}(\bz^*)}\big|^2,\quad\forall\bz^*\in\T^d_*.
\eeq
\end{proposition}
\begin{proof}
We notice that the bundle homomorphism $\mathfrak{I}_\B:\mathfrak{F}_\B\rightarrow\T^d_*\times\mathbb{C}^{n_\B}$ is in fact a smooth family of linear isometries $\T^d_*\ni\bz^*\mapsto\mathfrak{I}_{\B,\bz^*}\in\mathcal{L}\big(\tilde{\p}_*^{-1}(\bz^*),\mathbb{C}^{n_\B}\big)$.
For any ${\bz^*}\in\T^d_*$ let: 
\beq\label{DF-LB-theta}
\mathfrak{L}_\B({\bz^*}):=\big\{\mathfrak{I}_{\B,\bz^*}\mathfrak{v}\,,\, \mathfrak{v}\in\tilde{\p}_*^{-1}(\bz^*)\big\}\,\subset\,\mathbb{C}^{n_\B},
\eeq
so that $\mathfrak{I}_{\B,\bz^*}:\tilde{\p}_*^{-1}(\bz^*)\rightarrow\mathfrak{L}_\B({\bz^*})$ is invertible and let $Q_\B({\bz^*}):\mathbb{C}^{n_\B}\rightarrow\mathfrak{L}_\B({\bz^*})$ be the canonical orthogonal projection defined by this subspace.
If $\big\{\mathcal{e}_1,\ldots,\mathcal{e}_{n_\B}\big\}$ is the canonical orthonormal basis of $\mathbb{C}^{n_\B}$, we can define the following global sections:
\beq\label{DF-triv-sect}
\hat{\psi}_{p}({\bz^*}):=\mathfrak{I}_{\B,\bz^*}^{-1}\big(Q_\B(\bz^*)\mathcal{e}_p\big),\quad\forall (p,{\bz^*})\in\underline{n_\B}\times\T^d_*.
\eeq
Given any $\sigma\in\mathfrak{S}(\mathfrak{F}_\B,\T^d_*)$ we can write that:
\begin{align*}
\sigma(\bz^*)&=\mathfrak{I}_{\B,\bz^*}^{-1}\big(Q_\B(\bz^*)\mathfrak{I}_{\B,\bz^*}\sigma(\bz^*)\big)=\mathfrak{I}_{\B,\bz^*}^{-1}\Big(Q_\B(\bz^*)\Big(\underset{1\leq p\leq n_\B}{\sum}\big(\mathfrak{I}_{\B,\bz^*}\sigma(\bz^*)\,,\,\mathcal{e}_p\big)_{\mathbb{C}^{n_\B}}\,\mathcal{e}_p\Big)\Big)\\
&=\underset{1\leq p\leq n_\B}{\sum}\big(\mathfrak{I}_{\B,\bz^*}\sigma(\bz^*)\,,\,\mathcal{e}_p\big)_{\mathbb{C}^{n_\B}}\,\hat{\psi}_p(\bz^*)=\underset{1\leq p\leq n_\B}{\sum}\big(Q_\B(\bz^*)\mathfrak{I}_{\B,\bz^*}\sigma(\bz^*)\,,\,Q_\B(\bz^*)\mathcal{e}_p\big)_{\mathbb{C}^{n_\B}}\,\hat{\psi}_p(\bz^*)\\
&=\underset{1\leq p\leq n_\B}{\sum}\big(\sigma(\bz^*)\,,\,\hat{\psi}_p(\bz^*)\big)_{\tilde{\p}_*^{-1}(\bz^*)}\,\hat{\psi}_p(\bz^*).
\end{align*}
In a similar way we notice that:
\begin{align*}
\big\|\sigma(\bz^*)\big\|_{\tilde{\p}_*^{-1}(\bz^*)}^2&=\big\|\mathfrak{I}_{\B,\bz^*}\sigma(\bz^*)\big\|_{\mathbb{C}^{n_\B}}^2=\underset{1\leq p\leq n_\B}{\sum}\big|\big(Q_\B(\bz^*)\mathfrak{I}_{\B,\bz^*}\sigma(\bz^*),\mathcal{e}_p\big)_{\mathbb{C}^{n_\B}}\big|^2\\
&=\underset{1\leq p\leq n_\B}{\sum}\big|\big(\sigma(\bz^*),\hat{\psi}_p\big)_{\tilde{\p}_*^{-1}(\bz^*)}\big|^2.
\end{align*}
\end{proof}

{Let us {map} our $n_\B$ global smooth sections obtained in Proposition \ref{P-Pfr-Tstar-sect} {into $n_\B$} rapidly decaying {smooth} functions in $L^2(\X)$ by the inverse Bloch-Floquet transform \eqref{DF-BF-trsf}:}
\beq\label{F-PN-1}
\psi_{p}(\hat{x}+\gamma)\,:=\,(2\pi)^{-d}\int_{\T^d_*}d{\bz^*}\,e^{i<{\bz^*},\gamma>}\,\hat{\psi}_{p}(\hat{x};{\bz^*}).
\eeq
They define a well-localized finite-dimensional complex subspace of the isolated Bloch family subspace $P_\B\,L^2(\X)$. Moreover, for any $\alpha\in\Gamma$, the translated functions $\tau_\alpha\psi_{p}$ verify the equalities:
\beq \nonumber 
\big(\mathfrak{U}_{F\Gamma}\tau_\alpha\psi_{p}\big)(\hat{x};{\bz^*})=\underset{\gamma\in\Gamma}{\sum}e^{-i<{\bz^*},\gamma>}\,\psi_{p}(\hat{x}+\gamma+\alpha)=e^{i<{\bz^*},\alpha>}\hat{\psi}_{p}(\hat{x};{\bz^*})\in\big[\widehat{P}_\B({\bz^*})\,\mathscr{F}_{\bz^*}\big]
\eeq
and we conclude that:
$$\Big\{\tau_\alpha\,\psi_{p},\ \alpha\in\Gamma,\ 1\leq p\leq n_\B\Big\}\,\subset\,P_\B\,L^2(\X)\,.$$
\begin{proposition}\label{P-band-frame}
The family $\Big\{\tau_\alpha\,\psi_{p},\ \alpha\in\Gamma,\ 1\leq p\leq n_\B\Big\}$ defines a Parseval {tight-}frame (see Definition \ref{D-frame}) for the subspace $P_\B\,L^2(\X)$ associated with  the isolated Bloch family $\B$.
\end{proposition}
\begin{proof}
First let us prove that $P_\B\,L^2(\X)$ is the linear span of the family $\Big\{\tau_\alpha\,\psi_{p},\ \alpha\in\Gamma,\ 1\leq p\leq n_\B\Big\}$ over $\mathbb{C}$. Thus let us choose any $f\in P_\B\,L^2(\X)$ and consider its Bloch-Floquet transform:
	\beq \nonumber 
	\hat{f}_{\bz^*}(\hat{x}):=\big(\mathfrak{U}_{F\Gamma}f\big)(\hat{x},{\bz^*})=\big(\mathfrak{U}_{F\Gamma}P_\B f\big)(\hat{x},{\bz^*})=\hat{P}_\B({\bz^*})\big[\hat{f}_{\bz^*}(\hat{x})\big]
	\eeq
that defines a bounded measurable global section $\T^d_*\ni\bz^*\mapsto\hat{f}_{\bz^*}\in L^2(\mathcal{E})\simeq\tilde{\p}_*^{-1}(\bz^*)$ and thus an element in $\mathfrak{S}^2(\mathfrak{F}_\B;\T^d_*)$.
	Using Proposition \ref{P-Pfr-Tstar-sect} we can write that:
	\beq\nonumber 
	\hat{f}_{\bz^*}=\underset{1\leq p\leq n_\B}{\sum}\,\big(\hat{\psi}_{p}({\bz^*})\,,\,\hat{f}_{\bz^*}\big)_{\mathscr{F}_{\bz^*}}\hat{\psi}_{p}({\bz^*})
	\eeq
	and consequently:
	\begin{align*}
	f(\hat{x}+\gamma)\,&=\,\big(\mathfrak{U}_{F\Gamma}^{-1}\hat{f}\big)(\hat{x}+\gamma)=\int_{\T^d_*}d{\bz^*}\,e^{i<{\bz^*},\gamma>}\,\hat{f}_{\bz^*}(\hat{x})\\
		&=\underset{1\leq p\leq n_\B}{\sum}\,\int_{\T^d_*}d{\bz^*}\,e^{i<{\bz^*},\gamma>}\,\big(\hat{\psi}_{p}({\bz^*})\,,\,\hat{f}_{\bz^*}\big)_{\tp_*^{-1}(\bz^*)}\hat{\psi}_{p}(\hat{x},{\bz^*})\\
		&=\,\underset{1\leq p\leq n_\B}{\sum}\,\underset{\alpha\in\Gamma}{\sum}\,\big[\mathcal{F}_{\T^d_*}\big(\hat{\psi}_{p}(\cdot)\,,\,\hat{f}\big)_{\tp_*^{-1}(\cdot)}\big](\alpha)\,\psi_{p}(\hat{x}+\gamma-\alpha).
	\end{align*}
For $(\alpha,p)\in\Gamma\times\underline{n_\B}$ we introduce the notation:
\beq\nonumber 
\mathbf{\Phi}_\B(f)_{\alpha,p}:=\big[\mathcal{F}_{\T^d_*}\big(\hat{\psi}_{p}(\cdot)\,,\,\hat{f}\big)_{\tp_*^{-1}(\cdot)}\big](\alpha)
\eeq
and conclude that:
\beq\nonumber 
f=\underset{1\leq p\leq n_\B}{\sum}\,\underset{\alpha\in\Gamma}{\sum}\,\mathbf{\Phi}_\B(f)_{\alpha,p}\,[\tau_{-\alpha}\psi_p].
\eeq

In order to prove that it is a Parseval frame let us compute now the scalar products: $\big(\tau_{-\alpha}\psi_{p}\,,\,f\big)_{L^2(\X)}$ for some $f\in L^2(\X)$ and some indices $(\alpha,p)\in\Gamma\times\underline{n_\B}$:
	\begin{align*}
		\big(\tau_{-\alpha}\psi_{p},f\big)_{L^2(\X)}&=\underset{\gamma\in\Gamma}{\sum}\int_{\mathcal{E}}d\hat{x}\,f(\hat{x}+\gamma)\,\overline{\psi_{p}(\hat{x}+\gamma-\alpha)}=\int_{\T^d_*}d{\bz^*}\,e^{i<{\bz^*},\alpha>}\int_{\mathcal{E}}d\hat{x}\,\hat{f}_{\bz^*}(\hat{x})\,\overline{\hat{\psi}_{p}(\hat{x},{\bz^*})}\\
	&=\int_{\T^d_*}d{\bz^*}\,e^{i<{\bz^*},\alpha>}\,\big(\hat{\psi}_{p}({\bz^*})\,,\,\hat{f}_{\bz^*}\big)_{\tp_*^{-1}(\bz^*)}=\mathbf{\Phi}_\B(f)_{\alpha,p}.
	\end{align*}
Using once again Proposition \ref{P-Pfr-Tstar-sect} we obtain finally:
	\begin{align*}
		\big\|f\big\|_{L^2(\X)}^2&=\underset{\gamma\in\Gamma}{\sum}\int_{\mathcal{E}}d\hat{x}\,\big|f(\hat{x}+\gamma)\big|^2=\int_{\T^d_*}d{\bz^*}\,\big\|\hat{f}_{\bz^*}\big\|_{\mathscr{F}_{{\bz^*}}}^2=\int_{\T^d_*}d{\bz^*}\Big(\underset{1\leq p\leq n_\B}{\sum}\,\big|\big(\hat{\psi}_{p}({\bz^*})\,,\,\hat{f}_{\bz^*}\big)_{\mathscr{F}_{\bz^*}}\big|^2\Big)\\
		&=\underset{1\leq p\leq n_\B}{\sum}\,\underset{\alpha\in\Gamma}{\sum}\Big|\mathbf{\Phi}_\B(f)_{\alpha,p}\Big|^2=\underset{1\leq p\leq n_\B}{\sum}\,\underset{\alpha\in\Gamma}{\sum}\Big|\big(\tau_{-\alpha}\psi_{p},f\big)_{L^2(\X)}\Big|^2
	\end{align*}
 and we notice that our map $\bPhi$ is in fact the coordinate map associated with  the Parseval frame $\{\tau_{-\alpha}\psi_{p}\}_{(\alpha,p)\in\Gamma\times\underline{n_\B}}$ as recalled in \eqref{DF-fr-1} in Appendix \ref{Ap-B}.
\end{proof}

The previous identity leads to the following corollary:
\begin{corollary}
	The coordinate map defined as in \eqref{DF-fr-1}, associated with  a Parseval frame:
	\beq\nonumber 
	\mathfrak{C}_\B:P_\B\,L^2(\X)\ni\,f\,\mapsto\,\big((\tau_{\alpha}\psi_{p}\,,\,f)_{L^2(\X)}\big)_{(\alpha,p)\in\Gamma\times\underline{n_\B}}\,\in\,[\ell^2(\Gamma)]^{n_\B}\cong\ell^2(\Gamma)\otimes\Co^{n_\B}
	\eeq
	is an isometry.
\end{corollary}

Notice that this isometry  may not be surjective.

\subsubsection{Strongly off-diagonal localized matrices for operators in ${\mathbb{B}\big(P_\B\,L^2(\X)\big)}$.}
\label{SSS-band-W-frame}

\begin{definition}
	Let $\mathcal{e}_{\gamma,p}:=\mathcal{e}_\gamma\otimes\mathcal{e}_{p}$ be the canonical basis of $[\ell^2(\Gamma)]\otimes\Co^{n_\B}$, $\mathcal{K}_\B:=\ell^2(\Gamma)\otimes\Co^{n_\B}$ and let $\widetilde{P}_{\B}:[\ell^2(\Gamma)]^{n_\B}\,\rightarrow\,\mathfrak{C}_\B[P_\B\,L^2(\X)]$ denote the orthogonal projection on the image of $\mathfrak{C}_\B:P_\B\,L^2(\X)\rightarrow\mathcal{K}_\B$.
\end{definition}

	The following identities hold true:
\[
\mathfrak{C}_\B\,=\,\underset{(\alpha,p,j)\in\Gamma\times\underline{n_\B}}{\sum}[\mathcal{e}_{\alpha,p}]\bowtie[\tau_{\alpha}\psi_{p}],\quad\,\widetilde{P}_{\B}=\mathfrak{C}_\B\,\mathfrak{C}_\B^*,\quad\Id_{P_\B\,L^2(\X)}=\mathfrak{C}_\B^*\,\mathfrak{C}_\B\,,
\]
and the map $\widetilde{\mathfrak{C}}_\B:\mathbb{B}\big(P_\B\,L^2(\X)\big)\rightarrow\mathbb{B}\big(\mathcal{K}_\B\big)$ introduced  in Definition \ref{DF-coord-op-hom} is an isometric homomorphism of $C^*$-algebras, and thus an isometric isomorphism on its image that is equal to $\widetilde{P}_{\B}\,\mathbb{B}(\mathcal{K}_\B)\widetilde{P}_{\B}$.

We shall use the canonical orthonormal basis in $\mathcal{K}_\B\cong\ell^2(\Gamma)\otimes\Co^{n_\B}$ in order to associate infinite matrices with operators in $\mathbb{B}\big(P_\B\,L^2(\X)\big)$. On the algebra $\mathscr{M}_{n_\B}$ (see Notation \ref{N-matrix}) we shall work with the $C^*$-norm $\|\cdot\|_{\mathscr{M}_{n_\B}}$ defined by the operator norm for operators on $\Co^{n_\B}$ and on $\mathscr{M}_{\Gamma}^\circ[\mathscr{M}_{n_\B}]$ we shall work with the $C^*$-norm {$\norm{\cdot}$} defined by the operator norm of the matrix considered as linear operator on $\mathcal{K}_\B$.
\begin{definition}
Given a bounded operator $T\in\mathbb{B}\big(P_\B\,L^2(\X)\big)$, we define its infinite matrix $\mathfrak{M}[T]$ associated {to} the Parseval frame $\big\{\tau_{-\alpha}\,\psi_p\big\}_{(\alpha,p)\in\Gamma\times\underline{n_\B}}$ {through its elements}:
$$
\mathfrak{M}[T]_{(\alpha,p),(\beta,q)}\,:=\,\big(\tau_{-\alpha}\,\psi_p\,,\,T\,\tau_{-\beta}\,\psi_q\big)_{L^2(X)}{=\big(\mathcal{e}_{\alpha,p}\,,\,\widetilde{\mathfrak{C}}_\B(T)\,\mathcal{e}_{\beta,q}\big)_{\mathcal{K}_\B}.}
$$
\end{definition}

 \begin{remark}
Using some previous arguments one can easily prove that for any pair of operators $S$ and $T$ such that $\mathfrak{M}_\B[T]\in\mathscr{M}^\circ_\Gamma[\mathscr{M}_{n_\B}]$ and $\mathfrak{M}_\B[S]\in\mathscr{M}^\circ_\Gamma[\mathscr{M}_{n_\B}]$, we have that $$\mathfrak{M}_\B[T S]\,=\,\mathfrak{M}_\B[T]\,\cdot\, \mathfrak{M}_\B[S]\in\mathscr{M}^\circ_\Gamma[\mathscr{M}_{n_\B}]$$
	 where we have denoted by "$\cdot$"  the natural extension of the usual matrix product to infinite matrices  having rapid decay off the diagonal.
  \end{remark}

\begin{remark} 
For any bounded operator $T\in\mathbb{B}\big(P_\B\,L^2(\X)\big)$,  we have the equality:
		\beq\nonumber 
	\sigma\big(\mathfrak{M}_\B[T]\big)\,=\,\left\{\begin{array}{l}\vspace*{0.1cm}
		     \sigma\big(T\big)\,\bigcup\,\{0\}\quad\text{if $\widetilde{\mathfrak{C}}_\B$ is not injective},\\
		     \sigma\big(T\big)\quad\text{if $\widetilde{\mathfrak{C}}_\B$ is injective}.
		\end{array}\right.
		\eeq
\end{remark}

\begin{proposition}
Suppose given $T\in\mathbb{B}\big(P_\B\,L^2(\X)\big)$ that commutes with all the translations with elements from $\Gamma$ and its Bloch-Floquet transform defines a smooth map $\T_*\ni\bz^*\mapsto\widehat{T}(\bz^*)\in\mathbb{B}\big(\mathscr{F}_{\bz^*}\big)$. Then $\mathfrak{M}_{\B}[T]$ belongs to $\mathscr{M}_\Gamma^\circ[\mathscr{M}_{n_{\B}}]$ and for any $(\alpha,\beta)\in\Gamma\times\Gamma$ we have the identity $\mathfrak{M}_{\B}[T]_{\alpha,\beta}=\mathfrak{M}_{\B}[T]_{\alpha-\beta,0}\in\mathscr{M}_{n_{\B}}$.
\end{proposition}
\begin{proof}
We use the {smoothness} of the {tight-}frame given by Proposition \ref{P-Pfr-Tstar-sect} and of the Bloch-Floquet transform of $T\in\mathbb{B}\big(P_\B\,L^2(\X)\big)$ in order to obtain the rapid off-diagonal decay. The commutation with all the translations with elements from $\Gamma$ then implies the last equalities.
\end{proof}
\begin{remark}
We shall use the above proposition for the operators $P_{\B}$ and $\ham_{\B}$.
\end{remark}

We introduce the following notations:
\beq\label{DF-m-circ-B}
[\mathfrak{m}^\circ_\B]_\gamma\,:=\,\mathfrak{M}_\B[H_\B]_{\gamma,0}\in\mathscr{M}_{n_\B},\quad\widehat{\mathfrak{m}}_\B(\bz^*)_{p,q}:=\big(\hat{\psi}_p(\bz^*)\,,\,\widehat{H}_\B(\bz^*)\,\hat{\psi}_q(\bz^*)\big)_{\mathscr{F}_{\bz^*}}.
\eeq

\subsection{The magnetic strongly localized tight-frame}\label{SS-Peps-c-B-frame}

For studying the effective dynamics in the subspace $P^{\epsilon,\cc}_\B\,L^2(\X)$ associated with the isolated Bloch family $\B$ in the perturbing magnetic field \eqref{Hyp-magnField}, an important tool is an analog of the Parseval frame constructed in Paragraph \ref{SS-band-P-frame} in which we replace the usual translations with elements from $\Gamma$ by Zak magnetic translations. Following the ideas in \cite{Ne-RMP}, we observe  that for a constant magnetic field, one may define a twisted representation of $ \Gamma $ on $L^2(\X)$ (see point 1 in Proposition \ref{Zak-transl}), commuting with the magnetic quantization of any $\Gamma$-periodic symbol; it appears in a series of previous references (\cite{Z} and the mathematical developments in \cite{HS1,Sj,HS,Ne-LMP,Ne-RMP,CHN}) and we shall call its operators \textit{the Zak translations}. Considering weak non-constant perturbing magnetic fields obliged us to a rather complex construction (presented in the next paragraphs \ref{SS-Z-trsl} - \ref{SS-EndProof-T-III} of this subsection) due to the fact that the simple-minded extension of the Zak translations to non-constant magnetic fields does no longer provide us with a twisted group representation. 
Our argument has two main steps:
\begin{itemize}
    \item In a first step (Paragraphs \ref{SS-Z-trsl} - \ref{SSS-PB-eps0-Wframe}), we use the Zak translations associated with  the constant part $\epsilon\,B^\bullet$ of the perturbing magnetic field in order to define a Parseval frame for the subspace $P^{\epsilon,0}_\B\,L^2(\X)$. Then in Paragraph \ref{SSS-const-mf-Pfr} we consider the associated constant field magnetic matrices and their special form given in Proposition \ref{P-B2}.
    \item In a second step (Paragraph \ref{SSS-nonconst-mf-Pfr}),  we deform the subspace $P^{\epsilon,0}_\B\,L^2(\X)$ by taking into account the phase functions $\widetilde{\Lambda}^{\epsilon,\cc}$ defined  in \eqref{F-dec-p-m-phase} by the weak non-constant perturbing magnetic field $\epsilon\cc\,B^\epsilon$, obtaining a Parseval frame for the deformed subspace $Q^{\epsilon,\cc}_\B\,L^2(\X)$ that we prove to be unitary equivalent with $P^{\epsilon,\cc}_\B\,L^2(\X)$. \\
    Finally, in Paragraph \ref{SS-EndProof-T-III} we prove a number of estimations that allow to get  the conclusion \eqref{F-I} in Theorem \ref{T-III} from Theorem \ref{T-I} which was proven  in Section \ref{S-proof-T-I}.
\end{itemize}

\subsubsection{The magnetic infinite matrices}\label{SSS-m-inf-matr}

With the same  intuition as behind the construction of the modified Wannier function in {\cite{Ne-RMP,CMM} }together with our previous analysis in \cite{CHP-1}, let us  introduce the following system of functions using the notations \eqref{F-dec-p-m-phase}:
\beq\label{DF-m-W-f}
\mathfrak{T}^{\epsilon,\cc}_\gamma\psi_p\,:=\,\Lambda^{\epsilon,\cc}( \cdot ,\gamma)\,\tau_{-\gamma}\psi_p,\quad\forall(\gamma,p)\in\Gamma\times\underline{n_\B}.
\eeq
We emphasize that the above unitary operators $\mathfrak{T}^{\epsilon,\cc}_\gamma$, are different of the unitary Weyl translations appearing in \eqref{F-UAx}:
		\[
U^{\epsilon,\cc}_\gamma\,f=\Lambda^{\epsilon,\cc}(\cdot ,\cdot +\gamma)\,\tau_{\gamma}f.
		\] 

{{Our strategy is to modify the family \eqref{DF-m-W-f} using Proposition \ref{P-band-frame} in order to obtain a Parseval frame for the perturbed subspace $P^{\epsilon,\cc}_\B\,L^2(\X)$ and to use this frame (as briefly explained in Appendix \ref{Ap-B}) in order to transport the algebra of operators on $P^{\epsilon,\cc}_\B\,L^2(\X)$ into an algebra of infinite magnetic matrices (see \eqref{DF-C-eps-c-B}}.} Recalling the notation introduced in \eqref{F-dec-p-m-phase}, we introduce the following notations:}
\beq \label{dhc5}
\Lambda^{\epsilon,\cc}_\gamma(x):=\Lambda^{\epsilon,\cc}(x,\gamma),\quad\widetilde{\Lambda}^{\epsilon,\cc}_\gamma(x):=\widetilde{\Lambda}^{\epsilon,\cc}(x,\gamma),\quad\Lambda^{\epsilon}_\gamma(x):=\Lambda^{\epsilon}(x,\gamma).
\eeq

The analysis in this paragraph could be used in order to obtain a rough version of  Theorem \ref{T-III} in which $\mathfrak{m}^\epsilon_\B$ would have been replaced by $\mathfrak{m}^\circ_\B$, with the price of having an error of order $\epsilon$ in formula \eqref{F-I} in Theorem \ref{T-III} instead of $\cc\epsilon$. Having in mind the Remark \ref{R-main} and the fact that in many interesting situations, like those discussed in \cite{CHP-1, CHP-2}, the spectral islands produced by the quantization $\widetilde{\Op}^{\epsilon,\cc}$ are of order $\epsilon$, an error of the same order would make the result insignificant; the error of order $\epsilon(\epsilon+\cc)$ appearing in \eqref{F-I} is thus essential for the significance of our result. Looking more attentively at formula \eqref{F-4_12}  while recalling that $\Lambda_\gamma^{\epsilon,\cc}=\widetilde{\Lambda}_\gamma^{\epsilon,\cc}\,\Lambda^\epsilon_\gamma$ and using the constant magnetic field Zak translations defined in the following paragraph {and the Parseval frame defined in \eqref{dhc25} and its coordinate map \eqref{DF-C-eps-c-B}}, allows us to put into evidence the matrix-valued symbol $\mathfrak{m}^\epsilon_\B$ and to obtain the estimation in \eqref{F-I}. {Before starting this programme, let us consider the infinite matrix associated with  an integral operator with respect to the system of functions in \eqref{DF-m-W-f} and notice the following interesting property that they have.}

\begin{proposition}
Given a magnetic field as in Hypothesis \ref{H-magnField} and a $\Gamma$-periodic distribution kernel $\mathfrak{K}\in\mathring{\mathscr{S}}
(\X\times\X)_\Gamma$ (see the Notation \ref{N-Scirc} in Appendix \ref{A-m-PsiDO}), {the infinite matrix associated with  the operator with integral kernel $\Lambda^{\epsilon,\cc}\mathfrak{K}$ with respect to the system of functions in \eqref{DF-m-W-f},} is of the form 
\beq\label{F-4_12}
\big(\mathfrak{T}^{\epsilon,\cc}_\alpha\psi_p\,,\,\big[\Int\,\Lambda^{\epsilon,\cc}\mathfrak{K}\big]\mathfrak{T}^{\epsilon,\cc}_\beta\, \psi_q\big)_{L^2(\X)} =\Lambda^{\epsilon,\cc}(\alpha,\beta)\mathring{\mathfrak{K}}_{\alpha-\beta}+\epsilon\mathfrak{m}^{\epsilon,\cc}[\mathfrak{K}]_{\alpha,\beta}
\eeq
 where  $\mathring{\mathfrak{K}}\in{\cal{s}}(\Gamma;\MmN)$ is independent of $(\epsilon,\cc)$ and the  family 
 $\big\{\mathfrak{m}^{\epsilon,\cc}[\mathfrak{K}],\,(\epsilon,c)\in [0,\epsilon_0]\times [0,1]\big\}$  belongs to  a bounded set of  $\mathscr{M}^\circ_\Gamma[\MmN]$.
\end{proposition}
\begin{proof}
Let us compute:
\begin{align}\label{F-B-3}
\big(\mathfrak{T}^{\epsilon,\cc}_\alpha\psi_p\,,\,[\Int\,\Lambda^{\epsilon,\cc}\mathfrak{K}]\mathfrak{T}^{\epsilon,\cc}_\beta\, \psi_q\big)_{L^2(\X)}=\big\langle\,\mathfrak{K}\,,\,\big(\overline{\Lambda^{\epsilon,\cc}_\alpha}\otimes\Lambda^{\epsilon,\cc}_\beta\big)\,\Lambda^{\epsilon,\cc}\,\big(\overline{\tau_{-\alpha}\psi_p}\otimes\tau_{-\beta}\psi_q\big)\big\rangle_{\mathscr{S}(\X\times\X)}.
\end{align}
We notice that:
\begin{align*}
\big[\big(\overline{\Lambda^{\epsilon,\cc}_\alpha}\otimes\Lambda^{\epsilon,\cc}_\beta\big)\,\Lambda^{\epsilon,\cc}](x,y)&=\exp\Big[-i\Big(\int_{[\alpha,x]}A^{\epsilon,\cc}+\int_{[y,\beta]}A^{\epsilon,\cc}+\int_{[x,y]}A^{\epsilon,\cc}\Big]\\
&=\Lambda^{\epsilon,\cc}(\alpha,\beta)\,\exp\Big[-i\Big(\int_{<\alpha,x,y>}B^{\epsilon,\cc}\Big)\Big]\,\exp\Big[-i\Big(\int_{<\alpha,y,\beta>}B^{\epsilon,\cc}\Big)\Big];
\end{align*}
\begin{align*}
\Big|\big[\big(\overline{\Lambda^{\epsilon,\cc}_\alpha}\otimes\Lambda^{\epsilon,\cc}_\beta\big)\,\Lambda^{\epsilon,\cc}](x,y)\,-\,\Lambda^{\epsilon,\cc}(\alpha,\beta)\Big|\leq\,C\epsilon\big[<x-\alpha><x-y>+<y-\beta><x-y>\big].
\end{align*}
Let us also compute a derivative, for example:
\begin{align*}
\partial_{x_l}\int_{<\alpha,x,y>} & B^{\epsilon,\cc}=\partial_{x_l}\Big(\underset{1\leq j,k\leq d}{\sum}(x-\alpha)_j(y-x)_k\int_0^1ds\int_0^sdr\,B_{j,k}\big(\alpha+s(x-\alpha)+r(y-x)\big)\Big)\\
&=\underset{1\leq k\leq d}{\sum}(y-x)_k\int_0^1ds\int_0^sdr\,B_{l,k}\big(\alpha+s(x-\alpha)+r(y-x)\big)\\
&+\underset{1\leq k\leq d}{\sum}(x-\alpha)_k\int_0^1ds\int_0^sdr\,B_{k,l}\big(\alpha+s(x-\alpha)+r(y-x)\big)\\
&+\underset{1\leq j,k\leq d}{\sum}(x-\alpha)_j(y-x)_k\int_0^1ds\int_0^sdr\,(s-r)\big(\partial_{x_l}B_{j,k}\big)\big(\alpha+s(x-\alpha)+r(y-x)\big)
\end{align*}
in order to conclude by iteration that all the derivatives of $\big(\overline{\Lambda^{\epsilon,\cc}_\alpha}\otimes\Lambda^{\epsilon,\cc}_\beta\big)\,\Lambda^{\epsilon,\cc}$ may be bounded by polynomials of the form:
$$
\epsilon<x-\alpha>^{M_1}<y-\beta>^{M_2}<x-y>^{M_3}.
$$
Denoting by $\mathfrak{q}_N(x):=<x>^N$ we notice that:
\begin{align*}
\Big[\big[\big(\overline{\Lambda^{\epsilon,\cc}_\alpha}\otimes\Lambda^{\epsilon,\cc}_\beta\big)\,\Lambda^{\epsilon,\cc}]\,-\,\Lambda^{\epsilon,\cc}(\alpha,\beta)\Big]\big((\tau_{-\alpha}\mathfrak{q}_N)\otimes(\tau_{-\beta}\mathfrak{q}_N)\big)\big(\overline{\tau_{-\alpha}\psi_p}\otimes\tau_{-\beta}\psi_q\big)\in\,\mathscr{S}(\X\times\X).
\end{align*}
By hypothesis, for any $N\in\mathbb{N}$, the product of the smooth kernel $$(x,y)\mapsto \mathfrak{k}_N(x,y):=<x-y>^N$$ with $\mathfrak{K}$ defines a distribution kernel $\mathfrak{K}^\prime_N:=\mathfrak{k}_N\mathfrak{K}\in\mathring{\mathscr{S}}(\X\times\X)$ and thus
\begin{align*}
&\big\langle\,\mathfrak{K}\,,\,\big(\overline{\Lambda^{\epsilon,\cc}_\alpha}\otimes\Lambda^{\epsilon,\cc}_\beta\big)\,\Lambda^{\epsilon,\cc}\,\big(\overline{\tau_{-\alpha}\psi_p}\otimes\tau_{-\beta}\psi_q\big)\big\rangle_{\mathscr{S}(\X\times\X)}=\big\langle\,\mathfrak{K}\,,\,\Lambda^{\epsilon,\cc}(\alpha,\beta)\,\big(\overline{\tau_{-\alpha}\psi_p}\otimes\tau_{-\beta}\psi_q\big)\big\rangle_{\mathscr{S}(\X\times\X)}\\
&=\big\langle\,\mathfrak{K}^\prime_N\,,\,\mathfrak{k}_{-N}\Big[\big[\big(\overline{\Lambda^{\epsilon,\cc}_\alpha}\otimes\Lambda^{\epsilon,\cc}_\beta\big)\,\Lambda^{\epsilon,\cc}]\,-\,\Lambda^{\epsilon,\cc}(\alpha,\beta)\Big]\big(\overline{\tau_{-\alpha}\psi_p}\otimes\tau_{-\beta}\psi_q\big)\big\rangle_{\mathscr{S}(\X\times\X)}.
\end{align*}
Considering $\mathfrak{k}_{-N}\Big[\big[\big(\overline{\Lambda^{\epsilon,\cc}_\alpha}\otimes\Lambda^{\epsilon,\cc}_\beta\big)\,\Lambda^{\epsilon,\cc}]\,-\,\Lambda^{\epsilon,\cc}(\alpha,\beta)\Big]\big(\overline{\tau_{-\alpha}\psi_p}\otimes\tau_{-\beta}\psi_q\big)\in\mathscr{S}(\X\times\X)$ we notice that:
\begin{align*}
<x-y>^{-N}=<(x-\alpha)-(y-\beta)+(\alpha-\beta)>^{-N}\leq\,C_N<x-\alpha>^N<y-\beta>^N<\alpha-\beta>^{-N}
\end{align*}
so that
\begin{align*}
\Big[\mathfrak{k}_{-N}\Big[\big[\big(\overline{\Lambda^{\epsilon,\cc}_\alpha}\otimes\Lambda^{\epsilon,\cc}_\beta\big)\,\Lambda^{\epsilon,\cc}]\,-\,\Lambda^{\epsilon,\cc}(\alpha,\beta)\Big]\big(\overline{\tau_{-\alpha}\psi_p}\otimes\tau_{-\beta}\psi_q\big)\Big](x,y)=<\alpha-\beta>^{-N}\Phi_{N,\alpha,\beta}(x,y)
\end{align*}
with $\Phi_{N,\alpha,\beta}\in\mathscr{S}(\X\times\X)$ uniformly in $(\alpha,\beta)\in\Gamma\times\Gamma$.

Finally, the main contribution {reads as}:
\begin{align*}
\big\langle\,\mathfrak{K}\,,\,\Lambda^{\epsilon,\cc}(\alpha,\beta)\,\big(\overline{\tau_{-\alpha}\psi_p}\otimes\tau_{-\beta}\psi_q\big)\big\rangle_{\mathscr{S}(\X\times\X)}&=\Lambda^{\epsilon,\cc}(\alpha,\beta)\,\big\langle\,(\tau_\alpha\otimes\tau_\beta)\mathfrak{K}\,,\,\big(\overline{\psi_p}\otimes\psi_q\big)\big\rangle_{\mathscr{S}(\X\times\X)}\\
&=\Lambda^{\epsilon,\cc}(\alpha,\beta)\,\big(\psi_p\,,\,\Int\big[(\tau_{\alpha-\beta}\otimes\bb1)\mathfrak{K}\big]\psi_q\big)_{L^2(\X)},
\end{align*}
where we may identify 
\[
\big [\mathring{\mathfrak{K}}_{\gamma}\big ]_{pq}:=\big(\psi_p\,,\,\Int\big[(\tau_{\alpha-\beta}\otimes\bb1)\mathfrak{K}\big]\psi_q\big)_{L^2(\X)},\quad \forall \gamma\in\Gamma,\quad 1\leq p,q\leq n_{\B}.
\]
\end{proof}

\subsubsection{The Zak translations in constant magnetic field}\label{SS-Z-trsl}

\begin{definition}
Given a constant magnetic field of the form $\epsilon\,B^\bullet$, we call \emph{ Zak magnetic  translations} with vectors from $\Gamma$  the following family of twisted unitary translations on $L^2(\X)$ (using the notation in \eqref{F-UAx} and \eqref{F-dec-p-m-phase})
	\beq\label{DF-Z-trsl}	{\mathfrak{T}^{\epsilon}_\gamma f (x)} :=\Lambda^\epsilon(x,\gamma)\,f(x-\gamma)=\Lambda^\epsilon_\gamma(x)\,f(x-\gamma),\quad\
 forall\gamma\in\Gamma.
\eeq
\end{definition}
\noindent We notice that $\mathfrak{T}^{\epsilon}_\gamma=\mathfrak{T}^{\epsilon,0}_\gamma$ for $\mathfrak{T}^{\epsilon,\cc}_\gamma$ defined in \eqref{DF-m-W-f} and recall the following (see e.g. \cite{CHN}):

\begin{proposition}\label{Zak-transl}
	{The}  family  in \eqref{DF-Z-trsl} satisfies the following properties:
	\begin{enumerate}
		\item The map $\mathfrak{T}^{\epsilon}:\Gamma\ni\gamma\rightarrow\mathfrak{T}^{\epsilon}_\gamma\in\mathbb{U}\big(L^2(\X)\big)$ defines a projective representation with 2-cocycle $\overline{\Lambda^{\epsilon}}:\Gamma\times\Gamma\rightarrow\mathbb{S}^1$ i.e.
		$$\mathfrak{T}^{\epsilon}_\alpha\mathfrak{T}^{\epsilon}_\beta= \Lambda^\epsilon(\beta,\alpha)\, \mathfrak{T}^{\epsilon}_{\alpha+\beta}\,.$$
		\item The operator $\mathfrak{Op}^{\epsilon}(F)$ commutes with all the $\{\mathfrak{T}^{\epsilon}_\gamma\}_{\gamma\in\Gamma}$ if and only if $F\in\mathscr{S}^\prime(\Xi)$ is $\Gamma$-periodic with respect to the variable in $\X$.
	\end{enumerate}
\end{proposition}

While working with a constant magnetic field $\epsilon\,B^\bullet$ and with the 'transverse gauge' for its associated vector potential $A^\bullet_j(x)=(1/2)\underset{1\leq k\leq d}{\sum}\,B^\bullet_{k,j}\,x_k$, {by using that $B^\bullet_{k,j}=-B^\bullet_{j,k}$} we notice that:
\beq \label{R-wedge} \begin{split}
\int_{[x,y]}A^\bullet&=\int_0^1ds\,\underset{1\leq j\leq d}{\sum}A^\bullet_j(x+s(y-x))\,(y_j-x_j)\, \\
&=\int_0^1ds\,(1/2)\underset{1\leq j,k\leq d}{\sum}B^\bullet_{k,j}(x_k+s(y_k-x_k))\,(y_j-x_j)\\
&=(1/2)\underset{1\leq j,k\leq d}{\sum}B^\bullet_{k,j}\,x_k\,(y_j-x_j)=(1/2)\underset{1\leq j,k\leq d}{\sum}B^\bullet_{k,j}\,x_k\,y_j\,:=\,<B^\bullet,x\wedge y>\,,
\end{split}\eeq
and will make use of this last notation.

\subsubsection{The Parseval tight-frame with a constant magnetic field}
\label{SSS-constmf-Wframe}

Using the Zak translations as defined in \eqref{DF-Z-trsl}, we start by defining the family 
\beq\label{DF-epsilon-Wframe}
\big\{\mathfrak{T}^\epsilon_\gamma\,\psi_{p},\ \gamma\in\Gamma,\,p\in\underline{n_\B}\big\}\subset\mathscr{S}(\X).
\eeq

Due to the Parseval property of the non-magnetic frame in Proposition \ref{P-band-frame} we have the equalities:
\beq \label{dhc11}
\tau_\alpha\psi_p=\sum_{\gamma, r} \big (\tau_\gamma \psi_r,\tau_\alpha \psi_p\big )_{L^2(\X)}\,\tau_\gamma\psi_r,\quad\forall(\alpha,p)\in\Gamma\times\underline{n_\B}.
\eeq
We conclude that the projection operator associated with  the isolated Bloch band $\B$ has the integral kernel:
$$\mathfrak{K}[P_\B](x,y):= \sum_{\alpha,p} \tau_\alpha \psi_p(x)\, \overline{\tau_\alpha \psi_p(y)} $$
that is well defined due to the rapid decay of the functions of the chosen Parseval frame, is a pseudo-differential operator with regular symbol $p_\B\in S^{-\infty}(\X\times\X^*)$ and satisfies the equality:
\beq \label{dhc12}
\int_{\X} \mathfrak{K}[P_\B](x,y)\,\mathfrak{K}[P_\B](y,x')\, dy=\mathfrak{K}[P_\B](x,x').
\eeq

Having in mind \eqref{DF-Z-trsl}, the operator $\widetilde{P}_\B^\epsilon$ with integral kernel:
\beq\label{dhc23}
\mathfrak{K}[\widetilde{P}_\B^\epsilon](x,y):=\sum_{\alpha\in \Gamma,p\in \underline{n_\B}} \mathfrak{T}^\epsilon_\alpha \psi_p(x)\, \overline{\mathfrak{T}_\alpha^\epsilon \psi_p(y)}
\eeq
also obeys 
\beq\label{F-ker-tPBeps}
\mathfrak{K}[\widetilde{P}_\B^\epsilon](x,y)= \Lambda^\epsilon(x,y)\,\Big (\mathfrak{K}[P_\B](x,y)+\epsilon \mathfrak{K}_1^\epsilon(x,y) \Big )
\eeq
where $\mathfrak{K}^\epsilon_1\in BC^\infty(\X\times\X)$ and has rapid off-diagonal decay. Hence $\widetilde{P}_\B^\epsilon$ is bounded and self-adjoint, but it is not an orthogonal projection (it lacks idempotency). 

From \eqref{F-ker-tPBeps} we deduce that:
\beq\label{F-PepsB}
\widetilde{P}_\B^\epsilon=\Op^\epsilon(p_\B)\,+\,\epsilon\, \Int\,\big (\Lambda^\epsilon \mathfrak{K}^\epsilon_1\big).
\eeq
The operator $\Int\,\Big (\Lambda^\epsilon \mathfrak{K}^\epsilon_1\Big )$ is bounded and has a regular $\epsilon\,A^\bullet$-symbol. The magnetic pseudo-differential calculus then implies that $\widetilde{P}^\epsilon_\B$ also has a regular $\epsilon\,A^\bullet$-symbol and thus a regular integral kernel with rapid off-diagonal decay (which can also be seen from \eqref{F-ker-tPBeps}). 

Using \eqref{dhc12} and \eqref{F-ker-tPBeps} we can show that 
the kernel of $(\widetilde{P}_\B^\epsilon)^2-\widetilde{P}_\B^\epsilon$ is of order $\epsilon $ (for the operator norm) and has rapid off-diagonal decay. 
If $\epsilon$ is small enough, the spectral theorem implies that {there exists some constant $C>0$ such that}:
\beq\label{F-sp-tPepsB}
\sigma(\widetilde{P}_\B^\epsilon)\subset(-C\epsilon,C\epsilon)\bigcup(1-C\epsilon,1+C\epsilon).
\eeq
Let us choose a circle $\mathscr{C}\subset\Co$ centred in $1\in\Co$ with a radius $1/2$. Then
\beq \label{dhc22}
Q_\B^\epsilon:= -(2\pi i)^{-1}\oint_{\mathscr{C}}d\zz\,\big(\widetilde{P}_\B^\epsilon-\zz\bb1\big)^{-1}
\eeq
is an orthogonal projection corresponding to the spectral projection of $\widetilde{P}_\B^\epsilon$ on $(1-\epsilon,1+\epsilon)$.

The definition \eqref{dhc23} implies that $\widetilde{P}_\B^\epsilon$ -thus also its resolvent- commutes with all the Zak translations. Due to \eqref{dhc22}, the projection $Q^\epsilon_\B$ has the same commutation property.

By usual holomorphic functional calculus with the bounded self-adjoint operator $\widetilde{P}^\epsilon_\B$, if we consider the function $\zz\mapsto \zz^{-1/2}=e^{-2^{-1}{\rm Ln}(z)}$ that is holomorphic on a disk around $1$ not containing $0$, we can define:
\beq\label{dhc21}
\Theta^\epsilon_\B\,:=\,-(2\pi i)^{-1}\oint_{\mathscr{C}}d\zz\,\zz^{-1/2}\,\big(\widetilde{P}_\B^\epsilon-\zz\bb1\big)^{-1}.
\eeq
 Due to \eqref{dhc21} we deduce that $\Theta^\epsilon_\B$ commutes with $\widetilde{P}^\epsilon_\B$ and we also have the following equalities:
%\beq\label{F-PepsB-dec}\begin{split}
%Q_\B^\epsilon\,&=\cancel{\,P_\B^\epsilon\,P_\B^\epsilon}\\
%&=\,\Big[(2\pi i)^{-1}\oint_{\mathscr{C}}d\zz\,\sqrt{\zz^{-1}}\,\big(\widetilde{P}_\B^\epsilon-\zz\bb1\big)^{-1}\Big]\,\widetilde{P}_\B^\epsilon\,\Big[(2\pi i)^{-1}\oint_{\mathscr{C}}d\zz\,\sqrt{\zz^{-1}}\,\big(\widetilde{P}_\B^\epsilon-\zz\bb1\big)^{-1}\Big].
%\end{split}\eeq
%we can write \eqref{F-PepsB-dec} as:
\beq \label{F-PepsB-dec}\begin{aligned}
Q_\B^\epsilon\,&=\, -(2\pi i)^{-1}\oint_{\mathscr{C}}d\zz\,\big(\widetilde{P}_\B^\epsilon-\zz\bb1\big)^{-1}\,=\, \Theta^\epsilon_\B\,\widetilde{P}_\B^\epsilon\,\Theta^\epsilon_\B\,=\, [\Theta^\epsilon_\B]^2\,\widetilde{P}_\B^\epsilon\,=\, \widetilde{P}_\B^\epsilon\,[\Theta^\epsilon_\B]^2\,\\ 
&=\,\Theta^\epsilon_\B\,\Big[\sum_{\alpha\in \Gamma,p\in \underline{n_\B}} \mathfrak{T}^\epsilon_\alpha \psi_p(x)\, \overline{\mathfrak{T}_\alpha^\epsilon \psi_p(y)}\Big]\,\Theta^\epsilon_\B\,=\,\sum_{\alpha\in \Gamma,p\in \underline{n_\B}} \Big (\mathfrak{T}^\epsilon_\alpha \Theta^\epsilon_\B\,\psi_p(x)\Big )\, \Big (\overline{\mathfrak{T}_\alpha^\epsilon \Theta^\epsilon_\B\psi_p(y)}\Big ),
\end{aligned}
\eeq
where in the last equality we used that $\Theta^\epsilon_\B$ also commutes with the Zak translations. We notice that {$\Theta^\epsilon_\B \, = \Theta^\epsilon_\B\, Q^\epsilon_\B $ and} $\Rg\,Q^\epsilon_\B=\Rg\,\Theta^\epsilon_\B$.
\begin{definition}\label{D-P-fr-eps}
Let us fix the following family of functions in $L^2(\X)$:
$$
\psi^\epsilon_p\,:=\,\Theta^\epsilon_\B\,\psi_p,\quad\forall p\in\underline{n_\B}
$$
and let $\mathcal{L}^\epsilon_\B$ be the Hilbert subspace generated by the family $\big\{\psi^\epsilon_{\alpha,p}:=\mathfrak{T}^\epsilon_\alpha\psi^\epsilon_p,\ \forall(\alpha,p)\in\Gamma\times\underline{n_\B}\big\}$ in $L^2(\X)$.
\end{definition}
By definition, we have the inclusions $\Rg\,Q^\epsilon_\B\subset\mathcal{L}^\epsilon_\B\subset \Rg\,\Theta^\epsilon_\B$ and our remark above  implies the equality $\Rg\,Q^\epsilon_\B\,=\,\mathcal{L}^\epsilon_\B$.
\begin{lemma}\label{L-1}
There exists $\epsilon_0>0$ such that for any $N>0$ there exists $C_N>0$ such that for all $(\alpha,p)\in\Gamma\times\underline{n_\B}$ and $\epsilon\in[0,\epsilon_0]$ we have
$$
\big\|<\cdot-\alpha>^N\Big (\Theta^\epsilon_\B\,\mathfrak{T}^\epsilon_\alpha\psi_p-\mathfrak{T}^\epsilon_\alpha\psi_p\Big )\big\|_{L^2(\X)}\leq C_N\,\epsilon \, \Vert <\cdot-\alpha>^{N+1}\psi_p\Vert_{L^2(\X)}.
$$
\end{lemma}
\begin{proof}
Because $\Theta^\epsilon_\B$ commutes with the Zak translations, it is enough to prove this for $\alpha=0$. Since $P_\B \psi_p=\psi_p$, using \eqref{dhc21} we may write 
\beq\label{dhc20}
\begin{aligned}
\Theta^\epsilon_\B\,\psi_p-\psi_p&=-(2\pi i)^{-1}\oint_{\mathscr{C}}d\zz\,\zz^{-1/2}\,\Big (\big(\widetilde{P}_\B^\epsilon-\zz\bb1\big)^{-1} -\big(P_\B-\zz\bb1\big)^{-1}\Big )\, \psi_p\\
\\&=(2\pi i)^{-1}\oint_{\mathscr{C}}d\zz\,\zz^{-1/2}\,\big(\widetilde{P}_\B^\epsilon-\zz\bb1\big)^{-1} \Big (\widetilde{P}_\B^\epsilon-P_\B\Big )\, \big( P_\B-\zz\bb1\big)^{-1}\, \psi_p.
\end{aligned}
\eeq
From \eqref{F-ker-tPBeps} we see that  
$$\big (\widetilde{P}_\B^\epsilon-P_\B\big )f\, (x)\, =\int_\X dy\, \Big (\big (\Lambda^\epsilon(x,y)-1\big )\,\mathfrak{K}[P_\B](x,y)+\epsilon \Lambda^\epsilon(x,y)\mathfrak{K}_1^\epsilon(x,y)\Big )\, f(y)$$
and assuming that $f$ is strongly localized and the fact that $|\Lambda^\epsilon(x,y)-1|\leq C\, \epsilon\, |x|\, |x-y|$ we have
$$\Vert <\cdot>^N\, \big (\widetilde{P}_\B^\epsilon-P_\B\big )f\Vert_{L^2(\X)}\leq C_N\, \epsilon\, \Vert <\cdot>^{N+1} f\Vert_{L^2(\X)}.$$
Also, denoting by $X$ either $P_\B$ or $\widetilde{P}_\B^\epsilon$ we have that for any  $N\geq 1$ we have 
$$\sup_{\zz\in \mathscr{C}}\Big \Vert <\cdot >^N \big (X-\zz\bb1\big)^{-1}\, <\cdot >^{-N}\Big \Vert =C'_N<\infty,$$
which may be proved by showing that all possible multiple commutators between the position operator and the resolvent can be extended to bounded operators. Introducing all this in \eqref{dhc20} the proof is finished.
\end{proof}

\begin{proposition}
The family $\big\{\psi^\epsilon_{\alpha,p}\ \forall(\alpha,p)\in\Gamma\times\underline{n_\B}\big\}$ is a strongly localized Parseval frame for the closed subspace $\mathcal{L}^\epsilon_\B$ that it generates in $L^2(\X)$.
\end{proposition}
\begin{proof}
From \eqref{F-PepsB-dec} we have that $Q^\epsilon_\B$ is an orthogonal projection with a strongly localized integral kernel:
\beq\nonumber 
\mathfrak{K}[Q^\epsilon_\B](x,y)\,=\,\hspace*{-0.3cm}\underset{(\gamma,q)\in\Gamma\times\underline{n_\B}}{\sum}\hspace*{-0.3cm}\psi^\epsilon_{\gamma,q}(x)\,\overline{\psi^\epsilon_{\gamma,q}}(y),
\eeq
with the series controlled by the rapid decay of the functions $\psi^\epsilon_{\gamma,q}$. Thus, for any $f\in\mathcal{L}^\epsilon_\B$ we notice that:
\beq\nonumber 
f=Q^\epsilon_\B\,f=\hspace*{-0.3cm}\underset{(\gamma,q)\in\Gamma\times\underline{n_\B}}{\sum}\hspace*{-0.3cm}\psi^\epsilon_{\gamma,q}\,\big(\psi^\epsilon_{\gamma,q}\,,\,f\big)_{L^2(\X)}.
\eeq
\end{proof}
\begin{proposition}\label{P-est-tPepsB-QepsB}
There exists $\epsilon_0>0$ and $C>0$ such that:
$$
\big\|\widetilde{P}^\epsilon_\B\,-\,Q^\epsilon_\B\big\|_{\mathbb{B}(L^2(\X))}\,\leq\,C\, \epsilon,\quad\forall\epsilon\in[0,\epsilon_0].
$$
\end{proposition}
\begin{proof}
We notice that $Q^\epsilon_\B$ as defined in \eqref{dhc22} is a spectral projection of the bounded self-adjoint operator $\widetilde{P}^\epsilon_\B$ and taking into account also \eqref{F-sp-tPepsB}, the spectral theorem implies  the conclusion of the proposition.
\end{proof}

\subsubsection{A Parseval frame for $\mathbf{P^{\epsilon,0}_\B\,L^2(\X)}$}\label{SSS-PB-eps0-Wframe}

Let us denote by $P^\epsilon_\B:=P^{\epsilon,0}_\B$ and recall from the last paragraph that $Q^\epsilon_\B$ is the orthogonal projection on the closed subspace $\mathcal{L}^\epsilon_\B\subset L^2(\X)$. From \eqref{DF-p-epsilonc-B} we get, with $H^\epsilon_\bot=\overline{\Op^\epsilon(h_\bot)}$,  the formula:
\beq\nonumber 
P^\epsilon_\B\,=-\,(2\pi i)^{-1}\int_{\mathscr{C}_\bot}\,d\zz\,\big(H^\epsilon_\bot-\zz\bb1\big)^{-1}\,.
\eeq

Using the above formula, the second statement of Proposition \ref{Zak-transl} and the functional calculus, we prove that the orthogonal projection $P^\epsilon_\B$ commutes with the Zak translations $\mathfrak{T}^\epsilon_\gamma$ with $\gamma\in\Gamma$.

\begin{proposition}
	There exists $C>0$ {and $\epsilon_0 >0$ such that, for any $\epsilon \in [0,\epsilon_0]$,}
	$$\big\|P^\epsilon_\B\,-\,Q^\epsilon_\B\big\|_{\mathbb{B}(L^2(\X))}\leq\,C\epsilon\,.$$
\end{proposition}
\begin{proof}
	Proposition \ref{P-1_8-CIP} with $\cc=0$ implies that:
	\beq\nonumber 
	\Big\|P^\epsilon_\B\,-\,\Op^\epsilon\Big(-(2\pi i)^{-1}\int_{\mathscr{C}_\bot}d\z\,(h_\bot-\z)^-_{B^\circ}\Big)\Big\|_{\mathbb{B}(L^2(\X))}\leq C\epsilon.
	\eeq
Due to Proposition \ref{P-band-frame} the integral kernel of the isolated band projection:
	$$
	P_\B\,=\,\Op\Big(-(2\pi i)^{-1}\int_{\mathscr{C}_\bot}d\z\,\big(h_\bot-\z\big)^-_{B^\circ}\Big)
	$$
	has the following expression in the unperturbed Wannier frame of the isolated band:
	$$
	\mathfrak{K}[P_\B](x,y)\,=\,\underset{\alpha\in\Gamma}{\sum}\,\underset{p\in\underline{n_\B}}{\sum}\,\psi_{p}(x-\alpha)\overline{\psi_{p}(y-\alpha)}.
	$$
Let us consider the following magnetic pseudo-differential operator
	\beq\nonumber 
	\check{P}^\epsilon_\B\,:=\,\Op^\epsilon\Big(-(2\pi i)^{-1}\int_{\mathscr{C}_\bot}d\z\,\big(h_\bot-\z\big)^-_{B^\circ}\Big)\,=\,\Op^\epsilon(p_\B)
	\eeq
and notice that \eqref{F-PepsB} implies that there exists some $C>0$ and $\epsilon_0>0$ such that:
	\beq\nonumber 
 \big\|\widetilde{P}^\epsilon_\B\,-\,\check{P}^\epsilon_\B\big\|_{\mathbb{B}(L^2(\X))}\,\leq\,C\,\epsilon,\quad\forall\epsilon\in[0,\epsilon_0].
	\eeq
 These estimates and Proposition \ref{P-est-tPepsB-QepsB} imply the desired conclusion.
\end{proof}

\begin{remark}\label{R-100}
By their definition, both projections $P^\epsilon_\B$ and $Q^\epsilon_\B$ have symbols of class $S^{-\infty}(\Xi)_\Gamma$. 
\end{remark}

Let us denote by
\beq\label{F-P-Q}
\mathfrak{X}^\epsilon:=\epsilon^{-1}\big[P^\epsilon_\B-Q^\epsilon_\B\big]
\eeq
and consider the Sz.~Nagy unitary intertwining operator associated with the pair of orthogonal projections $P^\epsilon_\B$ and $Q^\epsilon_\B$ having the norm of the difference less then 1 (for $\epsilon$ smaller then some $\epsilon_1>0$): 
\beq\begin{aligned}\label{DF-U-eps-B}
 P^\epsilon_\B\, U^\epsilon_\B&= U^\epsilon_\B \,Q^\epsilon_\B,\\
	U^\epsilon_\B:&=\big(\bb1-(P^\epsilon_\B-Q^\epsilon_\B)^2\big)^{-1/2}\Big[P^\epsilon_\B\,Q^\epsilon_\B\,+\,(\bb1-P^\epsilon_\B)(\bb1-Q^\epsilon_\B)\Big]\\
	&=\big(\bb1-\epsilon^2(\mathfrak{X}^\epsilon)^2\big)^{-1/2}\big[\bb1+\epsilon\mathfrak{X}^\epsilon(2Q^\epsilon_\B-\bb1)\big]=\bb1\,+\,\epsilon\tilde{\mathfrak{X}}^\epsilon, 
\end{aligned}\eeq
with $\big\{\tilde{\mathfrak{X}}^\epsilon,\ 0\leq\epsilon\leq\epsilon_0\big\}$ in a bounded subset of $\mathbb{B}\big(L^2(\X)\big)$. 

\begin{lemma}\label{P-UepsB}
There exists $\epsilon_0>0$ such that $U^\epsilon_\B=\Op^\epsilon(\mathfrak{u}^\epsilon_\B)$ with $\mathfrak{u}^\epsilon_\B\in S^0_1(\X\times\X^*)$, remaining in a bounded subset of $S^0_1(\X\times\X^*)$ for $\epsilon\in[0,\epsilon_0]$ and with an integral kernel that has rapid off-diagonal decay.
\end{lemma}
\begin{proof}
We start from the defining Formula \eqref{DF-U-eps-B} and take into consideration that both $P^\epsilon_\B$ and $Q^\epsilon_\B$ belong to $S^{-\infty}(\X\times\X^*)$ and the properties of the magnetic pseudo-differential calculus (composition and inversion) in order to deduce that $U^\epsilon_\B=\Op^\epsilon(\mathfrak{u}^\epsilon_\B)$ has a symbol $\mathfrak{u}^\epsilon_\B$ of class $S^0_1(\X\times\X^*)$; the boundedness with respect to $\epsilon\in[0,\epsilon_0]$ writing  Formula \eqref{DF-U-eps-B} in terms of symbols and using the continuity properties of the magnetic Moyal calculus.

The off-diagonal rapid decay is equivalent to the boundedness of the multiple commutators of the form:
\[
[Q_{j_1},[Q_{j_2},[\dots, Q_{j_n},U^\epsilon_\B]\dots ]]
\]
for any $M\in\Nb$ and $j_k\in\underline{d}$ for all $k\in\underline{M}$ and this follows from the commutation properties of the two operators $P^\epsilon_\B$ and $Q^\epsilon_\B$ that have regular symbols; one has to use the holomorphic functional calculus of Dunford and Formula 6.72 in \cite{IMP-2} for the commutators with the resolvent.
\end{proof}

\begin{lemma}
$U^\epsilon_\B$ commutes with the Zak translations $\big\{\mathfrak{T}^\epsilon_\alpha\big\}_{\alpha\in\Gamma}$.
\end{lemma}
\begin{proof}
	View its explicit Formula \eqref{DF-U-eps-B} and the corresponding commutation properties of $P^\epsilon_\B$ and $Q^\epsilon_\B$ put into evidence above, the Nagy intertwining unitary $U^\epsilon_\B$ also commutes with the Zak translations.
\end{proof}

The previous two lemmas prove the following result:
\begin{proposition}\label{P-const-mf-Pfr}
	The  family $\big\{\mathfrak{T}^\epsilon_\alpha\, U^\epsilon_\B\, \psi^\epsilon_{p},\ \alpha\in\Gamma,\,p\in\underline{n_\B}\big\}$ is a strongly localized Parseval frame for $P^\epsilon_\B\,L^2(\X)$.
\end{proposition}

\subsubsection{The magnetic infinite matrices with a constant weak magnetic field}\label{SSS-const-mf-Pfr}

For a linear operator $T\in\,\mathbb{B}\big(P^\epsilon_\B\,L^2(\X)\big)$ we use the notation 
 \beq\label{DF-cmf-Bmatrix}
\mathfrak{M}^\epsilon_\B[T]\,:=\,\big(\mathfrak{T}^\epsilon_\alpha\,U^\epsilon_\B\,\psi^\epsilon_{p}\,,\,T\,\mathfrak{T}^\epsilon_\beta\,U^\epsilon_\B\,\psi^\epsilon_{q}\big)_{L^2(\X)}
\eeq
for its infinite matrix defined by the Parseval frame in Proposition \ref{P-const-mf-Pfr}. In order to simplify the formulas we shall use the notation
\beq\label{dhc24}
\tpsi^\epsilon_p\,:=\,U^\epsilon_\B\,\psi^\epsilon_{p}.
\eeq

\begin{proposition}\label{P-B2}
Given a magnetic operator $\Op^\epsilon(F)$ with $F\in\mathscr{S}^\prime(\Xi)$ a $\Gamma$-periodic distribution, there exists $\mathring{\mathfrak{M}}^\epsilon_\B[\Op^\epsilon(F)]\in{\cal{s}}(\Gamma;\MmN)$ such that:
\beq\nonumber 
\mathfrak{M}^\epsilon_\B[\Op^\epsilon(F)]_{\alpha,\beta}\,=\,\Lambda^\epsilon(\alpha,\beta)\,\mathring{\mathfrak{M}}^\epsilon_\B[\Op^\epsilon(F)]_{\alpha-\beta}.
\eeq
\end{proposition}
\begin{proof}
Taking into account Proposition \ref{Zak-transl} let us compute:
\begin{align*}
\big[\mathfrak{M}^\epsilon_\B[\Op^\epsilon(F)]_{\alpha,\beta}\big]_{p,q}&=\big(\mathfrak{T}^\epsilon_\alpha\,\tpsi^\epsilon_p\,,\,[\Int\,\Lambda^\epsilon\,\mathfrak{K}[F]]\,\mathfrak{T}^\epsilon_\beta\,\tpsi^\epsilon_q\big)_{L^2(\X)}\\
&=\Lambda^\epsilon(\alpha,\beta)\big(\mathfrak{T}^\epsilon_{\alpha-\beta}\,\tpsi^\epsilon_p\,,\,[\Int\,\Lambda^\epsilon\,\mathfrak{K}[F]]\,\tpsi^\epsilon_q\big)_{L^2(\X)}
\end{align*}
and identify:
\beq\nonumber 
\big [\mathring{\mathfrak{M}}^\epsilon_\B[\Op^\epsilon(F)]_\gamma \big ]_{p,q}\,:=\,\big(\mathfrak{T}^\epsilon_{\gamma}\,\tpsi^\epsilon_p\,,\,[\Int\,\Lambda^\epsilon\,\mathfrak{K}[F]]\,\tpsi^\epsilon_q\big)_{L^2(\X)}.
\eeq
\end{proof}

\begin{definition}\label{D-m-eps-B}
We define the matrix valued $\Gamma$-sequence announced in Subsection \ref{SS-MRes} and appearing in Theorem \ref{T-III}:
$$
\mathfrak{m}^\epsilon_\B\,:=\,\mathring{\mathfrak{M}}[\ham^{\epsilon,0}_\B].
$$
\end{definition}

	We denote by $\vec{\mathscr{M}}^\epsilon_\Gamma[\mathscr{M}_{n_\B}]$ the complex linear space of the matrices in $\mathscr{M}^\circ_\Gamma[\mathscr{M}_{n_\B}]$ satisfying the property in  Proposition \ref{P-B2}.

\begin{proposition}\label{P-matr}
The space $\vec{\mathscr{M}}^\epsilon_{\Gamma}[\MmN]$ is a normed subalgebra of $\mathring{\mathscr{M}}_{\Gamma}[\MmN]$ with involution and unity.
\end{proposition}
\begin{proof}
We reproduce here the proof of (3.28) in \cite{CHN}.
	Given $(S,T)\in\big[\mathscr{M}^\circ_{\Gamma}[\MmN]_\epsilon\big]^2$ we know that $S_{\alpha,\beta}=\Lambda^\epsilon(\alpha,\beta)\mathring{S}_{\alpha-\beta}$ and $T_{\alpha,\beta}=\Lambda^\epsilon(\alpha,\beta)\mathring{T}_{\alpha-\beta}$ with $(\mathring{S},\mathring{T})\in\big[{\cal{s}}\big(\Gamma;\MmN\big)\big]^2$ so that:
	\begin{align*}
		(ST)_{\alpha,\beta}:&=\underset{\gamma\in\Gamma}{\sum}S_{\alpha,\gamma}\,T_{\gamma,\beta}=\underset{\gamma\in\Gamma}{\sum}\Lambda^\epsilon(\alpha,\gamma)\mathring{S}_{\gamma-\alpha}\Lambda^\epsilon(\gamma,\beta)\mathring{T}_{\beta-\gamma}=\underset{\gamma^\prime\in\Gamma}{\sum}\Lambda^\epsilon(\alpha,\gamma^\prime)\Lambda^\epsilon(\gamma^\prime+\alpha,\beta)\mathring{S}_{\gamma^\prime}\mathring{T}_{\beta-\alpha-\gamma^\prime}\\
&=\Lambda^\epsilon(\alpha,\beta)\underset{\gamma\in\Gamma}{\sum}\Lambda^\epsilon(\alpha,\gamma)\Lambda^\epsilon(\gamma,\beta)\mathring{S}_{\gamma}\mathring{T}_{\beta-\alpha-\gamma}=\Lambda^\epsilon(\alpha,\beta)\underset{\gamma\in\Gamma}{\sum}\Lambda^\epsilon(\gamma,\beta-\alpha)\mathring{S}_{\gamma}\mathring{T}_{\beta-\alpha-\gamma}
\end{align*}
and we can define  $\mathring{ST}\in{\cal{s}}\big(\Gamma;\MmN\big)$ by the formula:
	\beq\nonumber 
[\mathring{ST}]_\alpha:=\underset{\gamma\in\Gamma}{\sum}\Lambda^\epsilon(\gamma,\alpha)\,\mathring{S}_{\gamma}\cdot \mathring{T}_{\alpha-\gamma}.
	\eeq
\end{proof}

\begin{proposition}\label{P-4.30}
There exists some $\epsilon_0>0$ and for any $n\in\mathbb{N}$ there exists some $C_n>0$ such that:
$$
<\gamma>^n\big|\big[[\mathfrak{m}^\epsilon_\B]_\gamma-[\mathfrak{m}^\circ_\B]_\gamma\big]_{p,q}\big|\,\leq\,C_n\epsilon,\quad\forall\epsilon\in[0,\epsilon_0].
$$
\end{proposition}
\begin{proof}
Starting from Definition \ref{D-m-eps-B} and \eqref{DF-m-circ-B} we can write:
\[
[\mathfrak{m}^\epsilon_\B]_{\gamma,(p,q)}=\mathring{\mathfrak{M}}[\ham^{\epsilon,0}_\B]_{(\gamma,p),(0,q)}=\big(\mathfrak{T}^\epsilon_{\gamma}\,\tpsi^\epsilon_p\,,\,\ham^{\epsilon,0}_\B\,\tpsi^\epsilon_q\big)_{L^2(\X)}=\big(\mathfrak{T}^\epsilon_{\gamma}\,\tpsi^\epsilon_p\,,\,H^{\epsilon,0}\,\tpsi^\epsilon_q\big)_{L^2(\X)},
\]
and 
\[
[\mathfrak{m}^\circ_\B]_{\gamma,(p,q)}=\mathfrak{M}_\B[H_\B]_{(\gamma,p),(0,q)}=\big(\tau_{-\gamma}\,\psi_p\,,\,H_\B\,\psi_q\big)_{L^2(X)},
\]
so that we have to estimate:
\begin{align*}
&<\gamma>^n\big|\big[[\mathfrak{m}^\epsilon_\B]_\gamma-[\mathfrak{m}^\circ_\B]_\gamma\big]_{p,q}\big|=<\gamma>^n\big|\big(\mathfrak{T}^\epsilon_{\gamma}\,\tpsi_p^\epsilon\,,\,H^{\epsilon,0}\,\tpsi^\epsilon_q\big)_{L^2(\X)}-\big(\tau_{-\gamma}\psi_p\,,\,H_\B\,\psi_q\big)_{L^2(X)}\big|\\ \nonumber
&\hspace*{12pt}=<\gamma>^n\big|\big(U^\epsilon_\B\,\Theta^\epsilon_\B\,\Lambda^\epsilon_{\gamma}\,\tau_{-\gamma}\psi_p\,,\,P^{\epsilon,0}_\B\,H^{\epsilon,0}\,P^{\epsilon,0}_\B\,U^\epsilon_\B\,\Theta^\epsilon_\B\,\psi_q\big)_{L^2(\X)}-\big(\tau_{-\gamma}\psi_p\,,\,H_\B\,\psi_q\big)_{L^2(X)}\big|\\ \nonumber
&\hspace*{12pt}\leq<\gamma>^n\big|\big(U^\epsilon_\B\,\Theta^\epsilon_\B\,\Lambda^\epsilon_{\gamma}\,\tau_{-\gamma}\psi_p\,,\,\Op^{\epsilon,0}(h_\B)\,U^\epsilon_\B\,\Theta^\epsilon_\B\,\psi_q\big)_{L^2(\X)}-\big(\tau_{-\gamma}\psi_p\,,\,\Op^{A^\circ}(h_\B)\,\psi_q\big)_{L^2(X)}\big|+\\
&\hspace*{24pt}+<\gamma>^n\big|\big(U^\epsilon_\B\,\Theta^\epsilon_\B\,\Lambda^\epsilon_{\gamma}\,\tau_{-\gamma}\psi_p\,,\,\Op^{\epsilon,0}(p^\epsilon_\B\sharp^\epsilon\,h\sharp^\epsilon\,p^\epsilon_\B-h_\B)\,U^\epsilon_\B\,\Theta^\epsilon_\B\,\psi_q\big)_{L^2(\X)}\big|.
\end{align*}
First, we notice that Proposition \ref{P-replP3_5}  and Corollary \ref{C-est-p-symb-pert} imply that $p^\epsilon_\B\sharp^\epsilon\,h\sharp^\epsilon\,p^\epsilon_\B-h_\B=\mathscr{O}(\epsilon)$ as symbols of bounded operators and the rapid decay of the functions $\psi_p$ for $p\in\underline{n_\B}$ allows to control the factor $<\gamma>^n$ by the decay of $\tau_{-\gamma}\psi_p$.
We recall \eqref{DF-U-eps-B} and write $U^\epsilon_\B=\bb1+\epsilon\tilde{\mathfrak{X}}^\epsilon$. In order to estimate $\Theta^\epsilon_\B\,\Lambda^\epsilon_{\gamma}\,\tau_{-\gamma}\psi_p$ we use Lemma \ref{L-1}. Thus we obtain the estimate:
\begin{align*}
<\gamma>^n&\big[[\mathfrak{m}^\epsilon_\B]_\gamma-[\mathfrak{m}^\circ_\B]_\gamma\big]_{p,q}=<\gamma>^n\big(\tau_{-\gamma}\psi_p\,,\,\Int\big[\big(\Lambda^\epsilon_{\gamma}\,\Lambda^\epsilon-1\big)\mathfrak{K}[h_\B]\big]\,\psi_q\big)_{L^2(\X)}\,+\,\mathscr{O}(\epsilon).
\end{align*}
In order to finish the proof we notice that:
\begin{align*}
<\gamma>^n&\big(\tau_{-\gamma}\psi_p\,,\,\Int\big[\big(\Lambda^\epsilon_{\gamma}\,\Lambda^\epsilon-1\big)\mathfrak{K}[h_\B]\big]\,\psi_q\big)_{L^2(\X)}=\\
&=<\gamma>^n\int_{\X}dx\int_{\X}dy\,\overline{\psi_p(x-\gamma)}\,\big(i\epsilon\langle\,B^\bullet,(\gamma\wedge x+y\wedge x)\rangle\big)\,\times\\
&\hspace{3.5cm}\times\,\Big(\int_0^1ds\,\exp\big(-is\epsilon\langle\,B^\bullet,x\wedge\gamma+x\wedge y\rangle\big)\Big)\,\mathfrak{K}[h_\B](x,y)\,\psi_q(y).
\end{align*}
As $\gamma\wedge x=(\gamma-x)\wedge x$ and $y\wedge x=y\wedge(x-y)$ we obtain that for any three numbers $M_1, M_2$ and $M_3$ in $\mathbb{N}$ we can write that:
\beq\begin{split}\nonumber
\Big|<\gamma>^n&\big(\tau_{-\gamma}\psi_p\,,\,\Int\big[\big(\Lambda^\epsilon_{\gamma}\,\Lambda^\epsilon-1\big)\mathfrak{K}[h_\B]\big]\,\psi_q\big)_{L^2(\X)}\Big|\leq\,C\,\epsilon\int_{\X}dx\int_{\X}dy\,\times\\
&\hspace*{12pt}\times\,<x-\gamma>^{M_1}|\psi_p(x-\gamma)|\,<x-y>^{M_2}|\mathfrak{K}[h_\B](x,y)|\,<y>^{M_3}|\psi_q(x)|\,\times\\
&\hspace*{12pt}\times\,<\gamma>^n<x-\gamma>^{-M_1}\big(<x-\gamma><x>+<x-y><y>\big)<x-y>^{-M_2}<y>^{-M_3}.
\end{split}\eeq
Finally we notice that $<\gamma>^n\leq\,C_n<\x-\gamma>^n<x>^n$ and $<x>^{n+1}\leq\,C^\prime_{n}<x-y>^{n+1}<y>^{n+1}$ and we only have to choose $M_1\geq n+1$, $M_2\geq n+2$ and $M_3\geq n+2$ in order to obtain a bound of order $\epsilon$ for the above scalar product and finish the proof.
\end{proof}

\subsubsection{The strongly localized Parseval frame for $P^{\epsilon,\cc}_\B\,L^2(\X)$}\label{SSS-nonconst-mf-Pfr}

In order to obtain a localized Parseval frame for the subspace $P^{\epsilon,\cc}_\B\,L^2(\X)$ let us consider the following system of functions in $L^2(\X)$: \beq\label{D-eps-c-fr}
 \big\{\widetilde{\Lambda}^{\epsilon,\cc}_\alpha\,\mathfrak{T}^\epsilon_\alpha\,U^\epsilon_\B\,\psi^\epsilon_{p},\,\alpha\in\Gamma,\,p\in\underline{n_\B}\big\}
 \eeq

We repeat the arguments and constructions in Paragraph \ref{SSS-constmf-Wframe}, starting with an operator $\widetilde{P}^{\epsilon,\cc}_\B$ defined by the regular integral kernel (with rapid off-diagonal decay):
\[
\mathfrak{K}[\widetilde{P}^{\epsilon,\cc}_\B](x,y):=\underset{(\alpha,p)\in\Gamma\times\underline{n_\B}}{\sum}\,\widetilde{\Lambda}^{\epsilon,\cc}_\alpha\,\mathfrak{T}^\epsilon_\alpha\,U^\epsilon_\B\,\psi^\epsilon_{p}(x)\,\overline{\widetilde{\Lambda}^{\epsilon,\cc}_\alpha\,\mathfrak{T}^\epsilon_\alpha\,U^\epsilon_\B\,\psi^\epsilon_{p}(y)}.
\]
We notice that:
\begin{align*}
\widetilde{\Lambda}^{\epsilon,\cc}_\alpha(x)\,\overline{\widetilde{\Lambda}^{\epsilon,\cc}_\alpha(y)}&=\widetilde{\Lambda}^{\epsilon,\cc}(x,\alpha)\widetilde{\Lambda}^{\epsilon,\cc}(\alpha,y)=\widetilde{\Lambda}^{\epsilon,\cc}(x,y)\big[\exp\big(-i\int_{<x,\alpha,y>}B^\epsilon\big)\big]\\
&=\widetilde{\Lambda}^{\epsilon,\cc}(x,y)\Big[1\,-\,i\epsilon\cc\,\Big(\int_{<x,\alpha,y>}B^\epsilon\Big)\Big(\int_0^1ds\,\exp\big(-is\epsilon\cc\int_{<x,\alpha,y>}B^\epsilon\big)\Big)\Big],\\
\Big|\int_{<x,\alpha,y>}B^\epsilon\Big|&\leq\,C<x-\alpha><y-\alpha>,\qquad\forall\epsilon\in[0,\epsilon_0].
\end{align*}
Thus:
\[
\mathfrak{K}[\widetilde{P}^{\epsilon,\cc}_\B]=\widetilde{\Lambda}^{\epsilon,\cc}\,\big(\mathfrak{K}[P^\epsilon_\B]\,+\,\epsilon\cc\mathfrak{K}^{\epsilon,\cc}_2\big)
\]
with $\mathfrak{K}^{\epsilon,\cc}_2$ a $BC^\infty(\X\times\X)$  integral kernel having rapid off-diagonal decay uniformly in $(\cc,\epsilon)\in[0,1]\times[0,\epsilon_0]$ and we may conclude also that $[\widetilde{P}^{\epsilon,\cc}_\B]^2=\widetilde{P}^{\epsilon,\cc}_\B\,+\,\mathscr{O}(\epsilon\cc)$ (as bounded operators on $L^2(\X)$). Thus, for some $\epsilon_0>0$ and any $(\cc,\epsilon)\in[0,1]\times[0,\epsilon_0]$
$$
\sigma(\widetilde{P}^{\epsilon,\cc}_\B)\subset(-\cc\epsilon,\cc\epsilon)\bigcup(1-\cc\epsilon,1+\cc\epsilon)
$$
and we may repeat the arguments given  in Paragraph \ref{SSS-constmf-Wframe} (see \eqref{F-PepsB} - \eqref{F-PepsB-dec}) and define by similar formulas:
\[\begin{split}
&Q^{\epsilon,\cc}_\B\,:=\,-(2\pi i)^{-1}\oint_{\mathscr{C}}d\zz\,\big(\widetilde{P}^{\epsilon,\cc}_\B-\zz\bb1\big)^{-1},\\
&\Theta^{\epsilon,\cc}_\B\,:=\,-(2\pi i)^{-1}\oint_{\mathscr{C}}d\zz\,\sqrt{\zz^{-1}}\,\big(\widetilde{P}^{\epsilon,\cc}_\B-\zz\bb1\big)^{-1},
\end{split}\]
so that we have the following formula:
\[\begin{split}
\mathfrak{K}[Q^{\epsilon,\cc}_\B](x,y)&=\hspace*{-0.3cm}\underset{(\alpha,p)\in\Gamma\times\underline{n_\B}}{\sum}\hspace*{-0.3cm}\Theta^{\epsilon,\cc}_\B\big[\widetilde{\Lambda}^{\epsilon,\cc}_\alpha\,\big(\mathfrak{T}^\epsilon_\alpha\,U^\epsilon_\B\,\psi^\epsilon_{p}\big)\big](x)\,\overline{\Theta^{\epsilon,\cc}_\B\big[\widetilde{\Lambda}^{\epsilon,\cc}_\alpha\,\big(\mathfrak{T}^\epsilon_\alpha\,U^\epsilon_\B\,\psi^\epsilon_{p}\big)\big](y)}.
\end{split}\]
Repeating the arguments  given in Paragraph \ref{SSS-constmf-Wframe} we obtain the following Parseval frame of the Hilbert subspace $Q^{\epsilon,\cc}_\B\,L^2(\X)$:
\beq\label{DF-Pfr-B-epsc}
\tpsi^{\epsilon,\cc}_{\alpha,p}:=\Theta^{\epsilon,\cc}_\B\big[\widetilde{\Lambda}^{\epsilon,\cc}_\alpha\,\big(\mathfrak{T}^\epsilon_\alpha\,U^\epsilon_\B\,\psi^\epsilon_{p}\big)\big],\quad\forall(\alpha,p)\in\Gamma\times\underline{n_\B}.
\eeq

We also recall the following identity, similar to \eqref{F-PepsB-dec}:
\beq\label{F-QThetaP-eps-c}
Q^{\epsilon,\cc}_\B\,=\,\Theta^{\epsilon,\cc}_\B\,P^{\epsilon,\cc}_\B\,\Theta^{\epsilon,\cc}_\B\,=\, [\Theta^{\epsilon,\cc}_\B]^2\,\widetilde{P}_\B^{\epsilon,\cc}\,=\, \widetilde{P}_\B^{\epsilon,\cc}\,[\Theta^{\epsilon,\cc}_\B]^2.
\eeq

\begin{proposition}\label{P-dif-eps-cc-proj} {There exist positive  $C$ and $\epsilon_0$ such that for $(\epsilon,c)\in [0,\epsilon_0]\times [0,1]$, we have: }

	\beq\nonumber 
	\big\|P^{\epsilon,\cc}_\B\,-\,Q^{\epsilon,\cc}_\B\big\|_{\mathbb{B}(L^2(\X))}\leq\,C\,\cc\epsilon.
	\eeq
\end{proposition}
\begin{proof}
	Using \eqref{DF-p-epsilonc-B} together with the statements of Proposition \ref{P-1_8-CIP}, we see that, for all  $(\epsilon,\cc) \in[0,\epsilon_0]\times[0,1]$, 
	
 \begin{align*}
		P^{\epsilon,\cc}_\B=&\Op^{\epsilon,\cc}\Big(-(2\pi i)^{-1}\int_{\mathscr{C}_\bot}d\zz\,(h_\bot-\zz)^-_{\epsilon,0}\Big)\,+\,\cc\epsilon\,\Op^{\epsilon,\cc}\Big(-(2\pi i)^{-1}\int_{\mathscr{C}_\bot}d\zz\,\mathcal{r}^{\epsilon,\cc}_{\zz}(h_\bot)\Big),\\
		&\Big\|\Op^{\epsilon,\cc}\Big((2\pi i)^{-1}\int_{\mathscr{C}_\bot}d\zz\,\mathcal{r}^{\epsilon,\cc}_{\zz}(h_\bot)\Big)\Big\|_{\mathbb{B}(L^2(\X))}\,\leq\,C\,.
	\end{align*}
	Moreover, we know that:
	\beq\nonumber 
	P^\epsilon_\B\,=\,\Op^{\epsilon,0}\Big(-(2\pi i)^{-1}\int_{\mathscr{C}_\bot}d\zz\,(h_\bot-\zz)^-_{\epsilon,0}\Big).
	\eeq
	Thus we have the equality:
	\begin{align*}
		&P^{\epsilon,\cc}_\B\,-\,Q^{\epsilon,\cc}_\B\\
		&\hspace*{12pt}=\Op^{\epsilon,\cc}\Big(-(2\pi i)^{-1}\int_{\mathscr{C}_\bot}d\z\,(h_\bot-\z)^-_{\epsilon,\cc}\Big)-\Op^{\epsilon,\cc}\Big(-(2\pi i)^{-1}\int_{\mathscr{C}_\bot}d\z\,(h_\bot-\z)^-_{\epsilon,0}\Big)\\
    &\hspace*{1cm}+\Int\,\widetilde{\Lambda}^{\epsilon,\cc}\cdot\mathfrak{K}[P^\epsilon_\B]-\Int\,\mathfrak{K}[Q^{\epsilon,\cc}_\B]\\
		&\hspace*{12pt}=\cc\epsilon\,\Big[\Op^{\epsilon,\cc}\Big((2\pi i)^{-1}\int_{\mathscr{C}_\bot}d\zz\,\mathfrak{x}^{\epsilon,\cc}_{\zz}(h_\bot)\Big)-\Int\,\mathfrak{K}^{\epsilon,\cc}_2\Big]\,+\,\widetilde{P}^{\epsilon,\cc}_\B-Q^{\epsilon,\cc}_\B
  \end{align*}
	and the spectral Theorem implies, as in the proof of Proposition \ref{P-est-tPepsB-QepsB}, that:
\[
\big\|\widetilde{P}^{\epsilon,\cc}_\B-Q^{\epsilon,\cc}_\B\big\|_{\mathbb{B}(L^2(\X))}\leq\,\cc\epsilon.
\]
\end{proof}

\begin{corollary}\label{C-W-PQ}
Using the Sz.~Nagy construction we obtain a unitary $W^{\epsilon,\cc}$ {such that $W^{\epsilon,\cc}-\bb1$ is a pseudodifferential operator} with an $A$-symbol of class $S^{-\infty}(\X\times\X^*)$ {with all its seminorms uniformly bounded by $\epsilon \cc$}, such that {$W^{\epsilon,\cc}\,Q^{\epsilon,\cc}=P^{\epsilon,\cc} \, W^{\epsilon,\cc}$}.
\end{corollary}
\begin{proof}
    The proof is analogue to the one of Lemma \ref{P-UepsB}, where one starts from the analogue of \eqref{DF-U-eps-B} when we insert the projections with $\cc\neq 0$. 
\end{proof}

We are now ready to construct the localized Parseval frame when $\cc\neq 0$. We start from the family $\big\{\tpsi^{\epsilon,\cc}_{\alpha,p}\big\}_{(\alpha,p)\in\Gamma\times\underline{n_\B}}$ in \eqref{DF-Pfr-B-epsc}
that is a Parseval frame for the subspace $\mathcal{L}^{\epsilon,\cc}_\B=Q_\B^{\epsilon,\cc}L^2(\X)$ with $\big\|P^{\epsilon,\cc}_\B\,-\,Q^{\epsilon,\cc}_\B\big\|_{\mathbb{B}(L^2(\X))}\leq\,C\,\cc\epsilon$. 
With $W^{\epsilon,\cc}$ being the unitary from Corollary \ref{C-W-PQ}, the family: 
\beq\label{dhc25}
\psi^{\epsilon,\cc}_{\alpha,p}\,:=\,{W^{\epsilon,\cc}}\, \tpsi^{\epsilon,\cc}_{\alpha,p},\quad\forall(\alpha,p)\in\Gamma\times\underline{n_\B}
\eeq
will define a {strongly localized} Parseval frame for the subspace $P^{\epsilon,\cc}_\B\,L^2(\X)$ to which we can associate 
the \textit{coordinate map}
\beq\label{DF-C-eps-c-B}
\mathfrak{C}^{\epsilon,\cc}_\B:\,P^{\epsilon,\cc}_\B\,L^2(\X)\ni\,f\,\mapsto\,\big\{\big(\psi^{\epsilon,\cc}_{\alpha,p}\,,\,f\big)_{L^2(\X)}\}_{(\alpha,p)\in\Gamma\times\underline{n_\B}}\in\mathcal{K}_\B.
\eeq
\subsubsection{The effective magnetic matrix when $\cc\neq 0$}\label{SS-EndProof-T-III}
%We have the magnetic band Hamiltonian $\ham^{\epsilon,\cc}_\B=P^{\epsilon,\cc}_\B\,H^{\epsilon,\cc}\,P^{\epsilon,\cc}_\B$ acting in he Hilbert space $P^{\epsilon,\cc}_\B\,L^2(\X)$. 
Given any bounded linear operator $T\in\mathbb{B}\big(P^{\epsilon,\cc}_\B\,L^2(\X)\big)$ one may define its \textit{magnetic $\B$-matrix} as the matrix $\mathfrak{M}^{\epsilon,\cc}_\B[T]\in\mathscr{M}_\Gamma[\MmN]$ given by:
	\[
\big[\mathfrak{M}^{\epsilon,\cc}_\B[T]_{\alpha,\beta}\big]_{p,q}\,:=\,\big(\psi^{\epsilon,\cc}_{\alpha,p}\,,\,T\,\psi^{\epsilon,\cc}_{\beta,q}\big)_{L^2(\X)}.
	\]
 %Let us consider the magnetic band Hamiltonian $\ham^{\epsilon,\cc}_\B=P^{\epsilon,\cc}_\B\,H^{\epsilon,\cc}\,P^{\epsilon,\cc}_\B$ and notice that:
%\beq
%\big[\mathfrak{M}^{\epsilon,\cc}_\B[\ham^{\epsilon,\cc}_\B]_{\alpha,\beta}\big]_{p,q}=\big(\psi^{\epsilon,\cc}_{\alpha,p}\,,\,\ham^{\epsilon,\cc}_\B\,\psi^{\epsilon,\cc}_{\beta,q}\big)_{L^2(\X)}=\big(\tpsi^{\epsilon,\cc}_{\alpha,p}\,,\,W^{\epsilon,\cc}\ham^{\epsilon,\cc}_\B[W^{\epsilon,\cc}]^{-1}\,\tpsi^{\epsilon,\cc}_{\beta,q}\big)_{L^2(\X)}.
%\eeq
We can write the following equality:
\begin{align*}
\big[\mathfrak{M}^{\epsilon,\cc}_\B[\ham^{\epsilon,\cc}_\B]_{\alpha,\beta}\big]_{p,q}&=\,\Big({W^{\epsilon,\cc}\, }\Theta^{\epsilon,\cc}_\B\big[\widetilde{\Lambda}^{\epsilon,\cc}_\alpha\,\big(\mathfrak{T}^\epsilon_\alpha\,U^\epsilon_\B\,\psi^\epsilon_{p}\big)\big]\,,\,H^{\epsilon,\cc}\,{W^{\epsilon,\cc}\, }\Theta^{\epsilon,\cc}_\B\big[\widetilde{\Lambda}^{\epsilon,\cc}_\beta\,\big(\mathfrak{T}^\epsilon_\beta\,U^\epsilon_\B\,\psi^\epsilon_{q}\big)\big]\Big)_{L^2(\X)}.
\end{align*}

\begin{proposition}\label{Prop-HC}
There exists some $\epsilon_0>0$ and for any $n\in\mathbb{N}$ there exists some $C_n>0$ such that:
\beq \label{F-m-matrix}
<\gamma>^n\Big|\big[\mathfrak{M}^{\epsilon,\cc}_\B[\ham^{\epsilon,\cc}_\B]_{\alpha,\beta}-\Lambda^{\epsilon,\cc}(\alpha,\beta)\, [\mathfrak{m}^\epsilon_\B]_{\alpha-\beta}\big]_{p,q}\Big|\,\leq\,C_n\epsilon\, \cc,\quad\forall (\epsilon,\cc)\in \in[0,\epsilon_0]\times [0,1].
\eeq
\end{proposition}

\begin{proof}
    The proof is analogue to the one of Proposition \ref{P-4.30}. Here we need to replace the symbols having $\cc\neq 0$ with those with $\cc=0$, and all the seminorms of the  errors will be of order $\epsilon \, \cc$. We note that when $\cc\neq 0$, the magnetic translations no longer form a projective representation of the group of translations and we have to work with two indices $\alpha$ and $\beta$. 
\end{proof}

We conclude that $\sigma\big(\ham^{\epsilon,\cc}_\B\big)$ as operator in $P^{\epsilon,\cc}_\B\,L^2(\X)$ is equal to $\sigma\big(\mathfrak{M}^{\epsilon,\cc}[\ham^{\epsilon,\cc}_\B]\big)\setminus\{0\}$ and using  Theorem \ref{T-I} we obtain the proof of the first statement in Theorem \ref{T-III}.

\section{The effective time evolution}\label{S-ev}

The proof of the second statement in Theorem \ref{T-III} will follow from the following theorem that we prove in this section.
\begin{theorem}\label{C-T-I} 
Under the hypotheses of Theorem \ref{T-I}, let us denote by $E_K( H^{\epsilon,\cc})$ the spectral projection on $K\subset \R$ for the self-adjoint operator $H^{\epsilon,\cc}$. Then for any  compact interval $J\subset J^\delta_\B$ (with $J^\delta_\B$ as in the Theorem \ref{T-I}) there exist $C>0$ {and $\epsilon_0 >0$ } such that for any $(\epsilon,c)\in [0,\epsilon_0]\times [0,1]$ and for all $v$ in the range of $E_J( H^{\epsilon,\cc})$, we have the estimation: 
\beq\label{F-ev-est}
\big\|e^{-itH^{\epsilon,\cc}}v - e^{-it\ham^{\epsilon,\cc}_{\B}}v\big\|_{L^2(\X)}\,\leq\,C\,\Big[\epsilon\,+\,{(1+|t|) ^3}\,\epsilon^2\Big]\,\|v\|_{L^2(\X)},\ \forall t\in\R.
\eeq
\end{theorem}
\begin{proof}
	In order to prove \eqref{F-ev-est}, we consider states $v$ with energies in a compact interval $J\subset J_\B$, and we fix two cut-off functions $\varphi$ and $\widetilde{\varphi}$ in $C^\infty_0(\mathbb{R})$ that are equal to $1$ on $J$, have their support included in $J^\delta_\B$ and verify the equality $\varphi=\widetilde{\varphi}\,\varphi$. 
	For any $v\in E^{\epsilon,\cc}_h(J)\mathcal{H}$ we may write:
	$$
	e^{-itH^{\epsilon,\cc}}v=E^{\epsilon,\cc}_h(J) e^{-itH^{\epsilon,\cc}}{\varphi}\big(H^{\epsilon,\cc}\big)E^{\epsilon,\cc}_h(J)v=\widetilde{\varphi}(H^{\epsilon,\cc})e^{-itH^{\epsilon,\cc}}{\varphi}\big(H^{\epsilon,\cc}\big)\widetilde{\varphi}(H^{\epsilon,\cc})v.
	$$
	For any $t\in\mathbb{R}$ let us define $\varphi_t\in C^\infty_0(\R)$ by $\varphi_t(s):=e^{-its}\varphi(s)$, so that the above equality becomes:
	\beq\label{F-t-evol}
	e^{-itH^{\epsilon,\cc}}v=\widetilde{\varphi}(H^{\epsilon,\cc}) \varphi_t\big(H^{\epsilon,\cc}\big)\widetilde{\varphi}(H^{\epsilon,\cc})v.
	\eeq

	We come again to the use of the Helffer-Sj\"{o}strand formula (see \cite{HS,D-95}), {as in the proof of Proposition \ref{R-p-symb}}. 
 We fix an auxiliary cut-off function $\chi\in C^\infty_0(\R;[0,1])$ having support in $\{|t|\leq 2\}$ and being equal to 1 on $\{|t|\leq1\}$ and define (notice a slight difference with \cite{D-95} and \eqref{hcd1} in the choice of such an extension):
	\beq\nonumber 
\overset{\frown}{\varphi}_{t,N}:\Co\rightarrow\R,\quad\overset{\frown}{\varphi}_{t,N}(x+iy):=\underset{0\leq k\leq N}{\sum}(\partial^k\varphi_t)(x)(iy)^k(k!)^{-1}\chi(y).
	\eeq
	Then for any $N\in\mathbb{N}$, the support of $\overset{\frown}{\varphi}_{t,N}$ is a compact set contained in $\supp\varphi_t\times[-2,2]$ and $\overset{\frown}{\varphi}_{t,N}$ is smooth. Moreover:
	\beq\label{dhc26}\begin{split}
		\frac{\partial\overset{\frown}{\varphi}_{t,N}}{\partial\overline{\z}}(\z)&=\frac{1}{2}\left(\frac{\partial\overset{\frown}{\varphi}_{t,N}}{\partial x}+i\frac{\partial\overset{\frown}{\varphi}_{t,N}}{\partial y}\right)\\
		&\hspace*{-2cm}=\frac{1}{2}\Big[i\underset{0\leq k\leq N}{\sum}(\partial^k\varphi_t)(x)(iy)^k(k!)^{-1}[(\partial\chi)(y)]+(\partial^{N+1}\varphi_t)(x)(iy)^N(N!)^{-1}\chi(y)\Big],
	\end{split}\eeq
and we see that for any $x\in\R$:
$$
\underset{y\rightarrow0}{\lim}\,\, \, |y|^{-N}\, \big | (\partial_{\overline{\zz}}\overset{\frown}{\varphi}_{t,N})(x+iy)\big | = \big |(\partial^{N+1}\varphi_t)(x)(N!)^{-1}\big |<\infty.
$$

 We notice the important fact that $(\partial^k\varphi_t)(x)$ is a polynomial of degree $k$ in $t\in\R$ whose coefficients are smooth complex functions of $x\in\R$ having all their support contained in $\supp(\varphi)$. As functions of $t$, these terms can grow at most like $<t>^N$.
 
Let us use \eqref{F-t-evol} and Proposition \ref{P-est-Jepsilon} (applied with $\varphi$ replaced by $\widetilde{\varphi}$), in order to get:
	\begin{equation}\label{jan1}
		\begin{array}{ll}
			e^{-itH^{\epsilon,\cc}}v&=\widetilde{\varphi}(H^{\epsilon,\cc}) \varphi_t\big(H^{\epsilon,\cc}\big)\widetilde{\varphi}(H^{\epsilon,\cc})v=P^{\epsilon,\cc}_\B\,\widetilde{\varphi}(H^{\epsilon,\cc}) \varphi_t\big(H^{\epsilon,\cc}\big)\widetilde{\varphi}(H^{\epsilon,\cc})\,P^{\epsilon,\cc}_\B\,v\,+\,\epsilon \,X_{\epsilon,\cc}v
		\end{array}
	\end{equation}
with $\|X_{\epsilon,\cc}\|_{\mathbb{B}(L^2(\X))}\leq 1$ uniformly for $(\epsilon,\cc)\in[0,\epsilon_0]\times[0,1]$.

 Using Theorem \ref{T-I} and Remark \ref{R-ext-T-I} we have the identity $P^{\epsilon,\cc}_{\B}R^{\epsilon,\cc}(\zz)P^{\epsilon,\cc}_{\B}=P^{\epsilon,\cc}_{\B}R^\sim_{\epsilon,\cc}(\zz)P^{\epsilon,\cc}_{\B}$ and the estimate:
	\begin{align*}
		&\Big\|P^{\epsilon,\cc}_{\B}\Big[R^\sim_{\epsilon,\cc}(\zz)\,-\, \big(\ham^{\epsilon,\cc}_{\B}-\zz P^{\epsilon,\cc}_{\B}\big)^{-1}\Big]\,P^{\epsilon,\cc}_\B \Big\|_{L^2(\X)}\\
		&\leq\Big\|P^{\epsilon,\cc}_\B H^{\epsilon,\cc}[R^\bot_{\epsilon,\cc}(\zz)]H^{\epsilon,\cc}P^{\epsilon,\cc}_\B\Big\|_{L^2(\X)}\leq C(\delta)\,(\Im\hspace*{-1pt}{\cal{m}}\zz)^{-2}\, \epsilon^2\,\big\|R^\bot_{\epsilon,\cc}(\zz)\big\|_{\mathbb{B}(L^2(\X))},
	\end{align*}
    which leads to:
\begin{align*}
		&\big\|P^{\epsilon,\cc}_\B {\varphi}\big(H^{\epsilon,\cc}\big)P^{\epsilon,\cc}_\B v\,-\, \varphi\big(\ham^{\epsilon,\cc}_{\B}\big)P^{\epsilon,\cc}_\B v\big\|_{L^2(\X)}\\
		&\leq \iint\big(\frac{d\zz d\bar \zz}{2\pi}  \big)\,\big|\big(\partial_{\overline{\zz}}\overset{\frown}{\varphi}_{N}\big)(\zz,\overline{\zz})\big|\,\Big\|P^{\epsilon,\cc}_{\B}\Big[R^{\epsilon,\cc}(\zz)\,-\, \big(\ham^{\epsilon,\cc}_{\B}-\zz P^{\epsilon,\cc}_{\B}\big)^{-1}\Big]\,P^{\epsilon,\cc}_\B v\Big\|_{L^2(\X)}=\mathscr{O}(\epsilon^2).
	\end{align*}

 Using the above estimate we also have
\begin{align*}
e^{-it\ham^{\epsilon,\cc}_{\B}}v&=e^{-it\ham^{\epsilon,\cc}_{\B}}\varphi(H^{\epsilon,\cc})\,v=e^{-it\ham^{\epsilon,\cc}_{\B}}P^{\epsilon,\cc}_\B\,\varphi(H^{\epsilon,\cc})\,P^{\epsilon,\cc}_\B\,v\,+\,\mathscr{O}(\epsilon)\,v\\
&=e^{-it\ham^{\epsilon,\cc}_{\B}}\varphi\big(\ham^{\epsilon,\cc}_{\B}\big)P^{\epsilon,\cc}_\B\,v\,+\,\mathscr{O}(\epsilon)\,v=\varphi_t\big(\ham^{\epsilon,\cc}_{\B}\big)P^{\epsilon,\cc}_\B\,v\,+\,\mathscr{O}(\epsilon)\,v.
\end{align*}

	The above identity and \eqref{jan1} lead us to:
	\begin{align*}
		\big\|e^{-itH^{\epsilon,\cc}}v\,-\,e^{-it\ham^{\epsilon,\cc}_{\B}}v\big\|_{L^2(\X)}\,&{=}\,\big\|P^{\epsilon,\cc}_\B {\varphi}_t\big(H^{\epsilon,\cc}\big)P^{\epsilon,\cc}_\B v\,-\,P^{\epsilon,\cc}_\B \varphi_t\big(\ham^{\epsilon,\cc}_{\B}\big)P^{\epsilon,\cc}_\B v\big\|_{L^2(\X)}\,+\,\mathscr{O}(\epsilon)\|v\|_{L^2(\X)}\\
		&\hspace*{-5cm}\leq \iint\big(\frac{d\zz d\bar \zz}{2\pi}  \big)\,\big|\big(\partial_{\overline{\zz}}\overset{\frown}{\varphi}_{t,N}\big)(\zz,\overline{\zz})\big|\,\Big\|P^{\epsilon,\cc}_{\B}\Big[R^{\epsilon,\cc}(\zz)\,-\, \big(\ham^{\epsilon,\cc}_{\B}-\zz P^{\epsilon,\cc}_{\B}\big)^{-1}\Big]\,P^{\epsilon,\cc}_\B v\Big\|_{L^2(\X)}+\,\mathscr{O}(\epsilon)\|v\|_{L^2(\X)}.
	\end{align*}

	Finally, by taking $N=2$ in the definition of $\overset{\frown}{\varphi}_{t,N}$ and using \eqref{dhc26}, we have the bound:
	\begin{align*}
		&\iint\,\big(-\frac i2  d\zz d\bar \zz\big)\,\big|\big(\partial_{\overline{\zz}}\overset{\frown}{\varphi}_{t,2}\big)(\zz,\overline{\zz})\big|\,|\Im\hspace*{-1pt}{\cal{m}}\zz|^{-2}=\iint_{\supp(\widetilde{\varphi}_{t,2})}\,dx\,dy\,\big|\big(\partial_{\overline{\zz}}\overset{\frown}{\varphi}_{t,2}\big)(x+iy)\big|\,|y|^{-2}\\
		&\hspace*{0.5cm}\leq\int_{\supp(\varphi)}\,dx\,\left[\underset{0\leq k\leq 2}{\sum}\big[\underset{x\in\R}{\sup}\big|\big(\partial^{k}\varphi_t\big)(x)\big|\big]\int_{1}^2\,dy\,y^{k-2}\right]\\
		&\hspace*{0.5cm}+\int_{\supp(\varphi)}\,dx\,\big[\underset{x\in\R}{\sup}\big|\big(\partial^{3}\varphi_t\big)(x)\big|\big]\int_{0}^2\,dy.
	\end{align*}
	
	One concludes that there exist  {$C>0$ and $\epsilon_0 >0$ } such that for any $t\in\R$ and for any $(\epsilon,c)\in[0,\epsilon_0]\times[0,1]$ we have the estimate:
	\begin{align*}
		\big\|e^{-itH^{\epsilon,\cc}}v\,-\,e^{-it\ham^{\epsilon,\cc}_{\B}}v\big\|_{L^2(\X)}\,&\leq\\
		&\hspace*{-4cm}\leq\,\Big[C\,\epsilon^2\, <t>^3 \big(\hspace*{-0.3cm}\underset{ {\scriptsize \begin{array}{c} \Re\hspace*{-1pt}\mathcal{e}\zz\in \supp(\varphi),\\|\Im\hspace*{-1pt}{\cal{m}}\zz|\leq2\end{array}}}{\sup}\hspace*{-0.3cm}\big\|[R^\bot_{\epsilon,\cc}(\zz)]\big\|_{\mathbb{B}(L^2(\X))}\big)\,+\,\mathscr{O}(\epsilon)\Big]\,\|v\|_{L^2(\X)}\, ,\quad\forall v\in E^{\epsilon,\cc}_h(J)\mathcal{H}.
	\end{align*}
 \end{proof}

Let  $\tilde{v}:=\mathfrak{C}^{\epsilon,\cc}_\B\,P_\B^{\epsilon,\cc} \, v\in\, \mathcal{K}_\B$ so that \eqref{F-ev-est} is equivalent to the following estimate on matrices:
\[
\big\|e^{-itH^{\epsilon,\cc}}v - [\mathfrak{C}^{\epsilon,\cc}_\B]^*e^{-it\ham^{\epsilon,\cc}_{\B}}\tilde{v}\big\|_{L^2(\X)}\leq\,C\big[\epsilon\,+\,(1+|t|) ^3\epsilon^2\big]\|v\|_{L^2(\X)}.
\]

But, as $\ham^{\epsilon,\cc}_{\B}$ leaves invariant the subspace $P^{\epsilon,\cc}_\B\,L^2(\X)$, we deduce the equality
\[
\mathfrak{M}^{\epsilon,\cc}[e^{-it\ham^{\epsilon,\cc}_{\B}}]\,=\,e^{-it\mathfrak{M}^{\epsilon,\cc}[\ham^{\epsilon,\cc}_{\B}]}
\]
From the results of Paragraph \ref{SS-EndProof-T-III} we know that:
\[
\mathfrak{M}^{\epsilon,\cc}[\ham^{\epsilon,\cc}_{\B}]_{\alpha,\beta}=\Lambda^{\epsilon,\cc}(\alpha,\beta)[\mathfrak{m}^\epsilon_\B]_{\alpha-\beta}+\mathscr{O}(\cc\epsilon)
\]
and thus, recalling \eqref{DF-Op-m-eps} we also have an estimate:
\[
\mathfrak{M}^{\epsilon,\cc}[e^{-it\ham^{\epsilon,\cc}_{\B}}]=e^{-it\widetilde{\Op}^{\epsilon,\cc}(\mathfrak{m}^\epsilon_\B)}\,+\,\mathscr{O}(<t>\cc\epsilon).
\]
Finally:
\[
\big\|e^{-itH^{\epsilon,\cc}}v - [\mathfrak{C}^{\epsilon,\cc}_\B]^*e^{-it\widetilde{\Op}^{\epsilon,\cc}(\mathfrak{m}^\epsilon_\B)}\tilde{v}\big\|_{L^2(\X)}\leq\,C\big[\epsilon\,+\,(1+|t|) ^4\epsilon(\cc+\epsilon)\big]\|v\|_{L^2(\X)},
\]
so that we finish the proof of the second point \eqref{F-II} of  Theorem \ref{T-III}.

\appendix
\section{Some properties of the magnetic pseudo-differential calculus}\label{A-m-PsiDO}

The results in \cite{MP-1} imply that the magnetic Moyal product is a  jointly continuous bilinear map $\mathscr{S}(\Xi)\times\mathscr{S}(\Xi)\rightarrow\mathscr{S}(\Xi)$ that has the following explicit formula:
\beq\nonumber 
\big(\Phi\sharp^B\Psi\big)(X)\,=\,4^{-d}\int_{\Xi}dZ\int_{\Xi}dZ^\prime\,e^{-2i(<\xi-\zeta,x-z^\prime>-<\xi-\zeta^\prime,x-z>)}\,e^{-i\Theta^B(x;z,z^\prime)}\,\Phi(Z)\,\Psi(Z^\prime)
\eeq
where $\Theta^B(x;z,z^\prime)$ is the flux of the $2$-form $B$ through the following triangle in $\X$:
\beq\nonumber 
\mathcal{T}_{x;z,z^\prime}:=\big\{P_{x;z,z^\prime}(s,s^\prime):=z+z^\prime-x+2s(x-z)+2s^\prime(z-z^\prime)\,,\,s\in[0,1],\,s^\prime\in[0,s]\big\}.
\eeq

The folllowing theorem is proved in  \cite{IMP-3}(Theorem 2.1)(see also Theorem 2.6 in \cite{IMP-1}):
\begin{theorem}
For $B\in\Fb^2(\X)$ the Moyal product extends to a bilinear, continuous map:
	$
	S^{p_1}_\rho(\Xi)\times S^{p_2}_\rho(\Xi)\ni(F,G)\,\mapsto\,F\sharp^BG\in S^{p_1+p_2}_\rho(\Xi)$. If moreover the magnetic field $B$ has $\Gamma$-periodic components, then the magnetic Moyal product restricted to $\Gamma$-periodic symbols, defines a bilinear continuous map :
	\beq
	S^{p_1}_\rho(\Xi)_\Gamma\times S^{p_2}_\rho(\Xi)_\Gamma\ni(F,G)\,\mapsto\,F\sharp^BG\in S^{p_1+p_2}_\rho(\Xi)_\Gamma.
	\eeq
	\end{theorem}
The next theorem is proved in \cite{IMP-3} (Theorem 2.6)  (see also Theorem 3.1 in \cite{IMP-1}):
\begin{theorem}
 Given $\rho \in (0,1)$, a magnetic field $B\in\Fb^2(\X)$ and an associated vector potential of class $\Fp^1(\X)$, if $f\in S^0_\rho(\Xi)$,  then $\Op^A(f)\in\mathbb{B}\big(L^2(\X)\big)$.
 \end{theorem}
We also need the following results proven in \cite{IMP-3} (Theorem 2.7 and Proposition 2.4):
\begin{theorem}
 For $F\in S^p_1(\Xi)$ positive and elliptic, with $p>0$, the operator $\Op^A(F):\mathscr{S}(\X)\rightarrow L^2(\X)$ is essentially self-adjoint and its closure has the domain:
	\beq
	\mathscr{H}^p(\X)\,:=\,\big\{f\in L^2(\X)\,,\,\Op(\mathfrak{m}_p)f\in L^2(\X)\big\}.
	\eeq
	Moreover, for any $\zz\in\mathbb{C}\setminus\sigma\big(\overline{\Op^A(F)}\big)$ there exists $\mathfrak{r}^B_\zz\in S^{-p}_1(\Xi)$ such that: 
	\beq\label{F-simb-rez}
	\big(\overline{\Op^A(F)}-\zz\bb1\big)^{-1}=\Op^A(\mathfrak{r}^B_\zz).
	\eeq
\end{theorem}

The explicit formula of $\Op^A(\Phi)\in\mathcal{L}\big(\mathscr{S}(\X);\mathscr{S}(\X)\big)$ in Definition \ref{D-OpA} allows us to obtain a simple formula for the integral kernel $\mathfrak{K}^A[\Phi]\in\mathscr{S}(\X\times\X)$ of this operator:
\beq\label{F-KerOpA}
\mathfrak{K}^A[\Phi]\,=\,\Lambda^A\,\big[\big((\bb1\otimes\mathcal{F}_{\X^*})\Phi\big)\circ\Upsilon\big],\quad\Upsilon:\X\times\X\ni(x,y)\mapsto\big((x+y)/2,x-y\big)\in\X\times\X.
\eeq

Using Propositions 1.3.3 and 1.3.6 in \cite{ABG} and Lemma A.4 in \cite{MPR1}, one obtains easily the following result:
\begin{proposition}\label{P-ker-OpA} Suppose we have a magnetic field $B\in\Fb^2(\X)$ with a vector potential $A\in\Fp^1(\X)$.
	\begin{enumerate}
		\item If $F\in S^p_1(\Xi)$ for any $(p,\rho)\in\R\times\{0,1\}$, then $\mathfrak{K}^A[F]\in\mathscr{S}^\prime(\X\times\X)$ as given in \eqref{F-KerOpA} is a smooth function on $\X\times\X\setminus\big\{(x,x)\in\X\times\X,\,x\in\X\big\}$ having rapid decay in the variable $x-y\in\X$.
		\item If $F\in S^p_1(\Xi)$ with $p<0$, then $\mathfrak{K}^A[F]\in L^1(\X\times\X)$.
		\item $F\in S^{-\infty}(\Xi)$ if and only if $\mathfrak{K}^A[F]\in C^\infty(\X\times\X)$ with rapid decay in the directions orthogonal to the diagonal of $\X\times\X$, together with all its derivatives.
	\end{enumerate}
\end{proposition}

\begin{notation}\label{N-Scirc}
We denote by $\mathring{\mathscr{S}}(\X\times\X)$ the space of tempered distributions on $\X\times\X$ that are smooth and with rapid decay outside the diagonal $\Delta_\X:=\{(x,x)\in\X\times\X,\ x\in\X\}$ and by $\mathring{\mathscr{S}}(\X\times\X)_\Gamma$ the subspace of $\Gamma$-periodic distribution kernels in $\mathring{\mathscr{S}}(\X\times\X)$, i.e. $\mathfrak{K}\in\mathring{\mathscr{S}}(\X\times\X)$ such that $(\tau_\gamma\otimes\tau_\gamma)\mathfrak{K}=\mathfrak{K}$ for any $\gamma\in\Gamma$.
\end{notation}

We use the following phase function:
\beq\label{DF-Omega}
\Omega^B(x,y,z):=\exp\Big(-i\int_{<x,y,z>}\hspace*{-4pt}B\hspace*{2pt}\Big)
\eeq
that due to Stokes Formula satisfies the following equality:
\beq\label{F-Omega-Stokes}
\Omega^B(x,y,z)=\Lambda^A(x,y)\,\Lambda^A(y,z)\,\Lambda^A(z,x).
\eeq

\section{Some basic properties of frames on  Hilbert spaces}\label{Ap-B}
In the whole section,  $\H$ denotes  a complex Hilbert space.

\begin{definition}\label{D-frame}~
 	\begin{enumerate}
		\item We call  frame of $\H$, an at most countable family  of vectors $\big\{\psi_p\,,\,p\in\mathbb{N}_\bullet\big\}\subset\H$ such that there exists two positive constants $0<A\leq B$ satisfying:
		\beq\label{DF-fr}
		A\big\|f\|_{\H}^2\,\leq\,\underset{M\nearrow\infty}{\lim}\underset{1\leq p\leq M}{\sum}\big|\big(\psi_p\,,\,f\big)_{\H}\big|^2\,\leq\,B\big\|f\|_{\H}^2,\quad\forall f\in\H.
		\eeq
		\item A  frame  of $\H$  is called a Parseval frame, when the following equality is verified for any $f\in\H$:
		\beq\label{DF-Pfr}
		\big\|f\big\|^2_{\H}=\underset{p\in\mathbb{N}_\bullet}{\sum}\,\big|\big(\psi_p\,,\,f\big)_{\H}\big|^2\,,
		\eeq
		with the series converging in $\ell^2(\mathbb{N}_\bullet)$.
		\item Given a frame $\big\{\psi_p\,,\,p\in\mathbb{N}_\bullet\big\}\subset\H$ we introduce its  coordinate map $\mathfrak{C}$ as: 
		\beq\label{DF-fr-1}
		\H\ni f\mapsto\big (\mathfrak{C} f ) :=\{\big(\psi_p\,,\,f\big)_{\H}\big\}_{p\in \Nb} \in\ell^2(\Nb).
		\eeq
	\end{enumerate}
\end{definition}
\begin{lemma}\label{L-Parseval}
	Given a Parseval frame in $\H$, the following identities are valid, for $(f,g)\in\H\times\H$:
	\beq\label{F-P-1}
	\big(f\,,\,g\big)_{\H}=\underset{p\in\mathbb{N}_\bullet}{\sum}\,\big(f,\psi_p\big)_{\H}\big(\psi_p,g\big)_{\H},\qquad
	f=\underset{p\in\mathbb{N}_\bullet}{\sum}\,\big(\psi_p\,,\,f\big)_{\H}\,\psi_p\,,
	\eeq
	with the series converging in $\ell^2(\mathbb{N}_\bullet)$ and resp. in the strong topology of $\H$ and the coordinate map $\mathfrak{C}:\H\rightarrow\ell^2(\Nb)$ is an isometry that may not be surjective.
\end{lemma}
\begin{proof}
	Let us write the polarization identity in $\H$:
	\beq\nonumber 
	\big(f\,,\,g\big)_{\H}=(1/4)\underset{0\leq k\leq 3}{\sum}i^k\big\|f+(-i)^kg\big\|^2_{\H}\,,
	\eeq
	and use \eqref{DF-Pfr} to obtain:
	\begin{align*}
		\big(f\,,\,g\big)_{\H}&=(1/4)\underset{0\leq k\leq 3}{\sum}i^k\big\|f+(-i)^kg\big\|^2_{\H}=(1/4)\underset{0\leq k\leq 3}{\sum}i^k\underset{p\in\mathbb{N}_\bullet}{\sum}\,\big|\big(\psi_p\,,\,f+(-i)^kg\big)_{\H}\big|^2\\
		&=\underset{p\in\mathbb{N}_\bullet}{\sum}\,(2/4)\Big(2\Re\big(\psi_p\,,\,f\big)_{\H}\overline{\big(\psi_p\,,\,g\big)}_{\H}-2\Im\big(\psi_p\,,\,f\big)_{\H}\overline{\big(\psi_p\,,\,g\big)}_{\H}\Big)\\
		&=\underset{p\in\mathbb{N}_\bullet}{\sum}\big(\psi_p\,,\,g\big)_{\H}\overline{\big(\psi_p\,,\,f\big)}_{\H}.
	\end{align*}
	For the second identity one proves using the first identity that $\sum\limits_{p=1}^M\,\big(\psi_p\,,\,f\big)_{\H}\psi_p$ is a sequence converging weakly to $f$ and being Cauchy for the strong topology, due to the Parseval condition.
\end{proof}
We also recall the following theorem in \cite{Chris} (Theorem 5.5.1.)
\begin{theorem} A sequence $\{f_k\}_{k=1}^\infty$ in $\H$ is a frame for $\H$ if and only if the application:
		\beq \nonumber 
	\{c_k\}_{k=1}^\infty \mapsto{\text{\tt{\bf F}}}\big(\{c_k\}_{k=1}^\infty\big)\,:=\,\underset{k\in\mathbb{N}_\bullet}{\sum}c_\ell \,f_\ell 
		\eeq
		is a well defined surjective map from $\ell^2(\mathbb{N}_\bullet)$ onto $\H$.
\end{theorem}
\begin{remark}
	If we denote by $\{\mathcal{e}_\gamma\}_{\gamma\in\Gamma}$ the canonical orthonormal basis of $\ell^2(\Nb)$ and use the notation introduced in the begining of section \ref{S-Parseval} we have the identities:
	\beq\nonumber 
	\mathfrak{C}=\underset{p\in\Nb}{\sum}\mathcal{e}_\p\bowtie\psi_p,\qquad\Id_{\H}=\mathfrak{C}^*\mathfrak{C}.
	\eeq
\end{remark}

Let us consider the map:
\beq\label{DF-coord-op-hom}
\mathbb{B}\big(\H\big)\ni\,T\,\mapsto\,\widetilde{\mathfrak{C}}[T]:=\mathfrak{C}\,T\,\mathfrak{C}^*\in\mathbb{B}\big(\ell^2(\Nb)\big)
\eeq
and for any linear bounded operator $T\in\mathbb{B}\big(\H\big)$ let $\mathfrak{M}[T]$ denote the infinite matrix of the operator $\widetilde{\mathfrak{C}}[T]\in\mathbb{B}\big(\ell^2(\Nb)\big)$ with respect to the canonical orthonormal basis $\{\mathcal{e}_p\}_{p\in\Nb}$.

\vspace{0.5cm}

\noindent{\bf Acknowledgements.} HC acknowledges support from Grant DFF–10.46540/2032-00005B of the Independent Research Fund Denmark $|$ Natural Sciences. RP acknowledges support from a grant of the Romanian Ministry of Research, Innovation and Digitization, CNCS -UEFISCDI, project number PN-IV-P1-PCE-2023-0264, within PNCDI IV. BH and RP acknowledge support from the CNRS International Research Network ECO-Math. We also thank our home institutions for hosting our reciprocal visits.


\begin{thebibliography}{99}
	
	\bibitem{ABG} W.O. Amrein, A. Boutet de Monvel, and V. Georgescu: \textit{$C^0$-Groups, Commutator Methods and Spectral Theory of N-Body Hamiltonians}, Birkh\"{a}user Verlag, 1996.

\bibitem{AMP} N. Athmouni, M. M\u{a}ntoiu, and  R. Purice: \textit{On the continuity of spectra for families of magnetic pseudodifferential operators}. Journal of Mathematical Physics 51, 083517 (2010).


\bibitem{AK} D. Auckly and P. Kuchment: \textit{On Parseval frames of exponentially decaying composite Wannier functions}. Contemp. Math. {\bf 717},  227--240 (2018)

\bibitem{Be1} J. Bellissard: {\it $C^*$ algebras in solid state physics. 2D electrons in a uniform magnetic field.}  Operator
algebras and applications {\bf 2}, 49--76, London Math. Soc. Lecture Note Ser., 136, Cambridge
Univ. Press, Cambridge, (1988).



\bibitem{Bo} R. Bott: \textit{Homogeneous Vector Bundles}. Annals of Mathematics \textbf{66}(2),  203-248 (1957)

\bibitem{BC} P. Briet, H.D. Cornean: \textit{Locating the spectrum for magnetic Schr\"odinger and Dirac operators}. Commun. P.D.E. {\bf 27} 1079--1101 (2002)

\bibitem{Chris} O. Christensen: \textit{An Introduction to Frames and Riesz Bases}. In the series Applied and Numerical Harmonic Analysis, Springer International Publishing Switzerland 2016, 

\bibitem{C-99} H.D. Cornean: \textit{On the essential spectrum of two dimensional periodic magnetic Schr\"odinger operators}. Lett. Math. Phys.  \textbf{49}, 197-211 (1999)

\bibitem{C-10} H.D. Cornean: {\it On the Lipschitz Continuity of Spectral Bands of Harper-Like and Magnetic Schr\"odinger Operators}. Ann. Henri Poincar{\' e} {\bf 11}, 973--990 (2010)

\bibitem{CHP-1} H. D. Cornean, B. Helffer, and R. Purice: \textit{Low lying spectral gaps induced by slowly varying magnetic fields}. Journal of Functional Analysis {\bf 273}(1), 206--282 (2017).

\bibitem{CHP-3} H. D. Cornean, B. Helffer, and R. Purice: \textit{A Beals criterion for magnetic pseudo-differential operators proved with magnetic Gabor frames}. Comm. in P.D.E. {\bf 43}(8), 1196--1204 (2018).

\bibitem{CHP-2} H. D. Cornean, B. Helffer, and R. Purice: \textit{Peierls' substitution for low lying spectral energy windows}. Journal of Spectral Theory {\bf 9}(4), 1179--1222 (2019).

\bibitem{CHP-4} H. D. Cornean, B. Helffer, and R. Purice: \textit{Spectral analysis near a Dirac type crossing in a weak non-constant magnetic field}. Transactions of the American Mathematical Society 374 (10),  7041--7104 (2021).



\bibitem{CHN} H. D. Cornean, I. Herbst, and G. Nenciu: {\it On the construction of composite Wannier functions}.  Ann. H. Poincar{\'e} {\bf 17}, 3361--3398 (2016).


\bibitem{CIP} H.D. Cornean, V. Iftimie, and R. Purice: \textit{Peierls substitution via minimal coupling and magnetic pseudo-differential calculus}. {Reviews in Mathematical Physics} {\bf 31}(3), 1950008 (2019).

\bibitem{CM}
H.D.~Cornean and D.~Monaco: {\it On the construction of Wannier functions in topological insulators: the 3D case}. Ann. H. Poincar{\'e} {\bf 18}, 3863--3902 (2017).

\bibitem{CMM}
H.D.~Cornean, D.~Monaco and M.~Moscolari: {\it Parseval Frames of Exponentially Localized Magnetic Wannier Functions}. Commun. Math. Phys. {\bf 371}(3), 1179--1230 (2019)

\bibitem{CN-98} H. D. Cornean, Gh. Nenciu: \textit{On eigenfunction decay of two dimensional magnetic Schr\"odinger operators}. Commun. Math. Phys. \textbf{192}, 671-685 (1998)

\bibitem{CN-00} H. D. Cornean, Gh. Nenciu: \textit{Two dimensional magnetic Schr\"odinger operators: width of mini-bands in the tight-binding approximation}.
Ann. Henri Poincar\'e \textbf{1}, 203-222 (2000)

\bibitem{CP-1} H.D. Cornean and   R. Purice: \textit{On the regularity of the Hausdorff distance between spectra of perturbed magnetic Hamiltonians}. Operator Theory: Advances and Applications 224, 55--66 (2012).

\bibitem{CP-2} H.D.  Cornean and  R. Purice: \textit{Spectral edge regularity of magnetic Hamiltonians}. {Journal of the London Mathematical Society}, {\bf 92}(1), 89--104 (2015).

\bibitem{CH} D. E. Crabtree, E. V. Haynsworth: \textit{An identity for the Schur complement of a matrix}. Proceedings of the American Mathematical Society. \textbf{22} (2), 364--366  (1969).

\bibitem{D-95} E. B. Davies: \textit{The Functional Calculus}, J. London Math. Soc. \textbf{52}, 166--176, (1995).

\bibitem{dNL} G. De Nittis and M. Lein: {\it Applications of magnetic $\Psi$DO techniques to space-adiabatic perturbation
	theory}. Rev. Math. Phys. {\bf 23}(3), 233--260 (2011).

\bibitem{DGR} M. Dimassi, Jean-Claude Guillot, James Ralston: \textit{On effective Hamiltonians for adiabatic perturbations of magnetic Schrödinger operators}, Assymp. Analysis {\bf 40}(2), 137--146 (2004).

\bibitem{Di} J. Dixmier: {\it Les alg\`{e}bres d'op\'{e}rateurs dans l'espace hilbertien}.  (2-\`{e}me edition), Paris, Gauthier-Villars \'{e}diteur, (1969).

\bibitem{FMP} D. Fiorenza, D. Monaco, and G. Panati:  {\it Construction of real-valued localized composite Wannier functions for insulators}. Ann. H. Poincar{\'e} {\bf 17}(1), 63--97 (2016).


\bibitem{FT} S. Freund and S. Teufel: {\it Peierls substitution for magnetic Bloch bands}. Anal. \& PDE {\bf 9}(4),
773--811 (2016).
\bibitem{GMSj} C. G\'erard, A. Martinez and J. Sj\"ostrand:
\newblock{\it  A mathematical approach to the effective Hamiltonian in perturbed periodic problems.}
\newblock  Comm. Math. Phys. 142, 217?244, 1991.



\bibitem{HM} B. Helffer and A. Mohamed: \textit{Asymptotic of the density of states for the Schr\"odinger operator with periodic electric potential}. Duke Math. J. 92(1), (1998), 1-60

\bibitem{HS} B. Helffer and  J. Sj\"ostrand: 
{\it Equation de Schr\"odinger avec
	champ magn{\'e}tique et  \'equation de Harper}, in LNP {\bf 345}, Springer-Verlag, Berlin, Heidelberg and New York,
118--197 (1989).

\bibitem{HS1}  B. Helffer and  J. Sj\"ostrand:
\newblock {\it Analyse  semi-classique pour l'\'equation de Harper (avec application \`a l'\'equation de Schr\"odinger avec champ magn\'etique).}
\newblock M\'emoire de la SMF, No 34 (1988).



\bibitem{HH} R. Hempel and I. Herbst: \textit{Bands and gaps for magnetic periodic Hamiltonians}. In \textit{Partial Differential Operators and Mathematical Physics}, OT 78, Birkh\"{a}user,  175--184 (1995).



\bibitem{H-3} L. H\"{o}rmander:
\textit{The Analysis of Linear Partial Differential Operators III: Pseudo-Differential Operators}. Springer-Verlag Berlin Heidelberg, (2007), pag. VIII + 525.

\bibitem{Hu} D. Husemoller: \textit{Fibre bundles}, 3rd edition. No. 20 in Graduate Texts in Mathematics. Springer-Verlag, New
York, (1994).

\bibitem{IMP-1} V. Iftimie, M. Mantoiu, and R. Purice: \textit{Magnetic pseudodifferential operators}. Publ. Res. Inst. Math. Sci. 43, 585--623 (2007).

\bibitem{IMP-2} V. Iftimie, M. Mantoiu, and R. Purice: \textit{Commutator criteria for magnetic pseudodifferential operators}. Comm. Partial Diff. Eq. 35, 1058--1094  (2010).

\bibitem{IMP-3} Viorel Iftimie, Marius Mantoiu, and Radu Purice: \textit{Quantum observables as magnetic pseudodifferential operators}, Revue Roumaine de Math\'{e}matiques Pures et Appliqu\'{e}es 64 (2-3),  197--223 (2019).


\bibitem{IP} Viorel Iftimie and Radu Purice: \textit{The Peierls–Onsager effective Hamiltonian in a complete gauge covariant setting: determining the spectrum}. J. Spec. Th. {\bf 5}, 1--87 (2015).




\bibitem{Ku} P. Kuchment: {\it An overview of periodic elliptic operators.} Bull. Amer. Math. Soc. {\bf 53}, 343--414 (2016).

\bibitem{Ku2} P. Kuchment: {\it Tight frames of exponentially decaying Wannier functions.} J. of Phys. A: Math. and Theor. {\bf 42}(2), Article number 025203 (2009)


\bibitem{Ma-51} G. W. Mackey: \textit{On Induced Representations of Groups}, American Journal of Mathematics, \textbf{73}, No. 3, 576-592 (1951).

\bibitem{MMP} G. Marcelli, M. Moscolari and G. Panati: {\it Localization of Generalized Wannier Bases Implies Chern Triviality in Non-periodic Insulators.} Ann. H. Poincar{\'e} {\bf 24}(3), 895--930 (2023) 

\bibitem{MP-1} M. M\u{a}ntoiu and  R. Purice: \textit{The magnetic Weyl calculus}. Journal of Mathematical Physics 45 (4), 1394--1417 (2004).

\bibitem{MPR1} M. M\u antoiu, R. Purice, and S. Richard: 
{\it Spectral and
	propagation results for magnetic Schr\"odinger operators; a $C^*$-Algebraic Approach}.
\newblock  J. Funct. Anal. {\bf 250}, 42--67 (2007).

\bibitem{Mo} A. Mohamed: \textit{Asymptotic of the density of states for the Schrödinger operator with periodic electromagnetic potential}. J. Math. Phys. {\bf 38}(8), 4023-4051 (1997)

\bibitem{Mon} D. Monaco: \textit{ Chern and Fu-Kane-Mele invariants as topological obstructions.} Springer INdAM Series {\bf 18}, 201--222 (2017) 

\bibitem{Ne-LMP} G. Nenciu: 
{\it Bloch electrons in a magnetic field, rigorous
	justification of the Peierls-Onsager effective Hamiltonian}.
\newblock Lett. Math.Phys.
{\bf 17}, 247--252 (1989).

\bibitem{Ne-RMP} G. Nenciu: 
{\it Dynamics of Bloch electrons in electric and
	magnetic fields, rigorous
	justification of the effective Hamiltonians}.
\newblock Rev. Mod.Phys.
{\bf 63}(1), 91--127 (1991).

\bibitem{Ne-02} G. Nenciu: {\it On asymptotic perturbation theory for quantum mechanics: almost invariant subspaces and gauge invariant magnetic perturbation theory.} J. Math. Phys. {\bf 43}(3),  1273--1298 (2002).



\bibitem{Pa} R.S. Palais: {\it Homotopy theory of infinite dimensional manifolds}. Topology {\bf 5}, 1--16 (1966).

\bibitem{PP} G. Panati, A. Pisante: \textit{Bloch Bundles, Marzari-Vanderbilt Functional and Maximally Localized Wannier Functions}, Commun. Math. Phys. \textbf{322}, 835--875 (2013).

\bibitem{PST} G. Panati, H. Spohn, and S. Teufel: {\it Effective dynamics for
	Bloch
	electrons: Peierls substitution and beyond}. 
\newblock Comm. Math. Phys. {\bf 242}(3), 547--578 (2003).

\bibitem{Pe} R.E. Peierls: {\it Quantum Theory of Solids}.
\newblock  Oxford University
Press, (1955).



\bibitem{Sj} J. Sj\"ostrand: {\it Microlocal analysis for the periodic magnetic Schr\"odinger equation and related questions}. CIME Lectures 1989. In Microlocal Analysis and Applications. 
237--332 (1989).




\bibitem{Z} J. Zak: {\it Magnetic translation group}. Phys. Rev {\bf 134}, 1602 A, (1964).
	
\end{thebibliography}
\end{document}